\documentclass[12pt, reqno]{amsart}  

\usepackage{main}

\usepackage{setspace}       
\usepackage{lipsum}         
\usepackage{pdflscape}      
\usepackage{longtable}      
\usepackage{booktabs}       
\usepackage{array}          
\usepackage{makecell}       
\usepackage{tabularx}       
\usepackage{multirow}       
\usepackage{adjustbox}      

\graphicspath{{Figures/}}

\usepackage[style=apa, backend=biber, natbib=true, backref=true]{biblatex}
\DefineBibliographyStrings{english}{%
  backrefpage = {},
  backrefpages = {}
}
\renewbibmacro*{pageref}{%
  \iflistundef{pageref}
    {}
    {\printtext[brackets]{\printlist[pageref][-\value{listtotal}]{pageref}}}}

\usepackage{titletoc}

\newgeometry{margin=1.25in}

\title{\Large Beyond the Null Effect: Unmasking the True Impact of Teacher--Child Interaction Quality on Child Outcomes in Early Head Start}

\author{JoonHo Lee and Alison Hooper}

\date{\footnotesize 
January 30, 2026; Lee \& Hooper: College of Education, The University of Alabama; Correspondence: jlee296@ua.edu; Acknowledgments: The authors are grateful for the support of the Administration for Children and Families (ACF) of the United States (U.S.) Department of Health and Human Services (HHS), Grant 90YR0160.
}

\numberwithin{equation}{section}  

\begin{document}


\begin{abstract}
In Early Head Start (EHS), teacher--child interactions are widely believed to shape infant--toddler outcomes, yet large-scale studies often find only modest or null associations. This study addresses four methodological sources of attenuation---item-level measurement error, center-level confounding, teacher- and classroom-level covariate imbalance, and overlooked nonlinearities---to clarify classroom process quality's true influence on child development. Using data from the 2018 wave of the Early Head Start Family and Child Experiences Survey (Baby FACES), we applied a three-level generalized additive latent and mixed model (GALAMM) to distinguish genuine classroom-level variability in process quality, as measured by the Classroom Assessment Scoring System (CLASS) and Quality of Caregiver--Child Interactions for Infants and Toddlers (QCIT), from item-level noise and center-level effects. We then estimated dose--response relationships with children's language and socioemotional outcomes, employing covariate balancing weights and generalized additive models. Results show that nearly half of each item's variance reflects classroom-level processes, with the remainder tied to measurement error or center-wide influences, masking true classroom effects. After correcting for these biases, domain-focused dose-response analyses reveal robust linear associations between cognitive/language supports and children's English communicative skills, while emotional-behavioral supports better predict social-emotional competence. Some domains display plateaus when pushed to extremes, underscoring potential nonlinearities. These findings challenge the ``null effect'' narrative, demonstrating that rigorous methodology can uncover the critical, domain-specific impacts of teacher--child interaction quality, offering clearer guidance for targeted professional development and policy in EHS.
\end{abstract}

\maketitle  

\noindent\textbf{Keywords:} Teacher--Child Interaction Quality; Early Head Start; Measurement Error; Center-Level Confounding; Dose--Response Function; Infants and Toddlers

\pagestyle{plain}  
\newpage



\section{Introduction}
\label{sec:introduction}

High-quality teacher--child interactions are widely viewed as a powerful catalyst for children's early development in cognition, language, and social-emotional skills \citep{Burchinal2018, PiantaHamre2009}. Across various theoretical frameworks, researchers emphasize \emph{proximal processes} \citep{Bronfenbrenner1998}: the repeated, reciprocal exchanges between children and caregivers that drive learning and growth in the earliest years. These interactions are a primary component of classroom process quality and shape policies for Head Start and Early Head Start (EHS), which aim to boost outcomes for low-income children by fostering warm, cognitively rich adult--child interactions \citep{Chaudry2017}.

Despite this strong theoretical and policy rationale, many large-scale studies report small or even null correlations between measured classroom process quality and child outcomes \citep[e.g.,][]{GuerreroRosada2021, Perlman2016, Weiland2013, Xue2022}. These weak associations raise a key puzzle: \emph{If theory so strongly affirms high-quality interactions, why do empirical studies often find only modest links?} We argue that this discrepancy likely stems from methodological artifacts---particularly measurement error in composite scores, unaddressed confounding at higher ecological levels, and the use of purely linear models---rather than a genuine absence of meaningful effects.

To assess these interactions, researchers rely on standardized observational instruments. The most widely used is the Classroom Assessment Scoring System \citep[CLASS;][]{Pianta2008}, utilized in policy initiatives like the Head Start Designation Renewal System (HSDRS). For infant-toddler settings, the CLASS-Toddler \citep{LaParo2012} assesses Emotional and Behavioral Support (EBS) and Engaged Support for Learning (ESL), while the CLASS-Infant \citep{Jamison2014} measures Responsive Caregiving. The Quality of Caregiver-Child Interactions for Infants and Toddlers \citep[QCIT;][]{AtkinsBurnett2015}, designed specifically for birth-to-3 settings, captures support for social-emotional, cognitive, and language/literacy development. Both instruments involve trained observers rating items on multi-point scales, aggregated into domain scores.

Scholars propose that methodological limitations may be obscuring a stronger relationship between process quality and child outcomes \citep{Burchinal2018}. Four key issues have been identified: (a) measurement error from averaging items with varying reliability \citep{Gilbert2024}, (b) center-level confounding in national datasets \citep{RabeHeskethSkrondal2022}, (c) confounding with teacher and structural features \citep{Manning2019}, and (d) overlooked nonlinear patterns including thresholds and plateaus \citep{Hatfield2016}.

These complexities are especially pronounced in EHS settings that serve infants and toddlers. For example, observing interactions with nonverbal or minimally verbal children may pose unique reliability challenges \citep{AtkinsBurnett2015}, and large variations across EHS centers can overshadow classroom-level effects: centers differ in resources, staff stability, and community contexts, making it difficult to isolate the role of day-to-day interactions \citep{Li2013, Vogel2015}.

Furthermore, while both CLASS and QCIT aim to capture high-quality interactions, they may differ in focus, operationalization, or sensitivity to specific behaviors in infant-toddler environments. The CLASS emphasizes broad domains of emotional and instructional support using age-specific versions, whereas the QCIT covers the full birth-to-three spectrum and provides a more granular look at domain-specific supports (e.g., separating cognitive from language/literacy support). Understanding how these widely used tools converge or diverge when observing the same environment is crucial for accurate assessment and interpretation of classroom quality, yet few studies have directly compared them using rigorous psychometric approaches free from measurement error.

To address these methodological challenges, this study analyzes data from the 2018 Baby FACES using a three-pronged approach. First, we implement a three-level generalized additive latent and mixed model \citep[GALAMM;][]{Sorensen2023} to separate genuine classroom-level quality from item-level noise and center-wide influences \citep{RabeHesketh2004}. Second, we apply covariate balancing techniques to ensure fair comparisons across quality levels. Third, we assess dose-response relationships using generalized additive models \citep{Wood2017} to detect potential nonlinear patterns. By systematically addressing measurement error, confounding, and nonlinearity, we aim to clarify whether the modest effect sizes commonly found in EHS stem from genuine weakness in the process quality-outcome link or from methodological limitations.

\section{Literature Review}
\label{sec:literature}

\subsection{Theoretical Foundations Linking Process Quality to Child Development}
\label{subsec:theoretical-foundations}

ECE classroom quality is typically conceptualized as \emph{structural} (regulable aspects like ratios and teacher qualifications) and \emph{process} (dynamic elements like teacher-child interactions and engagement; \citealp{Pianta2016}). Although both matter, process quality is considered more directly related to children's learning and development \citep{Burchinal2016}.

High-quality teacher--child interactions form the backbone of several influential developmental theories, all converging on the idea that infants and toddlers thrive when caregivers provide warm, responsive, and cognitively rich experiences \citep{Burchinal2018, PiantaHamre2009}. \Cref{tab:theoretical-framework} presents an integrative conceptual framework linking these theories to observable classroom practices and their operationalization in CLASS and QCIT.

Attachment theory \citep{Bowlby1969} emphasizes how infants develop security through consistent, sensitive caregiving---captured in Emotional and Behavioral Support (CLASS-Toddler), Responsive Caregiving (CLASS-Infant), and Support for Social-Emotional Development (QCIT). Sociocultural theory \citep{Vygotsky1978} highlights adult scaffolding within the zone of proximal development, central to Engaged Support for Learning (CLASS-Toddler) and Support for Cognitive and Language/Literacy Development (QCIT). Bioecological theory \citep{Bronfenbrenner1998} conceptualizes daily back-and-forth interactions as proximal processes---the primary engines of development that both instruments are designed to capture.

This framework guides our study in two ways. First, it justifies our focus on domain-specific relationships by illustrating how different theoretical traditions emphasize distinct pathways, underpinning our expectation of ``domain-matching'' effects. Second, it informs our comparison of CLASS and QCIT; while both operationalize similar theoretical constructs, CLASS organizes domains more broadly, whereas QCIT separates cognitive from language support, potentially offering different insights into these domain-matching effects.

\subsection{Prior Empirical Findings: Weak or Null Effects}
\label{subsec:prior-findings}

Despite robust theoretical frameworks, empirical studies using standardized observation measures have consistently documented weak or null quality--outcome relationships \citep{GuerreroRosada2021, Harding2025, Howard2024, McDoniel2022, VonSuchodoletz2023, Wang2023, Xue2022}. We next review these findings, first examining the extensive preschool literature and then focusing on infants and toddlers (I/T).

\begin{landscape}
\thispagestyle{empty}
\begin{table}[htbp]
\centering
\caption{Theoretical framework linking theory, practice, measurement, and hypothesized outcomes in infant-toddler classrooms}
\label{tab:theoretical-framework}
\scriptsize
\setlength{\tabcolsep}{5pt}
\renewcommand{\arraystretch}{1.15}
\newcolumntype{S}{>{\hsize=0.6\hsize\raggedright\arraybackslash}X}
\newcolumntype{M}{>{\hsize=1.2\hsize\raggedright\arraybackslash}X}
\begin{tabularx}{\linewidth}{@{} >{\raggedright\arraybackslash}p{2.2cm} S M M >{\raggedright\arraybackslash}p{2.8cm} @{}}
\toprule
\textbf{Theoretical Framework} & \textbf{Core Concepts} & \textbf{Examples of Classroom Practices (Infant/Toddler Context)} & \textbf{Operationalization in Measurement (Most Relevant CLASS \& QCIT Domains)} & \textbf{Primary Child Outcomes Hypothesized} \\
\midrule

\textbf{Attachment Theory} \newline \citep{Ainsworth2015, Bowlby1969} & 
$\bullet$ Secure base \newline
$\bullet$ Sensitive and consistent caregiving \newline
$\bullet$ Emotional security and co-regulation & 
$\bullet$ \textbf{Primary Caregiving \& Continuity of Care:} Maintaining long-term teacher-child relationships \newline
$\bullet$ \textbf{Responsive Routines:} Executing daily routines with sensitivity to cues and distress \newline
$\bullet$ \textbf{Emotional Availability:} Providing comfort, warm affect, and helping children manage emotions & 
\textbf{CLASS-Toddler:} Emotional and Behavioral Support (EBS) \newline
\textbf{CLASS-Infant:} Responsive Caregiving (RC) \newline
\textbf{QCIT:} Support for Social-Emotional Development (SE) \newline
\emph{(Indicators: Sensitivity, Positive Regard, Emotional Responsiveness)} & 
$\bullet$ Increased Social-Emotional Competence (BITSEA) \newline
$\bullet$ Decreased Problem Behaviors (BITSEA) \\
\midrule

\textbf{Sociocultural Theory} \newline \citep{Vygotsky1978} & 
$\bullet$ Scaffolding (ZPD) \newline
$\bullet$ Language as a cognitive tool \newline
$\bullet$ Social instruction & 
$\bullet$ \textbf{Language Modeling:} Frequent verbal interactions, including narrating actions and extending utterances \newline
$\bullet$ \textbf{Scaffolding Strategies:} Using open-ended questions, feedback loops, and gestures to extend learning \newline
$\bullet$ \textbf{Collaborative Learning:} Encouraging peer interaction & 
\textbf{CLASS-Toddler:} Engaged Support for Learning (ESL) \newline
\textbf{QCIT:} Support for Cognitive Development (Cog); Support for Language and Literacy Development (LL) \newline
\emph{(Indicators: Language Modeling, Concept Development, Quality of Feedback)} & 
$\bullet$ Increased Language and Communication Skills (CDI IRT) \\
\midrule

\textbf{Bioecological Theory (Proximal Processes)} \newline \citep{Bronfenbrenner1998} & 
$\bullet$ Proximal processes \newline
$\bullet$ Reciprocal interactions \newline
$\bullet$ Microsystem dynamics & 
$\bullet$ \textbf{Serve-and-Return:} Daily back-and-forth reciprocal exchanges \newline
$\bullet$ \textbf{Joint Attention:} Shared focus during play, routines, or exploration \newline
$\bullet$ \textbf{Adaptive Interactions:} Adjusting interactions to children's evolving interests and developmental levels & 
\textbf{All domains across CLASS and QCIT} \newline
\emph{(Classroom-level proximal processes embedded in 3-level structure: items $\to$ classrooms $\to$ centers)} & 
$\bullet$ Cumulative effects across domains \newline
$\bullet$ Moderation of quality-outcome associations \\
\bottomrule
\end{tabularx}

\vspace{0.3em}
\begin{minipage}{\linewidth}
\scriptsize
\emph{Note.} CLASS = Classroom Assessment Scoring System; QCIT = Quality of Caregiver-Child Interactions for Infants and Toddlers; ZPD = Zone of Proximal Development; BITSEA = Brief Infant Toddler Social Emotional Assessment; CDI IRT = MacArthur-Bates Communicative Development Inventories (Item Response Theory scaled).
\end{minipage}
\end{table}
\end{landscape}

\subsubsection{Findings for Preschool-Aged Children}
\label{subsubsec:findings-preschool}

We begin with preschool classrooms given the substantially larger evidence base \citep{Banghart2020, Burchinal2016, Mashburn2008}. However, I/T environments differ structurally---smaller groups, lower ratios, and predominantly routine-embedded, dyadic, often nonverbal interactions \citep{Barros2016, Cadima2022}---compared to preschool's more structured, group-based instruction \citep{Phillipsen1997}. These differences affect how process quality is operationalized and measured \citep{LaParo2014}.

A persistent pattern is that higher process quality---often assessed via CLASS Pre-K---shows only modest or null associations with developmental gains \citep{Burchinal2018}. Meta-analyses consistently report small magnitudes, with correlations generally below $r = .15$ and many failing to reach significance \citep{Egert2018, Perlman2016, Ulferts2019}. Large-scale investigations document similarly weak associations \citep{GuerreroRosada2021, Mashburn2008}. Although nonlinear models sometimes indicate robust gains in extremely high-quality classrooms \citep{Burchinal2010, Harding2025, Hatfield2016}, the overall portrait is that observed interactions show limited predictiveness, much smaller than theory suggests.

\subsubsection{Findings for Infant and Toddler Populations}
\label{subsubsec:findings-infants}

Evidence linking classroom process quality to outcomes among infants and toddlers is more limited and often similarly weak. Compared to the research on preschool-aged children, substantially fewer large-scale studies focus specifically on infants in center-based care \citep{Banghart2020}. The availability of robust measurement tools for assessing teacher-child interactions in this age group remains scarce, and existing instruments may be limited in their focus, applicability, reliability, or validity \citep{Nguyen2023}.

Studies employing observational measures commonly indicate limited, inconsistent, and often weak associations between observed process quality and child developmental outcomes \citep{Bandel2014, Xue2022}. This pattern is particularly evident within Early Head Start (EHS) contexts. Analysis of the Baby FACES 2018 study revealed minimal associations between observed quality measures and child outcomes in EHS classrooms \citep{Xue2022}. In infant classrooms, neither threshold analyses nor linear associations revealed significant links between quality measures and child outcomes. In toddler classrooms, only two associations emerged: higher scores in CLASS-Toddler's Emotional and Behavioral Support domain or QCIT's Support for Language and Literacy Development domain were associated with lower behavior problems. These findings align with previous research showing that CLASS-Toddler associations varied by child age, with minimal significant associations for most domains and, in some cases, paradoxical negative relationships \citep{Bandel2014}.

Although targeted interventions in infant-toddler classrooms (e.g., intensive coaching on responsive caregiving) have occasionally produced small-to-moderate gains in socioemotional or language development \citep{Chen2017, Moreno2015}, such results are not consistently replicated, and broader program evaluations often fail to detect more than trivial effects \citep{McKelvey2014}. Scholars attribute these null or weak findings to both practical and methodological factors, including the difficulty of capturing nuanced infant-caregiver interactions in standard observation protocols and restricted quality ranges in many samples \citep{Bandel2014, Banghart2020}. Nonetheless, the overarching conclusion resonates with the preschool literature: while theory posits a crucial role for day-to-day responsive interactions, empirical linkages between classroom quality ratings and developmental outcomes remain modest at best \citep{Aikens2015, Banghart2020, Xue2022}.

\subsection{Explaining the Discrepancy: Four Methodological Limitations}
\label{subsec:methodological-limitations}

Researchers have proposed several explanations for these weak or null associations: range restriction limiting variability at quality extremes \citep{Burchinal2018, Chaudry2017, Gordon2020, Nguyen2023, Zaslow2016}, single brief observations failing to capture day-to-day practices \citep{BratschHines2020, Hatfield2016, WeilandRosada2022}, and global measures obscuring domain-specific interactions more directly linked to outcomes \citep{Brunsek2017, Justice2018}. While these factors undoubtedly contribute, we focus on four methodological artifacts addressable through advanced modeling: item-level measurement error, center-level confounding, covariate imbalance, and overlooked nonlinear patterns.

\subsubsection{Measurement Error at the Item Level}
\label{subsubsec:measurement-error}

Classroom observation tools, such as the widely used CLASS and the more recently developed QCIT, typically involve summing or averaging multiple item ratings to produce domain scores, which are then regressed on child outcomes \citep[e.g.,][]{Aikens2015, Downer2012, Mashburn2008, McDoniel2022, Stephens2023, Xue2022}. While these domain composites are convenient and intuitively straightforward, they can embed substantial measurement error from item-level variation. In particular, each item may differ in how reliably or strongly it loads onto the latent process quality factor, yet standard practice weights all items equally. This mismatch inflates the error variance of the summed (or averaged) domain score, and when that noisy domain score is used as a predictor in regression, the resulting coefficients can be systematically attenuated, biased toward zero \citep{Pohl2016, Sengewald2019}. This means that the observed variation in the predictor partly reflects random item noise rather than true latent quality \citep{Frost2000, Gilbert2024, Hutcheon2010}. Even when an observation tool has acceptable overall reliability, differential item performance can reduce effective reliability, leading to 10--30\% smaller effect sizes or even null results in regression analyses \citep{Mashburn2008, Perlman2016}.

Simple composites remain valuable for professional development or program monitoring given their transparency and ease of interpretation \citep{Andersen2022}. However, when estimating the true magnitude of quality-outcome associations, latent variable models that account for differential item reliability can reveal associations 15--30\% stronger than simple averages \citep{Bihler2018, Levickis2024}. Our analysis demonstrates how this approach yields substantively different conclusions about teacher-child interactions in infant-toddler classrooms.

\subsubsection{Center- and State-Level Confounding}
\label{subsubsec:center-confounding}

Another explanation for weak findings lies in center-level confounding, particularly salient in nationally representative datasets. ECE centers vary dramatically in structural features, organizational resources, and community contexts \citep{Karoly2013}, differences that correlate with both process quality and child outcomes, thereby confounding the classroom-level relationship of interest.

Centers in higher-resourced communities often benefit from multiple advantages simultaneously---more qualified teachers, better facilities, and families with fewer economic stressors \citep{BurchinalNelson2000, Coley2014, Harding2025}. These center-level characteristics create clustering where both quality and outcomes tend to be higher due to shared center-wide advantages rather than causal classroom-level connections. Conversely, under-resourced centers face compounding challenges that suppress both quality ratings and outcomes \citep{Cloney2016, Levickis2024}.

The statistical solution involves within-center comparisons through center-level fixed effects or correlated random-effects models using the Mundlak specification \citep{Mundlak1978, RabeHeskethSkrondal2022}. Without this adjustment, a regression might effectively compare a high-quality classroom in suburban California with one in rural Alabama, attributing to classroom interactions what may actually reflect broader ecological differences.

This confounding becomes particularly problematic in national datasets lacking state identifiers. In Baby FACES, center-level confounding cannot be separated from state-level heterogeneity in licensing requirements, quality standards, and workforce policies \citep{Xue2022}. Without disentangling these nested sources of variation, studies risk misattributing ecological factors to classroom interactions.

\subsubsection{Confounding by Teacher Qualifications and Classroom Structural Features}
\label{subsubsec:teacher-confounding}

A third challenge arises when teacher characteristics and structural features confound the process quality--outcome linkage. Teachers with higher education or specialized training may display more sensitive, stimulating interactions, which---rather than interactions alone---could drive developmental gains \citep{Manning2019, Schaack2017}. Similarly, lower child-teacher ratios can support more individualized attention, elevating both quality metrics and outcomes \citep{Phillipsen1997, Xue2022}. Failing to adjust for these factors might produce spurious correlations, yet over-controlling can understate interaction effects if structural factors and process quality are deeply intertwined \citep{SolidayHong2019}. Researchers advocate applying robust causal inference techniques---inverse probability weighting or propensity score matching---to isolate interaction quality's unique impact \citep{Ruzek2014, Weiland2018}, though large-scale datasets often lack the granular covariate information needed for such adjustments.

\subsubsection{Missing Nonlinear, Threshold, or Plateau Effects}
\label{subsubsec:nonlinearity}

Finally, if the quality--outcome relationship follows a nonlinear pattern, conventional linear models will obscure meaningful effects \citep{Burchinal2010, Harding2025, Hatfield2016}. Several considerations suggest such nonlinearity in infant-toddler settings. Developmentally, infants may require a minimum threshold of emotional security before cognitive stimulation becomes meaningful \citep{Bernier2012, Landry2006}, while at very high quality levels, ceiling effects may limit detectable gains \citep{Watts2021}.

The nature of nonlinearity may also differ across domains. Social-emotional competence might respond linearly to emotional support, while language development could show threshold effects---minimal gains below a certain level followed by significant improvements once crossed \citep{Burchinal2016}. Behavioral outcomes may require surpassing high quality thresholds; \citet{Watts2021} found that for children in poverty, only those in classrooms with the highest organization levels showed substantial behavioral improvements---a pattern also documented by \citet{Burchinal2010} and \citet{Harding2025}. Traditional linear regressions cannot capture these varied functional forms, potentially missing strong effects in certain quality ranges. Although piecewise or spline models can approximate these patterns \citep{GuerreroRosada2021, Hatfield2016}, such techniques remain underutilized.

\subsubsection{Compounding Effects and Integrated Solutions}
\label{subsubsec:compounding}

Critically, these four limitations rarely operate in isolation. Measurement error reduces power to detect nonlinear patterns; center-level confounding can mask domain-specific relationships; teacher qualifications may moderate quality-outcome associations. When these limitations compound---as they typically do in large-scale observational studies---resulting attenuation can reduce observed associations to near zero.

Our integrated approach addresses these challenges simultaneously. Multilevel latent variable modeling extracts true quality signals from noisy item-level data while accounting for center-level confounding \citep{SkrondalRabeHesketh2004}. Covariate balancing creates fair comparisons across quality levels. Testing both linear and nonlinear specifications allows the data to reveal the true functional form. This comprehensive strategy offers the best opportunity to unmask genuine impacts hidden by methodological artifacts in previous research.

\subsection{Research Questions}
\label{subsec:research-questions}

Building on insights about how methodological limitations may conceal the true influence of classroom quality, we apply a three-level GALAMM measurement model and weighted dose--response analyses to address these challenges, yielding two primary research questions:

\begin{enumerate}
    \item How can advanced measurement modeling improve the precision of EHS classroom process quality measurement?
    \item What is the nature of the dose-response relationship between classroom process quality and child developmental outcomes when accounting for potential confounding and nonlinearity?
\end{enumerate}

\paragraph{RQ1: Improving the Measurement of EHS Classroom Process Quality.}

\begin{itemize}
    \item RQ1-A. How much true classroom-level heterogeneity in process quality exists once item-level error and center-level variation are partitioned out?
    \item RQ1-B. Which individual items show the strongest versus weakest loadings on their respective latent constructs?
    \item RQ1-C. How do the three QCIT-based factors correlate with the three CLASS-based factors when measuring the same classrooms?
\end{itemize}

\paragraph{RQ2: Examining Dose--Response Relationships with Child Outcomes.}

\begin{itemize}
    \item RQ2-A. Which teacher- or classroom-level characteristics confound the quality--outcome relationship, and how can covariate balancing isolate process quality's unique effect?
    \item RQ2-B. Do we observe linear or nonlinear dose--response patterns, including threshold or plateau effects?
\end{itemize}

\section{Data and Methods}
\label{sec:methods}

\subsection{Sample}
\label{subsec:sample}

We draw on Baby FACES 2018, a nationally representative dataset of Early Head Start programs \citep{Vogel2015}. This wave focused extensively on classroom context, offering robust observational measures of teacher--child interactions. We excluded home-based EHS components, which utilize distinct service models and observation protocols \citep{Raikes2002}, as well as classrooms lacking sufficient item-level data on QCIT or CLASS. Our final sample for RQ1 consists of 855 classrooms nested within 468 EHS centers. For RQ2, we analyzed 2,301 children matched to classrooms with available observational data; robustness checks using parent-reported outcomes yielded samples of 1,874--1,986 children (see \cref{app:sample} in the Online Supplemental Materials [OSM] for detailed sample derivation).

\subsection{Measures}
\label{subsec:measures}

\subsubsection{Classroom Process Quality}
\label{subsubsec:process-quality}

Baby FACES 2018 measured teacher--child interaction quality using two observation tools: the Quality of Caregiver--Child Interactions for Infants and Toddlers \citep[QCIT;][]{AtkinsBurnett2015} and the Classroom Assessment Scoring System (CLASS), comprising CLASS--Infant \citep{Jamison2014} or CLASS--Toddler \citep{LaParo2012}. During the same observation window, one trained observer administered the age-appropriate CLASS version while another conducted the QCIT, enabling direct examination of instrument convergence (RQ1-C) while controlling for contextual factors.

QCIT captures three domains: (1) Support for social-emotional development, (2) Support for cognitive development, and (3) Support for language and literacy development. \citet{Nguyen2023} demonstrated strong psychometric properties and measurement invariance across infant, toddler, and mixed-age settings. CLASS--Toddler assesses (4) Emotional and behavioral support and (5) Engaged support for learning; CLASS--Infant offers a single domain, (6) Responsive caregiving. These six latent constructs are each assessed via multiple observed items (\Cref{fig:variance-decomposition}). Of the 855 classrooms, 707 serve primarily toddlers (CLASS--Toddler) and 148 serve primarily infants (CLASS--Infant). We address these age-specific missing item blocks using a long-format generalized latent variable modeling approach assuming missing at random (see \cref{app:galamm}).

\subsubsection{Child Developmental Outcomes}
\label{subsubsec:outcomes}

We computed three teacher-reported outcome measures aggregated at the classroom level. Language and communication skills were assessed using IRT-scaled $T$-scores from the MacArthur--Bates Communicative Development Inventories \citep[CDI;][]{Fenson2000}. Social-emotional outcomes were measured via the Brief Infant Toddler Social Emotional Assessment \citep[BITSEA;][]{BriggsGowan2006}: the Competence subscale captures socioemotional strengths (higher scores indicate greater competence), while the Problem subscale assesses behavioral difficulties (higher scores indicate more challenges). All outcomes were group-mean-centered at the center level to isolate within-center variation (see \cref{app:dose-response} in the OSM).

To examine robustness to potential teacher-reporting bias, we analyzed parallel parent-reported versions of all three measures \citep{Magro2024, Perlman2016}, applying the same centering procedure (see \cref{app:robustness} in the OSM).

\subsubsection{Teacher- or Classroom-Level Covariates}
\label{subsubsec:covariates}

Finally, we compiled a range of teacher- and classroom-level characteristics to account for potential confounding in the link between classroom process quality and child outcomes. These include teacher demographics (e.g., race/ethnicity), qualifications (e.g., a bachelor's degree in ECE), and classroom structural features (e.g., child--adult ratio). We also considered teachers' psychological factors, such as depressive symptoms \citep[CESD-R;][]{Eaton2004}, caregiving beliefs, and job satisfaction. In total, we constructed 26 classroom-level covariates and used them for balancing (see \cref{app:dose-response} in the OSM).

\subsection{Analytic Strategy}
\label{subsec:analytic-strategy}

We fit a three-level GALAMM measurement model separating true classroom-level process quality (Level 2) from item-level measurement error (Level 1) and center-level characteristics (Level 3), yielding empirical Bayes (EB) predictions for six latent factors. We then apply covariate balancing to ensure EB-based quality measures are uncorrelated with teacher- and classroom-level covariates, and estimate dose--response relationships using both linear and generalized additive models (GAMs). Appendices B--E provide technical details.

\subsubsection{RQ1: Three-Level GALAMM Measurement Model}
\label{subsubsec:galamm}

Let $y_{i,j,k}$ denote the observed QCIT or CLASS item $i$ for classroom $j$ in center $k$. We assume:
\begin{equation}
y_{i,j,k} \mid \eta_{1,j}, \ldots, \eta_{6,j}, \alpha_k \sim \mathcal{N}\left(\beta_i + \sum_{f=1}^{6} \lambda_{i,f} \eta_{f,j} + \alpha_k, \sigma_\varepsilon^2\right),
\label{eq:galamm}
\end{equation}
where $\eta_{f,j}$ are six classroom-level latent factors, $\lambda_{i,f}$ indicates the factor loading for item $i$ on factor $f$. The random intercept $\alpha_k$ captures a center-wide ``shift'' in observed responses, $\beta_i$ is an item-specific intercept, and $\sigma_\varepsilon^2$ is the residual variance at the item level.

This specification addresses all three RQ1 sub-questions. RQ1-A: Variance components decompose total variation into Level 1 (item error), Level 2 (classroom factors), and Level 3 (center effects), revealing genuine between-classroom heterogeneity after removing measurement noise and center-level confounding (see \cref{app:icc-decomposition}). RQ1-B: Factor loadings $\lambda_{i,f}$ indicate which items are reliable indicators of each domain. RQ1-C: The six-factor covariance matrix $\boldsymbol{\Psi}^{(2)}$ reveals whether QCIT and CLASS measure overlapping or distinct dimensions.

\subsubsection{RQ2: Estimating the Weighted Dose--Response Curves}
\label{subsubsec:dose-response}

Let $Z_j$ denote the \emph{center-mean-centered} outcome for classroom $j$:
\begin{equation}
Z_j = \gamma_0 + \gamma_1 \widehat{\eta}_j^{EB} + e_j,
\label{eq:dose-response}
\end{equation}
where $\widehat{\eta}_j^{EB}$ is the EB posterior mean of the latent factors from the GALAMM. Using EB shrinkage estimates reduces attenuation bias from measurement error \citep[see \cref{app:eb-predictions} in the OSM]{Walters2024}.

For RQ2-A, we apply entropy balancing for continuous treatments \citep[EBCT;][]{Tubbicke2022} to produce weights ensuring the dose variable is uncorrelated with 26 classroom-level covariates (see \cref{app:dose-response} in the OSM). For RQ2-B, we fit both weighted linear models and GAMs with smooth functions of $\widehat{\eta}_j^{EB}$ \citep{Wood2017}. Effective degrees of freedom (edf) near 1 indicate linearity; edf substantially greater than 1 indicates nonlinearity.

\section{Results}
\label{sec:results}

\subsection{RQ1: Measurement of Classroom Process Quality}
\label{subsec:rq1-results}

\subsubsection{ICC-Style Item-Level Variance Decomposition (RQ1-A)}
\label{subsubsec:variance-decomposition}

\Cref{fig:variance-decomposition} displays results from our ICC-style variance decomposition, which partitions each item's total variance into three hierarchical levels---Level 1 (measurement error/residual), Level 2 (classroom latent factors), and Level 3 (center random effects). Overall, approximately 51.2\% of item-level variance can be attributed to classroom-level latent factors (Level 2). In contrast, 34.3\% reflects Level-1 measurement error, and 14.5\% stems from center-level unobserved heterogeneity. Thus, only about half of the total variance captures \emph{true} differences in teacher--child interaction quality across classrooms; the other half is linked to measurement noise or center-level effects. Relying on simple sum or average scores without separating these components risks attenuation bias from item-level error and bias from center-level confounding.

This decomposition also highlights items with especially large measurement error, making them less reliable indicators of classroom quality. For instance, several QCIT items (e.g., ``Supporting peer interaction'') show higher proportions of Level-1 error, suggesting limited utility for building composite classroom-quality indices.

\begin{figure}[htbp]
\centering
\includegraphics[width=0.85\textwidth]{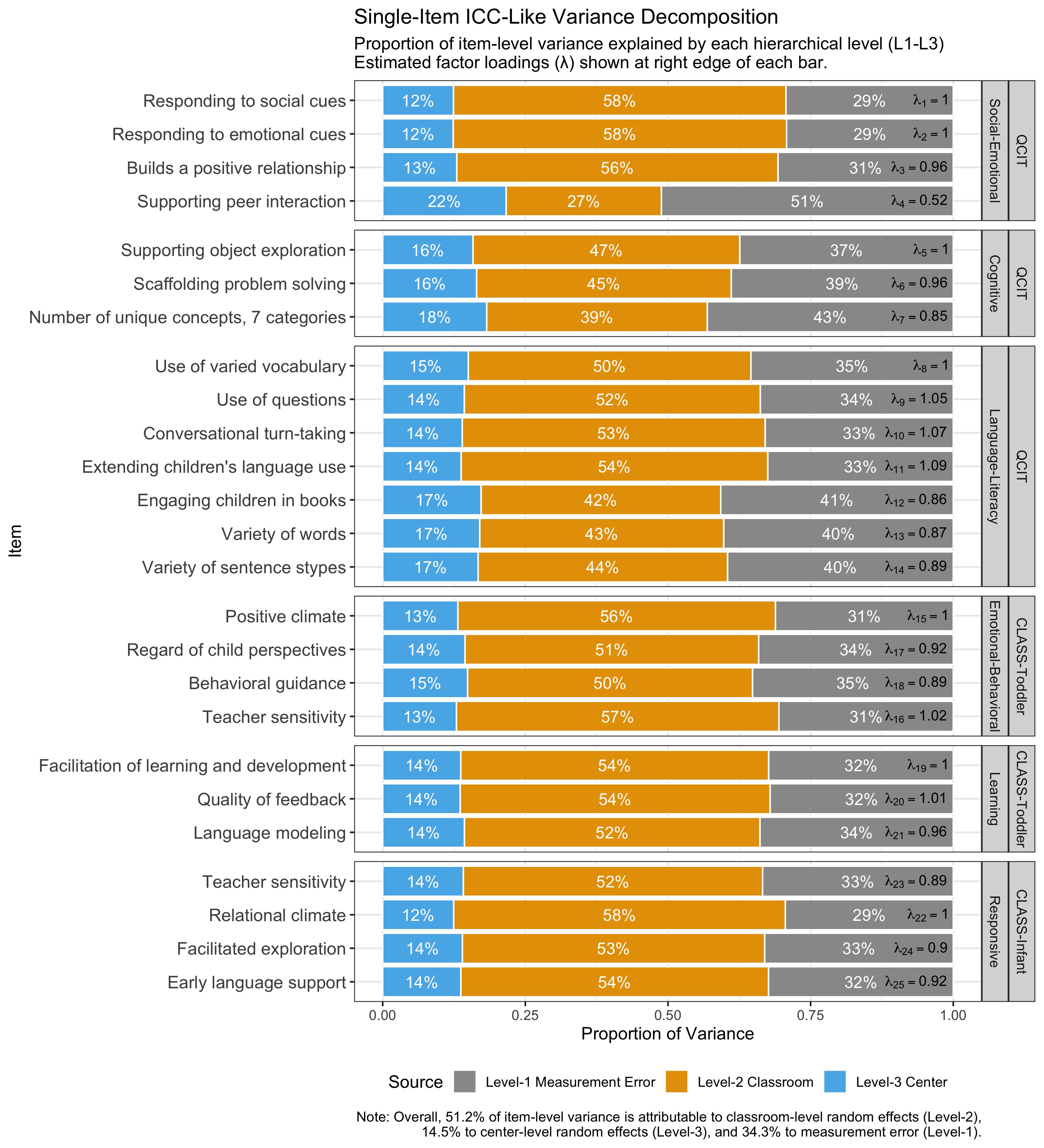}
\caption{Item-level ICC-style variance decomposition and factor loadings}
\label{fig:variance-decomposition}
\begin{minipage}{0.95\textwidth}
\small
\emph{Note.} Each bar depicts how an individual item's total variance is distributed among Level 1 (measurement error/residual), Level 2 (classroom latent factors), and Level 3 (center random effects). See \cref{app:icc-decomposition} in the OSM for technical details.
\end{minipage}
\end{figure}

\subsubsection{Factor Loadings and Implications for Reliability (RQ1-B)}
\label{subsubsec:factor-loadings}

\Cref{fig:variance-decomposition} also displays factor loadings ($\lambda_i$) indicating how strongly each item contributes to its respective latent factor. Items with smaller loadings exhibit proportionally larger Level-1 error and center-level variance, weakening their ability to discriminate classroom-level differences. Several QCIT items show relatively small loadings: ``Supporting peer interaction'' (QCIT Social-Emotional, $\lambda = 0.52$, with 51\% of variance at Level 1 and 22\% at Level 3), ``Number of unique concepts'' (QCIT Cognitive), and three Language/Literacy items (``Engaging children in books,'' ``Variety of words,'' ``Variety of sentence types''). These Cognitive and Language/Literacy items display modest loadings ($\lambda \approx 0.85$--$0.89$). By contrast, CLASS-based factors have more uniformly moderate-to-high loadings, with fewer items prone to attenuate the signal. These findings suggest QCIT items vary more widely in reliability. Future measurement enhancements might focus on revising low-loading items or providing additional rater training to reduce measurement noise.

\subsubsection{Correlations Among Classroom-Level Latent Factors (RQ1-C)}
\label{subsubsec:factor-correlations}

\Cref{tab:factor-correlations} presents estimated classroom-level correlations among the latent factors. After partitioning out Level-1 measurement error and Level-3 center effects, these represent ``true'' classroom-level relationships. Within QCIT, Language-Literacy shows strong associations with both Social-Emotional ($r = 0.76$) and Cognitive ($r = 0.74$), while Social-Emotional and Cognitive correlate more moderately ($r = 0.56$). Within CLASS-Toddler, Emotional-Behavioral and Learning correlate at $r = 0.65$. However, cross-instrument correlations between thematically similar domains are notably lower: CLASS-T Emotional-Behavioral and QCIT Social-Emotional share only $r = 0.38$, while CLASS-T Learning correlates just 0.04 with QCIT Cognitive and 0.22 with QCIT Language-Literacy.

For CLASS-Infant, the single Responsive Caregiving factor correlates most strongly with QCIT Social-Emotional ($r = 0.45$) and QCIT Language-Literacy ($r = 0.33$), but weakly with QCIT Cognitive ($r = 0.15$). Because infant and toddler classrooms do not overlap, cross-instrument correlations between CLASS-Infant and CLASS-Toddler require careful interpretation (\cref{app:galamm} in the OSM).

These patterns indicate that within-instrument correlations exceed those between thematically similar domains across instruments, suggesting each tool captures somewhat different facets of teacher--child interaction.

\begin{table}[htbp]
\centering
\caption{Estimated correlations among classroom-level latent factors (Level-2)}
\label{tab:factor-correlations}
\small
\begin{tabular}{@{} l c c c c c c @{}}
\toprule
Classroom Latent Factor & 1 & 2 & 3 & 4 & 5 & 6 \\
\midrule
1. QCIT Social-Emotional & --- & & & & & \\
2. QCIT Cognitive & .56 & --- & & & & \\
3. QCIT Language-Literacy & .76 & .74 & --- & & & \\
4. CLASS-T Emotional-Behavioral & .38 & .08 & .23 & --- & & \\
5. CLASS-T Learning & .29 & .04 & .22 & .65 & --- & \\
6. CLASS-I Responsive & .45 & .15 & .33 & .19 & .16 & --- \\
\bottomrule
\end{tabular}

\vspace{0.5em}
\begin{minipage}{0.95\linewidth}
\footnotesize
\emph{Note.} $N = 855$ classrooms in 468 EHS centers. Correlations represent classroom-level associations after partitioning out item-level error (Level-1) and center effects (Level-3). QCIT was measured in all classrooms; CLASS-T in 707 toddler classrooms; CLASS-I in 148 infant classrooms. Due to the age-specific design of CLASS instruments, correlations between CLASS-T and CLASS-I factors are based on cross-level covariance estimates rather than within-classroom observations.
\end{minipage}
\end{table}

\subsection{RQ2: Examining Dose--Response Relationships with Child Outcomes}
\label{subsec:rq2-results}

\subsubsection{Covariate Balancing for the ``Dose'' of Classroom Process Quality (RQ2-A)}
\label{subsubsec:covariate-balance}

\Cref{fig:covariate-balance} displays correlations between our designated ``dose''---the QCIT Cognitive latent factor---and teacher- and classroom-level covariates, before and after weighting. Panel (A) employs generalized propensity scores via generalized boosted models \citep[GBM;][]{McCaffrey2004}; Panel (B) uses entropy balancing, which achieves exact moment balance for continuous treatments \citep{Tubbicke2022}.

Most covariates have small initial correlations (under $|0.10|$) with the QCIT Cognitive factor, indicating modest potential for Level-2 confounding once item-level and center-level influences are accounted for. Nevertheless, a few characteristics, including teacher race and job satisfaction, display small but non-trivial associations. While GPS-based weighting reduces some correlations, entropy balancing more effectively pushes correlations toward zero, creating a pseudo-sample in which measured covariates no longer confound the latent ``dose.'' We therefore adopted entropy balancing for all subsequent dose--response analyses (see \cref{app:dose-response} in the OSM for methodological details).

\begin{figure}[htbp]
\centering
\includegraphics[width=\textwidth]{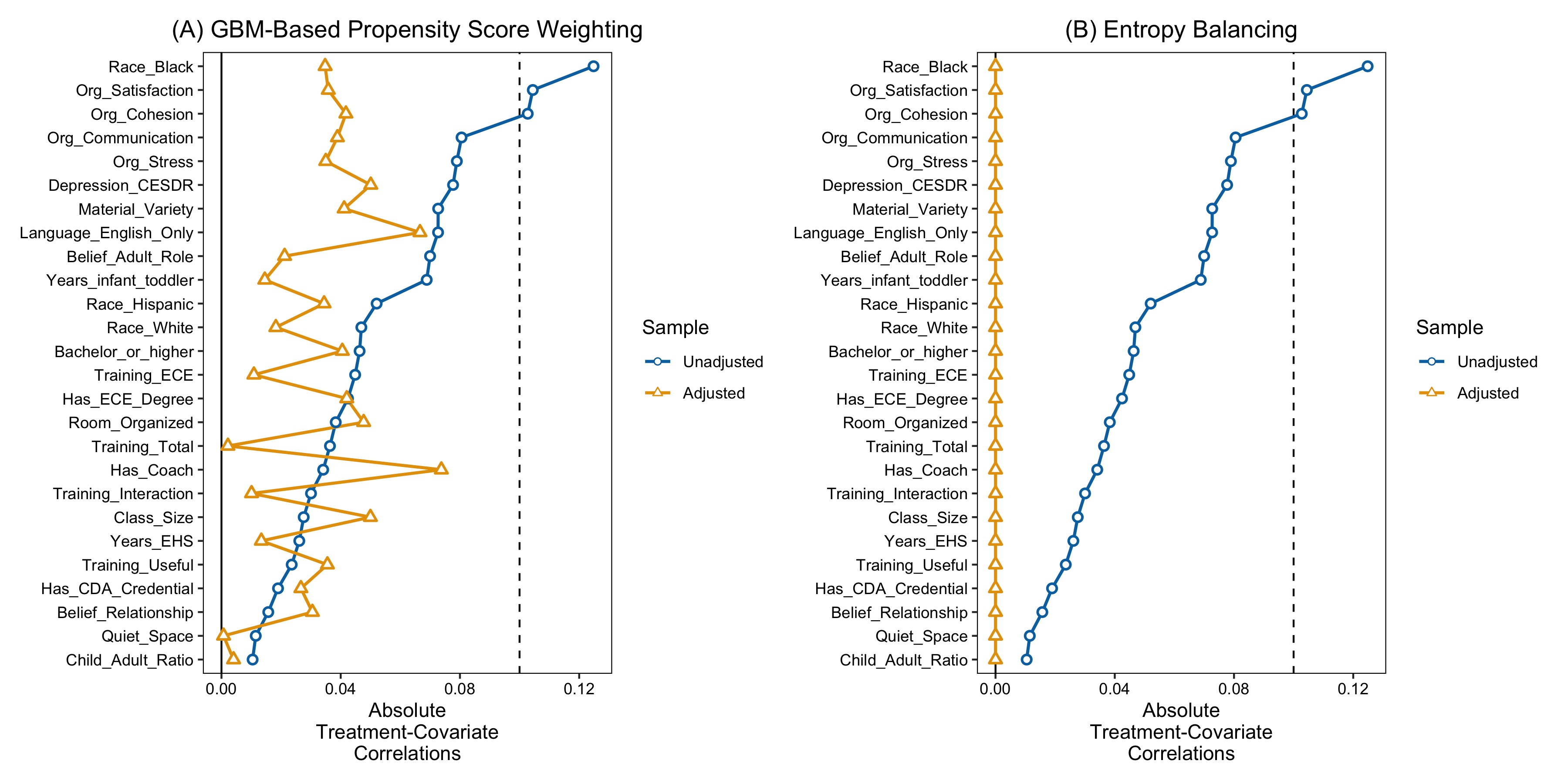}
\caption{Estimated absolute correlations between QCIT Cognitive ``dose'' and covariates before and after weighting}
\label{fig:covariate-balance}
\begin{minipage}{0.95\textwidth}
\small
\emph{Note.} Panel (A): generalized propensity score weighting via boosted models \citep{McCaffrey2004}. Panel (B): entropy balancing \citep{Tubbicke2022}. Blue = unweighted; yellow = weighted.
\end{minipage}
\end{figure}

\subsubsection{Linear vs.\ Nonlinear Dose--Response Patterns (RQ2-B)}
\label{subsubsec:dose-response-patterns}

\Cref{tab:dose-response} presents weighted linear regression and GAM estimates examining how classroom process quality relates to three child outcomes: CDI IRT scores (English communicative skills), BITSEA Competence (social-emotional functioning), and BITSEA Problem (behavioral difficulties). \Cref{fig:dose-response-cdi,fig:dose-response-bitsea} depict dose--response curves for CDI IRT and BITSEA Competence, where the latent ``dose'' ($x$-axis) is plotted against child outcomes ($y$-axis) in the covariate-balanced sample. Blue solid lines represent weighted linear fits; red dashed lines represent GAM fits; black dots mark decile means.

\paragraph{Effects on English Communicative Skills.} For CDI IRT, both QCIT Language-Literacy and QCIT Cognitive significantly predict higher language scores ($p < 0.05$), as does CLASS-T Learning. By contrast, QCIT Social-Emotional and CLASS-T Emotional-Behavioral show no significant associations, suggesting a domain-matching principle whereby language or cognitive scaffolding exerts stronger effects on communicative growth than socioemotional support alone. The QCIT Language-Literacy factor yields the highest coefficient (1.28, $p < 0.01$), implying robust gains in English vocabulary for classrooms supporting language-based interactions.

\paragraph{Effects on Social-Emotional Abilities.} For BITSEA Competence, the most notable effect emerges from CLASS-T Emotional-Behavioral (coef = 0.21, $p = 0.04$), indicating that classrooms with warm interactions and proactive behavior guidance promote stronger social-emotional skills. Domain-specific language or cognitive supports show no significant linear associations, though some nonlinear effects emerge: CLASS-T Learning becomes significant under the GAM (edf = 4.51, $p = 0.02$). As \Cref{fig:dose-response-bitsea} illustrates, BITSEA Competence peaks around moderate levels of learning support, with a plateau at higher extremes.

\begin{table}[htbp]
\centering
\caption{Weighted linear and GAM dose-response estimates for teacher-reported child outcomes across QCIT and CLASS domains}
\label{tab:dose-response}
\small
\setlength{\tabcolsep}{4pt}
\begin{tabular}{@{} l l r@{\hspace{3pt}}c@{\hspace{3pt}}l c c c @{}}
\toprule
& & \multicolumn{3}{c}{Linear Model} & \multicolumn{2}{c}{GAM} \\
\cmidrule(lr){3-5} \cmidrule(lr){6-7}
Response & Dose & Est. & & $p$ & edf & $p$ \\
\midrule
\multirow{6}{*}{\shortstack[l]{English CDI IRT\\(Teacher-reported)}}
 & QCIT Social-Emotional & 0.12 & (0.29) & 0.68 & 1.00 & 0.69 \\
 & \textbf{QCIT Cognitive} & \textbf{0.90} & \textbf{(0.39)} & \textbf{0.02} & \textbf{1.00} & \textbf{0.02} \\
 & \textbf{QCIT Language-Literacy} & \textbf{1.28} & \textbf{(0.34)} & \textbf{$<$0.01} & \textbf{1.14} & \textbf{$<$0.01} \\
 & CLASS-T Emotional-Behavioral & 0.21 & (0.30) & 0.49 & 1.44 & 0.58 \\
 & \textbf{CLASS-T Learning} & \textbf{0.72} & \textbf{(0.32)} & \textbf{0.02} & \textbf{1.00} & \textbf{0.03} \\
 & CLASS-I Responsive & 0.69 & (0.68) & 0.31 & 1.23 & 0.42 \\
\midrule
\multirow{6}{*}{\shortstack[l]{BITSEA Competence\\(Teacher-reported)}}
 & QCIT Social-Emotional & 0.09 & (0.09) & 0.31 & 1.54 & 0.40 \\
 & QCIT Cognitive & 0.08 & (0.13) & 0.54 & 2.16 & 0.28 \\
 & QCIT Language-Literacy & 0.21 & (0.11) & 0.06 & 1.93 & 0.10 \\
 & \textbf{CLASS-T Emotional-Behavioral} & \textbf{0.21} & \textbf{(0.10)} & \textbf{0.04} & \textbf{1.00} & \textbf{0.04} \\
 & CLASS-T Learning & 0.20 & (0.11) & 0.07 & 4.51 & 0.02 \\
 & CLASS-I Responsive & 0.16 & (0.25) & 0.52 & 1.00 & 0.52 \\
\midrule
\multirow{6}{*}{\shortstack[l]{BITSEA Problem\\(Teacher-reported)}}
 & QCIT Social-Emotional & $-$0.04 & (0.15) & 0.77 & 1.00 & 0.77 \\
 & QCIT Cognitive & 0.05 & (0.21) & 0.80 & 1.33 & 0.82 \\
 & QCIT Language-Literacy & $-$0.22 & (0.18) & 0.22 & 1.71 & 0.33 \\
 & CLASS-T Emotional-Behavioral & $-$0.24 & (0.17) & 0.17 & 1.00 & 0.18 \\
 & CLASS-T Learning & 0.00 & (0.19) & 0.99 & 1.00 & 1.00 \\
 & CLASS-I Responsive & $-$0.16 & (0.40) & 0.70 & 1.00 & 0.70 \\
\bottomrule
\end{tabular}

\vspace{0.5em}
\begin{minipage}{0.95\linewidth}
\footnotesize
\emph{Note.} Estimates from entropy-balanced weighted regressions adjusting for 26 covariates. ``Dose'' = empirical Bayes estimates of latent process quality from the three-level GALAMM. Outcomes are center-mean-centered. edf = effective degrees of freedom (edf $\approx 1$ indicates linearity; edf $> 1$ indicates nonlinearity). Significant effects ($p < 0.05$) in bold.
\end{minipage}
\end{table}

\begin{figure}[htbp]
\centering
\includegraphics[width=0.75\textwidth]{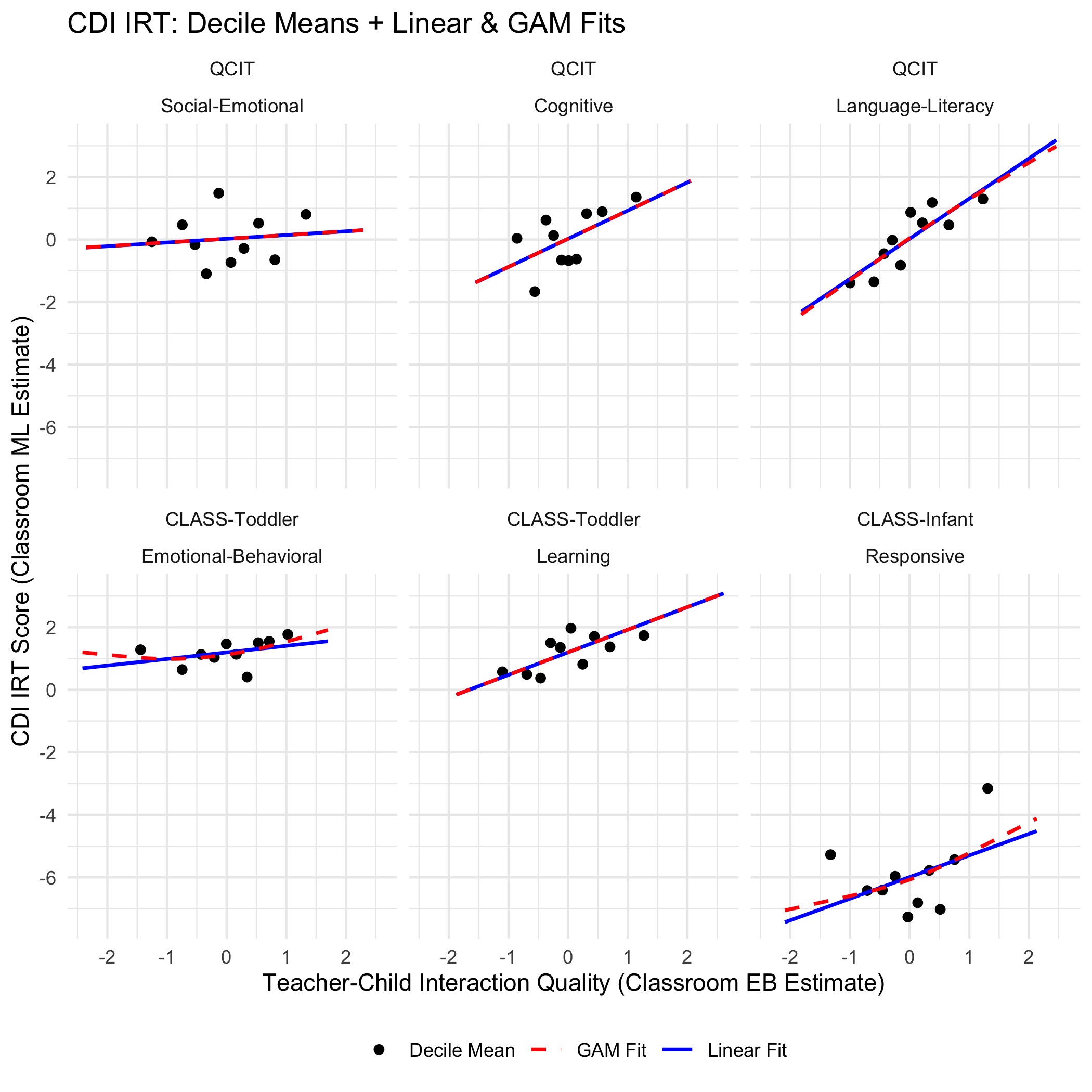}
\caption{Dose--response curves for teacher-reported CDI IRT scores across QCIT and CLASS domains}
\label{fig:dose-response-cdi}
\begin{minipage}{0.95\textwidth}
\small
\emph{Note.} $X$-axis: empirical Bayes estimates of process quality from the GALAMM. $Y$-axis: center-mean-centered outcomes. Black dots = decile means; blue lines = weighted linear fits; red dashed lines = weighted GAM fits. All estimates use entropy balancing weights.
\end{minipage}
\end{figure}

\paragraph{Effects on Behavioral Problems.} BITSEA Problem scores show no significant associations with any domain in either linear or GAM models (all $p > 0.10$; see \cref{app:robustness} in the OSM for dose--response curves). Negative coefficients for CLASS-T Emotional-Behavioral and QCIT Language-Literacy suggest potential protective effects, but estimates remain nonsignificant.

Collectively, CDI IRT outcomes link strongly to language and cognitive supports, BITSEA Competence responds to emotional-behavioral support while exhibiting nonlinearities for other domains, and BITSEA Problem shows weak effects overall. The presence of nonlinearity (edf $> 1$) underscores the value of GAM models in revealing threshold dynamics. Consistent with prior work \citep{Burchinal2010, Hatfield2016}, these findings emphasize that domain-focused interactions benefit matching developmental outcomes.

\begin{figure}[htbp]
\centering
\includegraphics[width=0.75\textwidth]{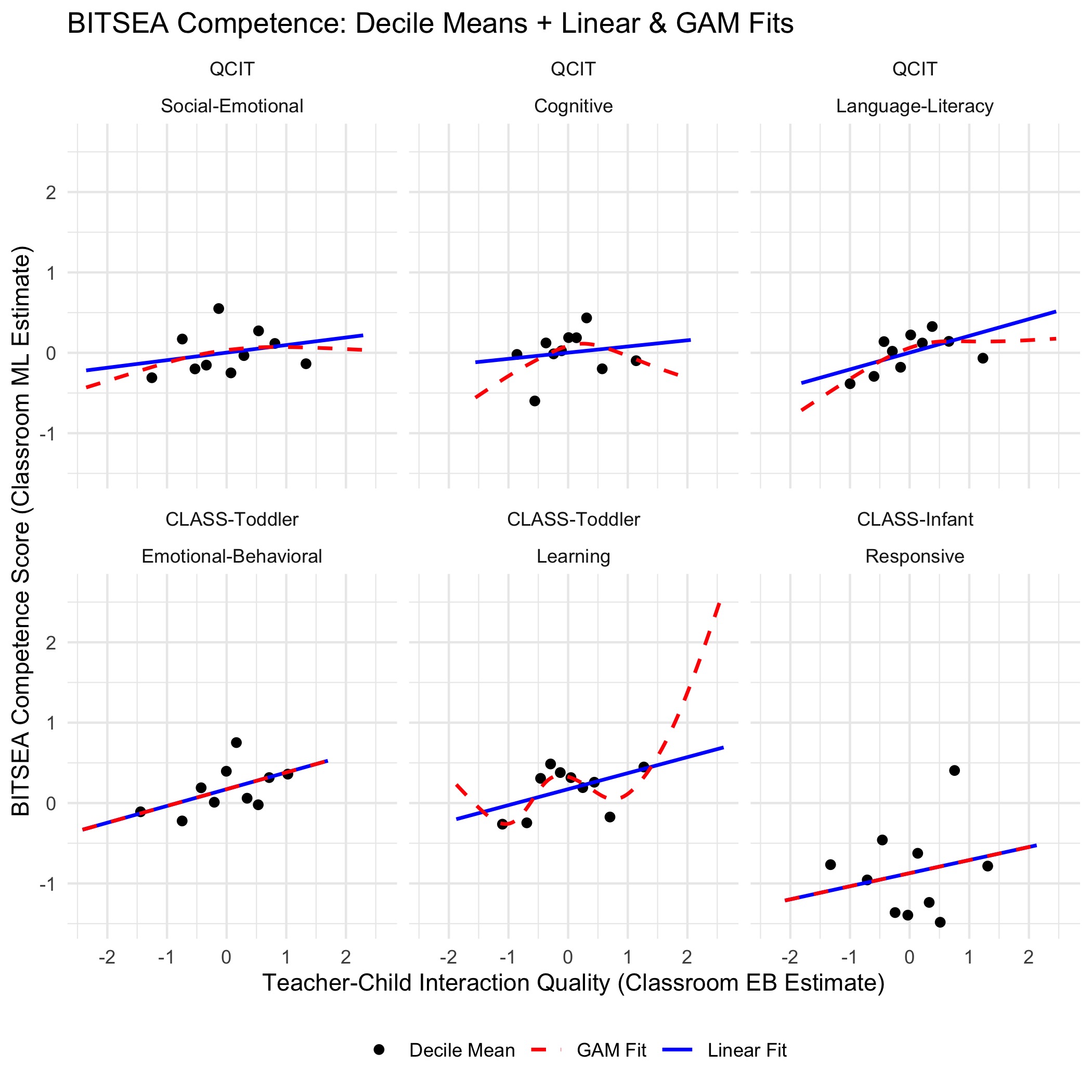}
\caption{Dose--response curves for teacher-reported BITSEA Competence scores across QCIT and CLASS domains}
\label{fig:dose-response-bitsea}
\begin{minipage}{0.95\textwidth}
\small
\emph{Note.} See \Cref{fig:dose-response-cdi} note for plot details.
\end{minipage}
\end{figure}

\subsubsection{Robustness Check: Addressing Teacher-Reporting Bias}
\label{subsubsec:robustness}

A potential concern is that teachers both provide the classroom interactions and assess child outcomes, potentially inflating associations through reporting bias. To address this, we replicated dose--response analyses using parent-reported outcomes (see \cref{app:robustness} in the OSM for full results). Parent-reported results largely corroborate the main findings. For CDI IRT, QCIT Cognitive (coef = 1.11, $p < 0.01$), QCIT Language-Literacy (coef = 1.27, $p < 0.01$), and CLASS-T Learning (coef = 1.04, $p < 0.01$) remain significant predictors with comparable effect sizes. For BITSEA Competence, CLASS-T Emotional-Behavioral maintains its significant positive association (coef = 0.16, $p = 0.04$). The domain-matching pattern persists across reporters: language/cognitive supports predict communicative outcomes, while emotional-behavioral support predicts social-emotional competence.

This consistency across independent reporters who observe children in different contexts and have different relationships with them provides robust evidence that findings reflect genuine associations rather than teacher-reporting artifacts. Given that parent-teacher concordance on child assessments is typically low \citep{Lim2021}, the replication of domain-specific patterns is particularly compelling.

\section{Discussion}
\label{sec:discussion}

This study set out to resolve the puzzling discrepancy between theoretical claims that high-quality teacher--child interactions are pivotal for infants' and toddlers' development and the modest or null effects often observed in empirical research. Drawing on the Baby FACES 2018 dataset, we identified four methodological limitations---item-level measurement error, center-level confounding, teacher-level covariate imbalance, and overlooked nonlinearities---that could systematically attenuate estimates of classroom process quality's true impact. Through a three-level GALAMM measurement model and weighted dose--response analyses, we found that these limitations---particularly item-level measurement error and center-level confounding---can mask genuine linkages between classroom process quality and child outcomes.

Our GALAMM approach illustrates how treating each observation item's contribution separately can mitigate the risk that item-specific measurement error obscures the true effects of classroom process quality. By incorporating item-level reliabilities directly into the model, rather than relying on one-size-fits-all ``domain'' composites, researchers can more accurately detect how teacher--child interactions support infants' and toddlers' development.

This approach represents a departure from standard observational measure scoring guidelines, highlighting an important distinction between different uses of quality measures. When instruments are employed for professional development or continuous quality improvement, the typical scoring structure may suffice due to its transparency and ease of interpretation. However, for accurately quantifying effects in research or program evaluation, multilevel latent variable modeling is indicated. To facilitate adoption, we provide complete R code on GitHub (\url{https://github.com/joonho112/baby-faces-2018-galamm}).

\subsection{Item-Level Error and the Need for Advanced Measurement Modeling}
\label{subsec:discuss-measurement}

A central insight from this study is the decomposition of item-level variance into Level 1 (measurement error), Level 2 (classroom latent factors), and Level 3 (center effects). We found that only about half of each item's total variance reflected genuine classroom-level differences, while the remaining half stemmed from item-level noise or center-level unobserved heterogeneity. In practice, summing or averaging all items equally can dilute or bias regression coefficients by overlooking how some items measure the underlying construct more reliably than others \citep{Gilbert2024, Howard2024, Levickis2024}.

\subsection{Center- and Classroom-Level Confounding: The Priority of Within-Center Comparisons}
\label{subsec:discuss-confounding}

After accounting for item-level noise and center-level heterogeneity, typical sources of within-center confounding, such as teacher qualifications, showed surprisingly small correlations with the latent process quality ``dose''. This suggests that once measurement error and cluster-level confounding are addressed, typical sources of within-center confounding (e.g., teacher qualifications) may pose less threat than previously assumed.

By contrast, center-level variance, approximately 15\% in our ICC decomposition, emerged as a nontrivial source of bias. In a nationally representative dataset like Baby FACES 2018, each ECE center encapsulates broader ecological conditions---ranging from regional resource disparities to state policy environments---thus combining multilayered potential confounders. The implications of this confounding are profound. When comparing two EHS classrooms---one in California and another in Alabama---without accounting for center- and state-level clustering, the observed differences may reflect community demographics, state funding streams, or local regulatory environments rather than genuine classroom-level interaction quality. Statistically, estimates that ignore state- or center-level clustering represent an ambiguous mixture: neither a pure between-center comparison nor a pure within-center comparison, but rather a weighted average of both that obscures the true relationship between classroom quality and child outcomes \citep{RabeHeskethSkrondal2022}. This methodological limitation extends to other national datasets such as the Head Start Family and Child Experiences Survey (FACES) and the Early Childhood Longitudinal Study, Birth Cohort (ECLS-B), suggesting that previous null findings regarding classroom process quality may stem from a failure to conduct appropriately localized comparisons.

From an ecological systems perspective \citep{Bronfenbrenner1998}, distinguishing microsystem influences (classroom interactions) from exosystem and macrosystem factors (center resources, community contexts, state policies) becomes essential for accurate estimation. Our study addresses this challenge through a two-pronged approach: using GALAMM to extract classroom-level predictors purged of center-level heterogeneity, and employing center-mean centering for child outcomes to isolate within-center variation. This strategy ensures that our estimates represent ``within-center'' comparisons---the difference in mean child outcomes between two classrooms within the same center that differ by one unit in process quality.

\subsection{Domain Matching and Linearities}
\label{subsec:discuss-domain-matching}

Our findings provide empirical support for the domain-specific developmental pathways hypothesized in our theoretical framework (\Cref{tab:theoretical-framework}). The observed domain-matching pattern---where language/cognitive supports predicted communicative outcomes while emotional-behavioral supports predicted social-emotional competence---aligns with theoretical predictions from both attachment \citep{Bowlby1969} and sociocultural \citep{Vygotsky1978} perspectives.

Linear patterns dominate many of these domain--outcome matches, indicating that no high ``quality threshold'' is strictly necessary before children reap benefits. Even modest incremental improvements in domain-relevant practices, such as additional language modeling, can translate into proportional gains for infants and toddlers. Nonetheless, nonlinearities emerge when a domain is overemphasized or mismatched with the child outcome of interest. Our results show that pushing CLASS--T Learning to extreme levels can yield diminishing returns for social-emotional competence, highlighting a \emph{plateau effect}. Similarly, classrooms heavily oriented toward cognitive or language activities do not automatically reduce behavior problems, reflecting minimal cross-domain spillover. These patterns illustrate the complexity of early learning processes: while domain-focused strategies can be powerful, they may not always benefit other developmental areas and, in certain cases, may lose effectiveness beyond a moderate ``dose.''

\subsection{Limitations}
\label{subsec:limitations}

Several limitations warrant consideration. First, while parent reports corroborated our findings, child outcomes were largely staff-reported, which may carry unobserved shared variance. Second, while the GALAMM partitions item-level noise, we lacked specific data on individual assessors, and residual error may reflect systematic assessor-level differences that could not be further decomposed. Third, the Baby FACES 2018 data lacks state identifiers, which means that center-level variance conflates with state-level heterogeneity, such as differing licensing regulations and quality standards. Finally, these data represent a single observational window, which may not capture the full temporal stability of interactions over a school year.

\subsection{Implications for Research, Policy, and Practice}
\label{subsec:implications}

These findings have significant implications for researchers, policymakers, and practitioners. While our analytic approach represents one solution to addressing measurement noise, the broader methodological principle applies regardless of the specific technique employed, establishing parameters that capture within-center effects while controlling for unobserved heterogeneity at higher ecological levels. By targeting these localized comparisons, researchers can disentangle genuine classroom-level associations from the complex web of center, community, and state influences.

The practical utility of a latent variable modeling approach can inform rater training protocols within research, evaluation, and practice. By identifying items with weaker factor loadings or higher error variance, rater training can strategically allocate resources where they are most needed. This targeted approach is particularly critical when using local raters rather than research team members, as previous studies have documented systematic rating differences between these groups \citep{Vitiello2018}. Programs could conduct pilot observations, apply GALAMM to identify problematic items, and then focus rater training accordingly.

For policymakers, the high level of measurement error in standard composites raises concerns regarding their use in high-stakes accountability systems, such as the HSDRS and Quality Rating and Improvement Systems. Classrooms or programs near quality thresholds may be misclassified due to factors other than actual classroom practices. Funding decisions and program sanctions based on these scores should incorporate uncertainty intervals or use multiple assessment points to reduce the impact of measurement error. At a minimum, policymakers should recognize that observed scores represent imperfect indicators of true quality and adjust accountability frameworks accordingly to avoid penalizing programs for measurement noise.

For practitioners, these results clarify how specific interactions drive distinct developmental outcomes. Teachers can apply the principle of domain-matching by tailoring interactions to children's needs. For example, a teacher noticing a plateau in children's social-emotional growth might look beyond general warmth to specific emotional-behavioral supports, such as proactive behavioral guidance. Child care programs aiming to improve children's communicative skills may choose to prioritize explicit language and cognitive scaffolding in teacher professional development. Moreover, the linear dose--response observed in certain domains suggests that incremental gains achieved through tailored coaching or modest resource shifts can steadily improve children's outcomes.

\section{Conclusion}
\label{sec:conclusion}

In conclusion, our findings challenge the ``null effect'' narrative and reinforce that robust measurement and analytic rigor can bring empirical evidence more in line with longstanding theoretical expectations \citep{PiantaHamre2009}. High-quality teacher--child interactions, especially those aligned to specific developmental domains, appear critical in supporting language and social-emotional outcomes in EHS. As EHS policymakers, practitioners, and researchers work to enhance infant--toddler care, rigorous methodological frameworks such as ours provide clearer guidance on what works, for whom, and under what conditions, helping ensure that young children receive the responsive, high-quality interactions they need to thrive.


\newpage
\printbibliography[title={References}]


\newpage
\appendix
\setcounter{section}{0}
\renewcommand{\thesection}{\Alph{section}}
\numberwithin{equation}{section}
\numberwithin{figure}{section}
\numberwithin{table}{section}

\renewcommand{\theHequation}{App.\thesection.\arabic{equation}}
\renewcommand{\theHfigure}{App.\thesection.\arabic{figure}}
\renewcommand{\theHtable}{App.\thesection.\arabic{table}}

\begin{center}
{\Large\bfseries Online Supplemental Materials}\\[1em]
{\large Beyond the Null Effect: Unmasking the True Impact of\\
Teacher--Child Interaction Quality on Child Outcomes in Early Head Start}\\[1.5em]
{\normalsize JoonHo Lee and Alison Hooper}\\[0.5em]
{\small The University of Alabama}
\end{center}

\vspace{2em}

\startcontents[appendices]
\printcontents[appendices]{}{1}{\textbf{Contents}\vskip1em\hrule\vskip1em}
\vskip1em\hrule\vskip2em

\addcontentsline{toc}{section}{Online Supplemental Materials}

\section{Analytic Sample Derivation and Attrition Analysis}
\label{app:sample}

This appendix documents how analytic samples were derived from Baby FACES 2018 and reports the exact analysis $N$s used in each model. It is designed to satisfy the main-text callout (``see Appendix A'') by providing transparent exclusion criteria, outcome- and reporter-specific $N$s for reproducibility, and a basic attrition check indicating whether the final analytic sample differs meaningfully from excluded cases on observed classroom characteristics.

\begin{figure}[H]
\centering
\includegraphics[width=0.95\textwidth]{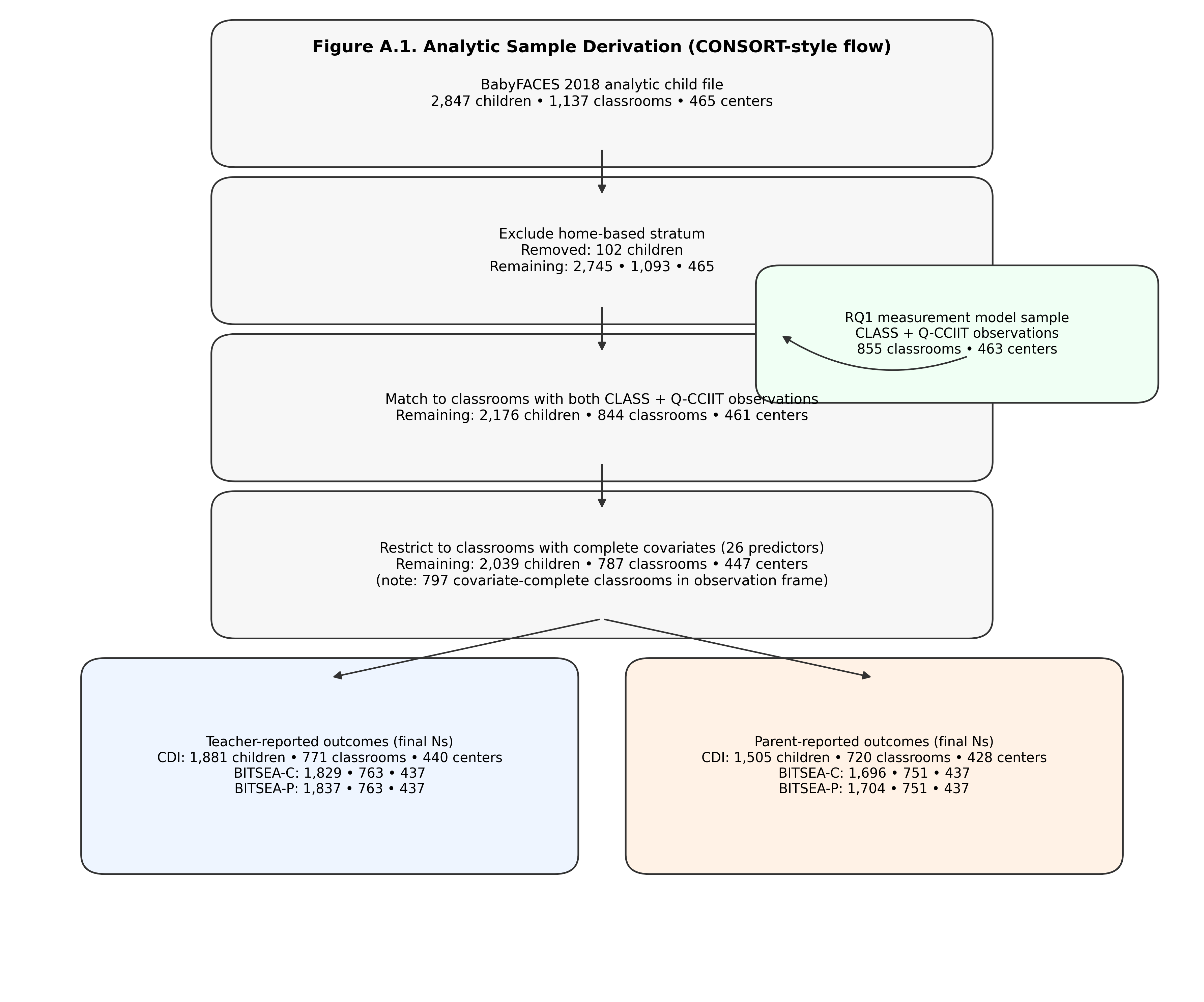}
\caption{CONSORT-style flow diagram depicting sequential exclusions from the processed Baby FACES 2018 analytic child file to the final analytic samples for (i) the RQ1 measurement model (classroom observation frame) and (ii) the RQ2 dose--response analyses (final $N$s by outcome and reporter).}
\label{fig:A1}
\end{figure}

\subsection{Units of Analysis and Counting Conventions}
\label{subsec:units}

RQ1 focuses on measurement. Accordingly, the unit of analysis is the item-by-classroom-by-center response used in the three-level generalized latent variable model (GALAMM). The RQ1 observation frame consists of classrooms with available observational data on both the Classroom Assessment Scoring System (CLASS) and the Quality of Caregiver--Child Interactions for Infants and Toddlers (QCIT; labeled as Q-CCIIT in the Baby FACES files).

RQ2 focuses on dose--response estimation. The ``dose'' is defined at the classroom level using empirical Bayes predictions of the latent process-quality factors from the RQ1 model. Child outcomes are therefore aggregated to the classroom level (classroom means) and centered within centers to isolate within-center variation, consistent with the main manuscript. Because the weighting and primary dose--response models rely on 26 classroom-level covariates, the RQ2 analytic sample additionally requires complete covariate information for these predictors.

Throughout this appendix, we report $N$s at three levels---children, classrooms, and centers---because these counts need not coincide at each stage. For example, some observed classrooms have no sampled children in the processed child file during the outcome window, and outcome availability differs by reporter and instrument.

\subsection{Classroom Observation Sample Derivation (RQ1)}
\label{subsec:rq1-sample}

The processed teacher/classroom file contains 859 observed classrooms across 464 centers. After requiring available observation data on both CLASS and QCIT (Q-CCIIT), the final RQ1 measurement-model frame includes 855 classrooms nested within 463 centers (\cref{tab:A1}). Within this frame, 148 classrooms were assessed using CLASS-Infant and 707 classrooms using CLASS-Toddler; the measurement model accommodates the resulting disjoint item sets via a long-format specification under a missing-at-random assumption (see \cref{app:galamm}).

\subsection{Teacher-Reported Outcome Sample Derivation (RQ2 Primary Analyses)}
\label{subsec:rq2-teacher}

The teacher-reported outcome pipeline begins with the processed analytic child file (2,847 children; 1,137 classrooms; 465 centers). We first excluded children in the home-based stratum (102 children removed). This restriction is by design: home-based EHS services employ a distinct service model (home visitors rather than classroom teachers), and classroom observation protocols (CLASS and QCIT) are not administered in home-visiting contexts. Our analyses therefore target center-based and mixed-service settings where classroom-level process quality can be observed (\cref{tab:A1}; \cref{fig:A1}).

\begin{table}[H]
\centering
\caption{Counts at Each Step (Children, Classrooms, Centers)}
\label{tab:A1}
\small
\begin{adjustbox}{max width=\textwidth}
\begin{tabular}{@{}llrrr@{}}
\toprule
Analysis Stream & Step & Children & Classrooms & Centers \\
\midrule
RQ1 classroom observations & Teacher/classroom records in processed file & --- & 859 & 464 \\
RQ1 classroom observations & CLASS observation data available & --- & 857 & 464 \\
RQ1 classroom observations & Q-CCIIT observation data available & --- & 855 & 463 \\
RQ1 classroom observations & Both CLASS + Q-CCIIT available (final RQ1) & --- & 855 & 463 \\
\addlinespace
RQ2 teacher-reported & Child sample in processed child file & 2,847 & 1,137 & 465 \\
RQ2 teacher-reported & Exclude home-based stratum & 2,745 & 1,093 & 465 \\
RQ2 teacher-reported & Match to classrooms with CLASS + Q-CCIIT & 2,176 & 844 & 461 \\
RQ2 teacher-reported & Restrict to complete covariates (26 predictors) & 2,039 & 787 & 447 \\
RQ2 teacher-reported & Teacher-reported CDI (IRT) observed (final) & 1,881 & 771 & 440 \\
\addlinespace
RQ2 parent-reported & Child sample in processed child file & 2,847 & 1,137 & 465 \\
RQ2 parent-reported & Exclude home-based stratum & 2,745 & 1,093 & 465 \\
RQ2 parent-reported & Match to classrooms with CLASS + Q-CCIIT & 2,176 & 844 & 461 \\
RQ2 parent-reported & Restrict to complete covariates (26 predictors) & 2,039 & 787 & 447 \\
RQ2 parent-reported & Parent-reported CDI (IRT) observed (final) & 1,505 & 720 & 428 \\
\bottomrule
\end{tabular}
\end{adjustbox}
\end{table}

Next, we linked children to the classroom observation frame by retaining children assigned to classrooms that have \emph{both} CLASS and QCIT observations. This matching step yielded 2,176 children in 844 classrooms across 461 centers. The reduction from 855 observed classrooms (the RQ1 frame) to 844 matched classrooms reflects that some observed classrooms did not have sampled children with child-level records in the processed child file during the assessment period.

Because covariate balancing requires complete values on the 26 classroom-level covariates used for weighting, we then restricted the analytic sample to covariate-complete classrooms. This step yielded 2,039 children in 787 classrooms across 447 centers. Finally, analytic $N$s vary by outcome depending on whether each teacher-reported measure is observed. \Cref{tab:A2} summarizes the final $N$s for teacher-reported CDI IRT language scores (1,881 children; 771 classrooms; 440 centers), BITSEA Competence (1,829; 763; 437), and BITSEA Problem (1,837; 763; 437).

\begin{table}[H]
\centering
\caption{Final Analytic $N$ by Outcome and Reporter}
\label{tab:A2}
\begin{tabular}{@{}llrrr@{}}
\toprule
Reporter & Outcome & Children & Classrooms & Centers \\
\midrule
Teacher-reported & CDI (IRT score) & 1,881 & 771 & 440 \\
Teacher-reported & BITSEA Competence & 1,829 & 763 & 437 \\
Teacher-reported & BITSEA Problem & 1,837 & 763 & 437 \\
\addlinespace
Parent-reported & CDI (IRT score) & 1,505 & 720 & 428 \\
Parent-reported & BITSEA Competence & 1,696 & 751 & 437 \\
Parent-reported & BITSEA Problem & 1,704 & 751 & 437 \\
\bottomrule
\end{tabular}
\end{table}

\subsection{Parent-Reported Outcome Sample Derivation (Robustness Analyses)}
\label{subsec:rq2-parent}

Parent-reported outcomes are analyzed as robustness checks to address potential same-reporter bias. The parent-report pipeline mirrors the teacher-report pipeline through observation matching and covariate completeness, using the same classroom-level dose definitions and the same set of covariates. Outcome-specific $N$s are generally smaller than the corresponding teacher-reported $N$s because parent-reported measures require completion of the parent interview and parent child-report modules, which exhibit higher nonresponse and partial-interview rates than staff reports. Final $N$s are reported in \cref{tab:A2} for parent-reported CDI IRT (1,505 children; 720 classrooms; 428 centers), BITSEA Competence (1,696; 751; 437), and BITSEA Problem (1,704; 751; 437).

\subsection{Missingness and Exclusion Summary}
\label{subsec:missingness}

Two mechanisms account for the largest reductions in sample size. First, children cannot contribute to RQ2 analyses unless they can be linked to a classroom with both CLASS and QCIT observations, because the exposure (``dose'') is defined at that classroom level. Second, the weighting strategy requires complete classroom-level covariates, and outcome-specific $N$s additionally depend on reporter- and measure-specific response patterns, which are more restrictive for parent reports. \Cref{app:galamm} provides additional detail on variable-level missingness and the handling of structurally missing items due to the two CLASS versions.

\subsection{Attrition/Selection Check}
\label{subsec:attrition}

To assess whether the final analytic sample differs systematically from excluded observed classrooms on measured characteristics, \cref{tab:A3} compares classrooms included in the teacher-reported CDI analytic sample with excluded observed classrooms (excluded due to missing covariates and/or missing linked CDI outcomes after observation matching). Differences are generally small to moderate in magnitude, with most standardized mean differences (SMDs) below 0.20 in absolute value and the largest approaching approximately 0.34. This pattern suggests that exclusion is not strongly patterned by the observed teacher/classroom characteristics shown in \cref{tab:A3}. While such comparisons cannot address unobserved selection, they provide a basic defense that the analytic sample is not markedly atypical on key measured covariates used in weighting.

\begin{table}[H]
\centering
\caption{Attrition/Exclusion Check: Included vs.\ Excluded Classrooms}
\label{tab:A3}
\begin{tabular}{@{}lccr@{}}
\toprule
Characteristic & Included & Excluded & SMD \\
\midrule
Years in Early Head Start & 6.52 (6.95) & 4.51 (4.71) & 0.30 \\
Years teaching infants/toddlers & 9.03 (7.67) & 9.21 (8.39) & $-0.02$ \\
Toddler classroom (vs.\ infant) & 83.1\% & 78.6\% & 0.12 \\
Teacher race: White, non-Hispanic & 35.3\% & 42.9\% & $-0.16$ \\
Teacher race: Black, non-Hispanic & 30.0\% & 23.8\% & 0.14 \\
Teacher ethnicity: Hispanic & 28.9\% & 23.8\% & 0.12 \\
Bachelor's degree or higher & 29.1\% & 40.5\% & $-0.24$ \\
Has ECE degree & 60.4\% & 61.9\% & $-0.03$ \\
Has CDA credential & 59.4\% & 54.8\% & 0.09 \\
English-only language profile & 57.3\% & 51.2\% & 0.12 \\
Has coach & 67.7\% & 56.0\% & 0.24 \\
Number of training topics & 7.73 (1.94) & 7.40 (1.98) & 0.17 \\
Child--adult ratio (observed) & 2.78 (0.78) & 2.84 (0.79) & $-0.07$ \\
Children enrolled in classroom & 7.68 (1.39) & 7.58 (1.59) & 0.07 \\
Quiet space available & 36.7\% & 28.6\% & 0.17 \\
Observed classrooms per center (in RQ1 frame) & 1.93 (0.26) & 1.83 (0.37) & 0.34 \\
\bottomrule
\end{tabular}

\smallskip
\footnotesize\textit{Note.} Values are means (standard deviations) for continuous variables and percentages for categorical variables. SMD = standardized mean difference. Included classrooms are those in the teacher-reported CDI analytic sample ($N = 771$); excluded classrooms are observed classrooms not in the final analytic sample ($N = 84$).
\end{table}

\section{GALAMM Measurement Model and Design-Based Missingness (Infant/Toddler CLASS)}
\label{app:galamm}

This appendix provides the full technical specification of the three-level Gaussian generalized additive latent and mixed model (GALAMM) used for RQ1 in the main manuscript. The model serves two roles. First, it extracts classroom-level latent process-quality factors while separating out (a) item-level measurement error and (b) center-level heterogeneity. Second, it accommodates the planned (design-based) item missingness induced by administering two age-specific CLASS instruments---CLASS-Infant in infant classrooms and CLASS-Toddler in toddler classrooms---while QCIT is administered in all classrooms and therefore forms a common measurement block.

All results in this appendix use the final RQ1 classroom-observation frame: $N = 855$ classrooms nested within $K = 463$ centers (148 infant classrooms; 707 toddler classrooms; see \cref{app:sample} for derivation). The analytic frame excludes the small number of observed classrooms that contain no QCIT item data (i.e., that lack the common measurement block required to link infant and toddler classrooms in a single joint model).

We write the measurement model in the language of multilevel latent variable models \citep[generalized multilevel SEM;][]{SkrondalRabeHesketh2004,RabeHesketh2004} while also making explicit its computational implementation as a linear mixed model in \texttt{galamm} \citep{Sorensen2023}.

\subsection{Model Specification (Main-Text Eq.\ 1 Restated)}
\label{subsec:model-spec}

\subsubsection{Observed Data Structure, Indices, and Standardization}
\label{subsubsec:data-structure}

Let $k = 1,\ldots,K$ index Early Head Start (EHS) centers, and let $j = 1,\ldots,J_k$ index classrooms nested in center $k$. Let $i = 1,\ldots,I$ index observed items, where $I=25$ in the analytic model (14 QCIT items, 7 CLASS-Toddler items, and 4 CLASS-Infant items; \cref{tab:B1}). Denote the raw item score by $y^{\mathrm{raw}}_{i,j,k}$.

Items are treated as approximately continuous and modeled with Gaussian errors. This is common in CFA/SEM practice when item scales have several ordered categories (e.g., 5--7 points), and it allows the model to be cast as a linear mixed model for scalable maximum likelihood estimation \citep[e.g.,][]{rhemtulla2012categorical}.

To place items on a common metric and to improve numerical conditioning of the likelihood, we standardize each item across its observed values in the analytic sample:
\begin{align}
y_{i,j,k} &= \frac{y^{\mathrm{raw}}_{i,j,k} - \bar y_i}{s_i}, \notag\\[6pt]
\bar y_i &= \frac{1}{n_i}\sum_{(j,k):\,R_{i,j,k}=1} y^{\mathrm{raw}}_{i,j,k}, \qquad
s_i^2 = \frac{1}{n_i-1}\sum_{(j,k):\,R_{i,j,k}=1}\left(y^{\mathrm{raw}}_{i,j,k}-\bar y_i\right)^2,
\label{eq:B1}
\end{align}
where $R_{i,j,k}\in\{0,1\}$ is the item-observability indicator (\cref{subsec:design-missingness}) and $n_i=\sum_{j,k} R_{i,j,k}$ is the number of observed responses for item $i$. For notational simplicity, all subsequent expressions use the standardized response $y_{i,j,k}$.

\subsubsection{Three-Level Gaussian Measurement Model}
\label{subsubsec:threelevel-model}

Let $\boldsymbol{\eta}_j=(\eta_{1,j},\ldots,\eta_{6,j})^\top$ denote the six classroom-level latent process-quality factors for classroom $j$. These correspond to three QCIT domains, two CLASS-Toddler domains, and one CLASS-Infant domain (\cref{tab:B1}). Let $\alpha_k$ denote a center-level random intercept capturing a center-wide shift shared by all classrooms in center $k$.

The paper fits the following three-level Gaussian measurement model (restating the main-text Eq.\ 1):
\begin{equation}
y_{i,j,k} \mid \boldsymbol{\eta}_j,\alpha_k 
\sim 
\mathcal{N}\!\left(
\beta_i + \sum_{f=1}^{6}\lambda_{i,f}\,\eta_{f,j} + \alpha_k,\;\sigma_\varepsilon^2
\right),
\label{eq:B2}
\end{equation}
where $\beta_i$ is an item-specific intercept, $\lambda_{i,f}$ is the factor loading of item $i$ on factor $f$, and $\sigma_\varepsilon^2$ is the common Level-1 residual variance. In line with the long-format implementation (\cref{subsubsec:long-format}), $\sigma_\varepsilon^2$ is shared across standardized items; heterogeneity in item reliability is primarily expressed through differences in the loadings $\lambda_{i,f}$ and the factor variances in $\boldsymbol{\Psi}^{(2)}$ (Appendix~C).

At Level 2 (classrooms), the latent-factor vector is modeled as multivariate normal:
\begin{equation}
\boldsymbol{\eta}_j \sim \mathcal{N}(\mathbf{0},\boldsymbol{\Psi}^{(2)}),
\label{eq:B3}
\end{equation}
where $\boldsymbol{\Psi}^{(2)}$ is a $6\times 6$ symmetric positive-definite covariance matrix (unstructured). At Level 3 (centers), the random intercept satisfies
\begin{equation}
\alpha_k \sim \mathcal{N}(0,\sigma_\alpha^2),
\label{eq:B4}
\end{equation}
with $\sigma_\alpha^2>0$. We assume mutual independence among $\boldsymbol{\eta}_j$, $\alpha_k$, and $\varepsilon_{i,j,k}$, and conditional independence of item responses given $(\boldsymbol{\eta}_j,\alpha_k)$, as in standard multilevel SEM/random-effects formulations \citep{SkrondalRabeHesketh2004,RabeHesketh2004}.

\subsubsection{Vector Form and Model-Implied Covariance}
\label{subsubsec:vector-form}

Define the full $I$-vector of standardized item responses for classroom $j$ in center $k$ as $\mathbf{y}_{j,k}=(y_{1,j,k},\ldots,y_{I,j,k})^\top$, the intercept vector $\boldsymbol{\beta}=(\beta_1,\ldots,\beta_I)^\top$, the $I\times 6$ loading matrix $\boldsymbol{\Lambda}=[\lambda_{i,f}]$, and $\mathbf{1}\in\mathbb{R}^I$ as the all-ones vector. Then \cref{eq:B2} can be written as
\begin{equation}
\mathbf{y}_{j,k}=\boldsymbol{\beta}+\boldsymbol{\Lambda}\boldsymbol{\eta}_j+\mathbf{1}\alpha_k+\boldsymbol{\varepsilon}_{j,k},
\qquad 
\boldsymbol{\varepsilon}_{j,k}\sim \mathcal{N}(\mathbf{0},\sigma_\varepsilon^2\mathbf{I}_I).
\label{eq:B5}
\end{equation}

Marginalizing over $\boldsymbol{\eta}_j$ and $\alpha_k$ yields the model-implied marginal covariance of the full item vector:
\begin{equation}
\Var(\mathbf{y}_{j,k})=\boldsymbol{\Lambda}\boldsymbol{\Psi}^{(2)}\boldsymbol{\Lambda}^\top+\sigma_\alpha^2\,\mathbf{1}\mathbf{1}^\top+\sigma_\varepsilon^2\mathbf{I}_I.
\label{eq:B6}
\end{equation}
Appendix~C derives the ICC-style Level-1/Level-2/Level-3 variance decomposition implied by \cref{eq:B6}.

In practice, no classroom is observed on all $I=25$ items because the CLASS block depends on classroom age group (\cref{subsec:design-missingness}). Let $\mathcal{O}_{j,k}\subset\{1,\ldots,I\}$ denote the observed-item set for classroom $(j,k)$, and let $\mathbf{S}_{j,k}$ be the selection matrix that extracts indices in $\mathcal{O}_{j,k}$. Then the observed item vector and its model-implied covariance are
\begin{equation}
\mathbf{y}^{\mathrm{obs}}_{j,k}=\mathbf{S}_{j,k}\mathbf{y}_{j,k},
\qquad
\Var(\mathbf{y}^{\mathrm{obs}}_{j,k})=\mathbf{S}_{j,k}\,\Var(\mathbf{y}_{j,k})\,\mathbf{S}_{j,k}^\top,
\label{eq:B6prime}
\end{equation}
so planned missingness enters the Gaussian likelihood by restricting each classroom's contribution to the appropriate observed sub-vector/submatrix.

\subsection{Item Inventory and Factor Assignment}
\label{subsec:item-inventory}

The measurement model uses 25 observed items assigned a priori to six latent factors under a simple-structure CFA (no cross-loadings; \cref{subsec:identification}). The three QCIT factors---QCIT Social--Emotional (4 items), QCIT Cognitive (3 items), and QCIT Language--Literacy (7 items)---are observed in all classrooms and therefore constitute the common measurement block. The CLASS block differs by classroom age group: CLASS-Toddler contributes two factors---Emotional--Behavioral Support (4 items) and Engaged Support for Learning (3 items)---measured only in toddler classrooms, whereas CLASS-Infant contributes one factor---Responsive Caregiving (4 items)---measured only in infant classrooms.

Consistent with the estimation scripts used for the paper, Negative Climate is excluded from CLASS-Toddler prior to model fitting. Empirically, this indicator exhibited severe skewness and instability when combined with the positively keyed CLASS-Toddler items, and conceptually it is reverse-coded relative to the remaining indicator set; excluding it improves numerical stability and yields a coherent positive-quality factor block.

\subsection{Identification and Parameter Constraints}
\label{subsec:identification}

\subsubsection{Simple-Structure Loading Pattern}
\label{subsubsec:simple-structure}

Let $f(i)\in\{1,\ldots,6\}$ denote the designated factor for item $i$ (\cref{tab:B1}). The loading matrix $\boldsymbol{\Lambda}=[\lambda_{i,f}]$ follows a simple structure:
\begin{equation}
\lambda_{i,f}=0 \quad \text{for all } f\neq f(i), 
\qquad 
\lambda_{i,f(i)} \text{ free (except for marker constraints below)}.
\label{eq:B7}
\end{equation}
Thus, each item measures exactly one latent factor, and cross-loadings are fixed to zero. This structure removes the rotational indeterminacy of unrestricted factor models and aligns the measurement model with the substantive domain definitions of QCIT and CLASS \citep{bollen1989structural}.

\subsubsection{Factor Scaling and Invariances}
\label{subsubsec:factor-scaling}

Absent constraints, latent variable models are invariant to reparameterizations of the latent space. In general, for any invertible matrix $\mathbf{A}\in\mathbb{R}^{6\times 6}$, the transformation
\begin{equation}
\boldsymbol{\eta}_j^*=\mathbf{A}\boldsymbol{\eta}_j,
\qquad
\boldsymbol{\Lambda}^*=\boldsymbol{\Lambda}\mathbf{A}^{-1},
\qquad
\boldsymbol{\Psi}^{*(2)}=\mathbf{A}\boldsymbol{\Psi}^{(2)}\mathbf{A}^\top,
\label{eq:B8a}
\end{equation}
produces an observationally equivalent model because $\boldsymbol{\Lambda}^*\boldsymbol{\eta}_j^*=\boldsymbol{\Lambda}\boldsymbol{\eta}_j$ and $\boldsymbol{\Lambda}^*\boldsymbol{\Psi}^{*(2)}\boldsymbol{\Lambda}^{*\top}=\boldsymbol{\Lambda}\boldsymbol{\Psi}^{(2)}\boldsymbol{\Lambda}^\top$ \citep{bollen1989structural}.

Under the simple-structure pattern in \cref{eq:B7}, admissible transformations that preserve the enforced zero pattern reduce primarily to diagonal rescalings of each factor. For any diagonal matrix $\mathbf{D}=\mathrm{diag}(d_1,\ldots,d_6)$ with nonzero $d_f$,
\begin{equation}
\boldsymbol{\Lambda}\boldsymbol{\eta}_j=
(\boldsymbol{\Lambda}\mathbf{D})(\mathbf{D}^{-1}\boldsymbol{\eta}_j),
\qquad 
\mathbf{D}^{-1}\boldsymbol{\eta}_j \sim \mathcal{N}\!\left(\mathbf{0},\,\mathbf{D}^{-1}\boldsymbol{\Psi}^{(2)}\mathbf{D}^{-\top}\right),
\label{eq:B8}
\end{equation}
so only the product $\boldsymbol{\Lambda}\boldsymbol{\Psi}^{(2)}\boldsymbol{\Lambda}^\top$ is identified from second-order moments of the indicators \citep[cf.][]{bollen1989structural}. We fix factor scales via marker-variable (reference-loading) constraints: for each factor $f$, we choose a reference item $i_f$ and set
\begin{equation}
\lambda_{i_f,f}=1,
\qquad f=1,\ldots,6.
\label{eq:B9}
\end{equation}
The reference items are indicated in \cref{tab:B1} (``Anch.\ = Yes''). All other nonzero loadings are estimated freely.

The model fixes factor means to zero, $\E[\boldsymbol{\eta}_j]=\mathbf{0}$, consistent with a random-effects interpretation. Item means are captured by $\beta_i$. Because indicators are standardized per \cref{eq:B1}, the fitted $\beta_i$ are expected to be close to 0; they are nevertheless retained for notational and implementation completeness.

\begin{table}[htbp]
\centering
\caption{Item-to-factor assignment (25 items $\to$ 6 factors).}
\label{tab:B1}
\footnotesize
\begin{adjustbox}{max width=\textwidth}
\begin{tabular}{@{}lllllc@{}}
\toprule
Item & Variable & Instr. & Factor & Anch. & Range \\
\midrule
Responding to social cues & resp\_social\_cues & QCIT & QCIT Soc-Emot & Yes & 1.67--6.83 \\
Responding to emotional cues & resp\_emotional\_cues & QCIT & QCIT Soc-Emot & No & 1.33--6.83 \\
Builds positive relationship & builds\_pos\_relation & QCIT & QCIT Soc-Emot & No & 1.00--7.00 \\
Supporting peer interaction & sup\_peer\_interaction & QCIT & QCIT Soc-Emot & No & 1.00--6.17 \\
\addlinespace
Supporting object exploration & sup\_object\_explore & QCIT & QCIT Cognitive & Yes & 1.00--7.00 \\
Scaffolding problem solving & scaff\_problem\_solve & QCIT & QCIT Cognitive & No & 1.00--7.00 \\
Number of unique concepts & unique\_concepts\_7cat & QCIT & QCIT Cognitive & No & 1.00--7.00 \\
\addlinespace
Use of varied vocabulary & caregiver\_varied\_vocab & QCIT & QCIT Lang-Lit & Yes & 1.00--6.60 \\
Use of questions & caregiver\_questions & QCIT & QCIT Lang-Lit & No & 1.33--6.67 \\
Conversational turn-taking & conv\_turn\_taking & QCIT & QCIT Lang-Lit & No & 1.33--7.00 \\
Extending child language & extend\_child\_lang & QCIT & QCIT Lang-Lit & No & 1.00--7.00 \\
Engaging children in books & engage\_in\_books & QCIT & QCIT Lang-Lit & No & 1.00--7.00 \\
Variety of words & variety\_words & QCIT & QCIT Lang-Lit & No & 1.00--7.00 \\
Variety of sentence styles & variety\_sent\_types & QCIT & QCIT Lang-Lit & No & 1.00--7.00 \\
\addlinespace
Positive climate & toddler\_pos\_climate & CLASS-T & CLASS-T EmotBeh & Yes & 2.50--7.00 \\
Teacher sensitivity (Toddler) & toddler\_teacher\_sens & CLASS-T & CLASS-T EmotBeh & No & 2.50--7.00 \\
Regard for child perspectives & toddler\_child\_persp & CLASS-T & CLASS-T EmotBeh & No & 2.00--7.00 \\
Behavioral guidance & toddler\_behav\_guidance & CLASS-T & CLASS-T EmotBeh & No & 1.50--6.75 \\
\addlinespace
Facilitation of learning/dev. & toddler\_fac\_learning & CLASS-T & CLASS-T Learning & Yes & 1.00--7.00 \\
Quality of feedback & toddler\_quality\_feedback & CLASS-T & CLASS-T Learning & No & 1.00--6.75 \\
Language modeling & toddler\_lang\_model & CLASS-T & CLASS-T Learning & No & 1.00--6.50 \\
\addlinespace
Relational climate & infant\_rel\_climate & CLASS-I & CLASS-I Respons & Yes & 3.00--7.00 \\
Teacher sensitivity (Infant) & infant\_teacher\_sens & CLASS-I & CLASS-I Respons & No & 2.25--7.00 \\
Facilitated exploration & infant\_fac\_explore & CLASS-I & CLASS-I Respons & No & 1.75--6.50 \\
Early language support & infant\_early\_lang & CLASS-I & CLASS-I Respons & No & 1.50--7.00 \\
\bottomrule
\end{tabular}
\end{adjustbox}
\vspace{0.5em}

\noindent\footnotesize\textit{Note.} All items have expected direction $+$ (higher = higher quality). Instr.\ = Instrument; Anch.\ = Anchor; Soc-Emot = Social-Emotional; Lang-Lit = Language-Literacy; EmotBeh = Emotional-Behavioral; Respons = Responsive Caregiving.
\end{table}

\subsubsection{Residual Structure and Center Effects}
\label{subsubsec:residual-structure}

The Level-1 residuals are assumed conditionally independent with a common variance,
\begin{equation}
\varepsilon_{i,j,k}\overset{\mathrm{iid}}{\sim}\mathcal{N}(0,\sigma_\varepsilon^2),
\label{eq:B10}
\end{equation}
so that $\Var(\boldsymbol{\varepsilon}_{j,k})=\sigma_\varepsilon^2\mathbf{I}_I$ in \cref{eq:B5}. In a classical CFA, one often allows item-specific residual variances. Here we impose a common residual variance for two reasons: (i) the long-format mixed-model likelihood used by \texttt{galamm} in the Gaussian case naturally assumes a single residual variance parameter, and (ii) with standardized indicators, this constraint yields a parsimonious and numerically stable model while still allowing substantial between-item differences in marginal reliability through item-specific loadings and freely estimated factor variances (Appendix~C).

At Level 3, the center random intercept $\alpha_k\sim\mathcal{N}(0,\sigma_\alpha^2)$ captures center-wide shifts common to all item responses. This term ensures that between-center heterogeneity does not inflate classroom-level factor variances and covariances (Appendix~C), and it makes explicit that the classroom factors $\boldsymbol{\eta}_j$ represent within-center deviations from a center-wide baseline (conditional on $\alpha_k$).

\subsubsection{Identifiability under Age-Specific Item Blocks}
\label{subsubsec:identifiability-age}

Because CLASS-Infant and CLASS-Toddler are never observed in the same classroom, the design implies that not all elements of the unstructured classroom covariance matrix $\boldsymbol{\Psi}^{(2)}$ are identified from the observed data. In particular, covariances between the infant-only CLASS factor and the toddler-only CLASS factors are not informed by within-classroom co-measurement and therefore drop out of the observed-data likelihood under \cref{eq:B2,eq:B3,eq:B4}. \Cref{subsec:disjoint-factors} formalizes this identification logic and clarifies how QCIT provides the empirical overlap needed to identify QCIT--CLASS correlations despite disjoint CLASS blocks.

\subsection{Estimation and Implementation}
\label{subsec:estimation}

\subsubsection{Observed-Data Likelihood under Planned Missingness}
\label{subsubsec:obs-likelihood}

Let $\mathcal{O}_{j,k}\subset\{1,\ldots,I\}$ denote the set of items observed for classroom $j$ in center $k$. Under the design, $|\mathcal{O}_{j,k}|=18$ for infant classrooms (14 QCIT + 4 CLASS-Infant) and $|\mathcal{O}_{j,k}|=21$ for toddler classrooms (14 QCIT + 7 CLASS-Toddler), up to a small amount of incidental missingness (\cref{subsec:design-missingness}). Let $\mathbf{y}^{\mathrm{obs}}_{j,k}\in\mathbb{R}^{|\mathcal{O}_{j,k}|}$ be the observed sub-vector obtained by selecting indices in $\mathcal{O}_{j,k}$.

Denote the full parameter vector by
\begin{equation}
\theta=\{\boldsymbol{\beta},\boldsymbol{\Lambda},\boldsymbol{\Psi}^{(2)},\sigma_\alpha^2,\sigma_\varepsilon^2\}.
\label{eq:B11}
\end{equation}
Because the model is Gaussian and the missingness is ignorable (\cref{subsec:design-missingness}), the observed-data likelihood is obtained by integrating the joint density of the observed indicators over the random effects:
\begin{align}
L(\theta)&=
\prod_{k=1}^{K}
\int
\Bigg[
\prod_{j=1}^{J_k}
\int
\bigg\{
\prod_{i\in \mathcal{O}_{j,k}}
\phi\!\left(y_{i,j,k}\,;\,\beta_i+\sum_{f=1}^6\lambda_{i,f}\eta_{f,j}+\alpha_k,\sigma_\varepsilon^2\right)
\bigg\} \notag\\
&\qquad\qquad\qquad\qquad\times
\phi(\boldsymbol{\eta}_j;\mathbf{0},\boldsymbol{\Psi}^{(2)})
\,d\boldsymbol{\eta}_j
\Bigg]
\phi(\alpha_k;0,\sigma_\alpha^2)\,d\alpha_k,
\label{eq:B12}
\end{align}
where $\phi(\cdot;\mu,\sigma^2)$ is the univariate normal density and $\phi(\cdot;\mathbf{0},\boldsymbol{\Psi}^{(2)})$ is the multivariate normal density.

Because the integrand is Gaussian in $(\{\boldsymbol{\eta}_j\},\alpha_k)$ for fixed $\theta$, the marginal likelihood \cref{eq:B12} is equivalent to the multivariate normal likelihood of a linear mixed model \citep{laird1982random}. In \texttt{galamm}, this marginal likelihood is maximized by a scalable algorithm based on Laplace approximation and sparse matrix operations \citep{Sorensen2023}. In the present Gaussian setting, the Laplace approximation is exact.

In addition, letting $\mathbf{y}_k$ stack all observed responses in center $k$ (\cref{subsubsec:long-format}) and letting $\mathbf{V}_k$ denote its marginal covariance, the marginal Gaussian log-likelihood can be written explicitly as
\begin{align}
\ell(\theta)&=\sum_{k=1}^{K}\bigg[
-\frac{n_k}{2}\log(2\pi)
-\frac{1}{2}\log|\mathbf{V}_k| \notag\\
&\qquad\qquad -\frac{1}{2}(\mathbf{y}_k-\mathbf{X}_k\boldsymbol{\beta})^\top\mathbf{V}_k^{-1}(\mathbf{y}_k-\mathbf{X}_k\boldsymbol{\beta})
\bigg],
\label{eq:B12prime}
\end{align}
where $n_k=\dim(\mathbf{y}_k)$. This form makes explicit the two computational bottlenecks (log-determinant and solve), which are handled efficiently via sparse methods in \texttt{galamm} \citep{Sorensen2023}.

\subsubsection{Long-Format Mixed-Model Representation}
\label{subsubsec:long-format}

To connect \cref{eq:B2,eq:B3,eq:B4} to implementation, consider the long-format data representation where each row corresponds to one observed pair $(i,j,k)$ with $i\in\mathcal{O}_{j,k}$. Stack all observed item responses in center $k$ into a vector $\mathbf{y}_k\in\mathbb{R}^{n_k}$, where $n_k=\sum_{j=1}^{J_k}|\mathcal{O}_{j,k}|$. Then we can write
\begin{equation}
\mathbf{y}_k=
\mathbf{X}_k\boldsymbol{\beta}
+
\mathbf{Z}_k\mathbf{b}_k
+
\boldsymbol{\varepsilon}_k,
\qquad 
\boldsymbol{\varepsilon}_k\sim \mathcal{N}(\mathbf{0},\sigma_\varepsilon^2\mathbf{I}_{n_k}),
\label{eq:B13}
\end{equation}
where $\mathbf{X}_k$ is the fixed-effect design matrix for item intercepts, implemented as \texttt{(0+item)} in R so that each row selects the relevant $\beta_i$. The center-specific random-effects vector stacks all classroom factor vectors and the center intercept,
\begin{equation}
\mathbf{b}_k=\big(\boldsymbol{\eta}_{1},\ldots,\boldsymbol{\eta}_{J_k},\alpha_k\big)^\top,
\label{eq:B14}
\end{equation}
and $\mathbf{Z}_k$ is the corresponding sparse random-effects design matrix: it is block-diagonal across classrooms for the $\boldsymbol{\eta}_j$ components and includes an additional column of ones for $\alpha_k$. Crucially, rows of the classroom blocks implement the loading pattern: the design row for an observed response $y_{i,j,k}$ is the loading vector $\boldsymbol{\lambda}_i^\top=(\lambda_{i,1},\ldots,\lambda_{i,6})$, with zeros enforced for non-designated factors per \cref{eq:B7}.

Under \cref{eq:B3,eq:B4}, the random effects satisfy
\begin{equation}
\mathbf{b}_k \sim \mathcal{N}\!\left(\mathbf{0},\,\mathbf{G}_k\right),
\qquad 
\mathbf{G}_k=
\begin{pmatrix}
\mathbf{I}_{J_k}\otimes\boldsymbol{\Psi}^{(2)} & \mathbf{0}\\
\mathbf{0}^\top & \sigma_\alpha^2
\end{pmatrix},
\label{eq:B15}
\end{equation}
so the marginal model for $\mathbf{y}_k$ is multivariate normal:
\begin{equation}
\mathbf{y}_k \sim \mathcal{N}\!\left(\mathbf{X}_k\boldsymbol{\beta},\, \mathbf{Z}_k\mathbf{G}_k\mathbf{Z}_k^\top + \sigma_\varepsilon^2\mathbf{I}_{n_k}\right).
\label{eq:B16}
\end{equation}
\Cref{eq:B16} makes explicit how design-based missingness enters the likelihood: classrooms with fewer administered items simply contribute fewer rows to $\mathbf{X}_k$ and $\mathbf{Z}_k$, producing an ``unbalanced'' but fully likelihood-based multilevel model.

\subsubsection{Empirical Bayes (Posterior Mean) Factor Scores}
\label{subsubsec:eb-scores}

For classroom $j$ in center $k$, the empirical Bayes (EB) estimate $\hat{\boldsymbol{\eta}}_j$ is the posterior mean of $\boldsymbol{\eta}_j$ given all observed items in center $k$:
\begin{equation}
\hat{\boldsymbol{\eta}}_j=
\E_{\hat\theta}\!\left[\boldsymbol{\eta}_j \mid \{\mathbf{y}^{\mathrm{obs}}_{j',k}\}_{j'=1}^{J_k}\right],
\label{eq:B17}
\end{equation}
evaluated at the fitted parameter estimates $\hat\theta$. In the Gaussian case, EB estimates coincide with BLUPs/random-effect predictors and exhibit the usual shrinkage toward zero proportional to posterior uncertainty \citep{robinson1991blup}. These EB predictions are used as measurement-error-adjusted classroom ``dose'' variables in the dose--response analyses (Appendix~D). For age-specific factors, EB predictions are only straightforwardly interpretable for classrooms where the corresponding item block is administered. If a factor has no directly administered indicators in a classroom, its posterior can still update indirectly via correlations with factors that \textit{are} measured (e.g., via QCIT factors), but such a prediction is best viewed as a model-based imputation driven by cross-factor covariance assumptions rather than by direct measurement in that classroom. Accordingly, the paper's downstream uses of factor scores restrict attention to factors that are directly measured in the relevant classroom type (Appendices~D--E).

\subsection{Design-Based Missingness: Infant vs.\ Toddler CLASS Blocks}
\label{subsec:design-missingness}

\subsubsection{Deterministic Missingness Indicators and Observed-Item Sets}
\label{subsubsec:missingness-indicators}

Let $G_{j,k}\in\{\mathrm{Infant},\mathrm{Toddler}\}$ denote the observed classroom age group. Let $A_i\in\{\mathrm{QCIT},\mathrm{CLASS\text{-}T},\mathrm{CLASS\text{-}I}\}$ denote the instrument/block membership of item $i$ (\cref{tab:B1}). Define the design-based observability indicator
\begin{align}
R^{\mathrm{design}}_{i,j,k}&=
\mathbf{1}\{A_i=\mathrm{QCIT}\}
+
\mathbf{1}\{A_i=\mathrm{CLASS\text{-}T},\,G_{j,k}=\mathrm{Toddler}\} \notag\\
&\quad+
\mathbf{1}\{A_i=\mathrm{CLASS\text{-}I},\,G_{j,k}=\mathrm{Infant}\}.
\label{eq:B18}
\end{align}
Thus, QCIT items satisfy $R^{\mathrm{design}}_{i,j,k}=1$ for all classrooms, CLASS-Toddler items satisfy $R^{\mathrm{design}}_{i,j,k}=1$ only for toddler classrooms, and CLASS-Infant items satisfy $R^{\mathrm{design}}_{i,j,k}=1$ only for infant classrooms. The observed-item set in classroom $(j,k)$ is then
\begin{equation}
\mathcal{O}_{j,k}=\{i: R_{i,j,k}=1\},
\qquad 
R_{i,j,k}=R^{\mathrm{design}}_{i,j,k}\times R^{\mathrm{incidental}}_{i,j,k},
\label{eq:B19}
\end{equation}
where $R^{\mathrm{incidental}}_{i,j,k}$ captures the small amount of non-design missingness (e.g., unscored items or processing loss). \Cref{tab:B2} summarizes the resulting blockwise observation structure in the analytic sample.

\begin{table}[htbp]
\centering
\caption{Design-based missingness block structure (administered items by classroom type).}
\label{tab:B2}
\small
\begin{tabular}{@{}lcccccc@{}}
\toprule
 & & \multicolumn{3}{c}{Items Observed} & \\
\cmidrule(lr){3-5}
Type & $N$ & QCIT (14) & CLASS-T (7) & CLASS-I (4) & Not Admin. \\
\midrule
Infant & 148 & 13.78 (98.5\%) & --- & 4.00 (100\%) & 7 \\
Toddler & 707 & 13.91 (99.4\%) & 7.00 (100\%) & --- & 4 \\
\bottomrule
\end{tabular}

\vspace{0.5em}
\noindent\footnotesize\textit{Note.} Percentages indicate observed items relative to the maximum possible within each block. ``---'' indicates items not administered by design.
\end{table}

\subsubsection{MAR-by-Design and Ignorability}
\label{subsubsec:mar-ignorability}

Planned missingness designs are often described as ``missing by design'' or ``matrix sampling'' \citep{graham2006planned}. Here, the missingness mechanism is deterministically driven by the observed age-group label $G_{j,k}$: conditional on $G_{j,k}$, the instrument block (CLASS-Infant vs.\ CLASS-Toddler) administered in classroom $(j,k)$ is fixed. Formally, for the design component,
\begin{equation}
\mathbb{P}\!\left(R^{\mathrm{design}} \mid \mathbf{Y}^{\mathrm{full}},\,G\right)=
\mathbb{P}\!\left(R^{\mathrm{design}} \mid G\right),
\label{eq:B20}
\end{equation}
where $\mathbf{Y}^{\mathrm{full}}$ denotes the (conceptual) complete item-response array. Thus the design-based missingness satisfies Rubin's missing-at-random condition (MAR) given $G$ \citep{rubin1976inference}, and because the measurement-model parameters $\theta$ are distinct from the missingness mechanism, the missingness is ignorable for likelihood-based inference \citep{little2002statistical,enders2010applied}. Concretely, maximum likelihood estimation of \cref{eq:B12} is valid using only the observed rows in the long-format dataset; no explicit imputation is required.

The only substantive requirement for linking the two age groups in a single joint model is the presence of the common measurement block. In BabyFACES, QCIT provides that overlap: it is observed in (almost) all classrooms irrespective of age group. Classrooms with no QCIT observations provide no information about the QCIT factors and no bridge between infant and toddler measurement blocks; these are excluded from the analytic sample, yielding the $N$ and $K$ above.

\subsubsection{Schematic Representation of the Design}
\label{subsubsec:schematic}

\begin{figure}[htbp]
\centering
\includegraphics[width=0.95\textwidth]{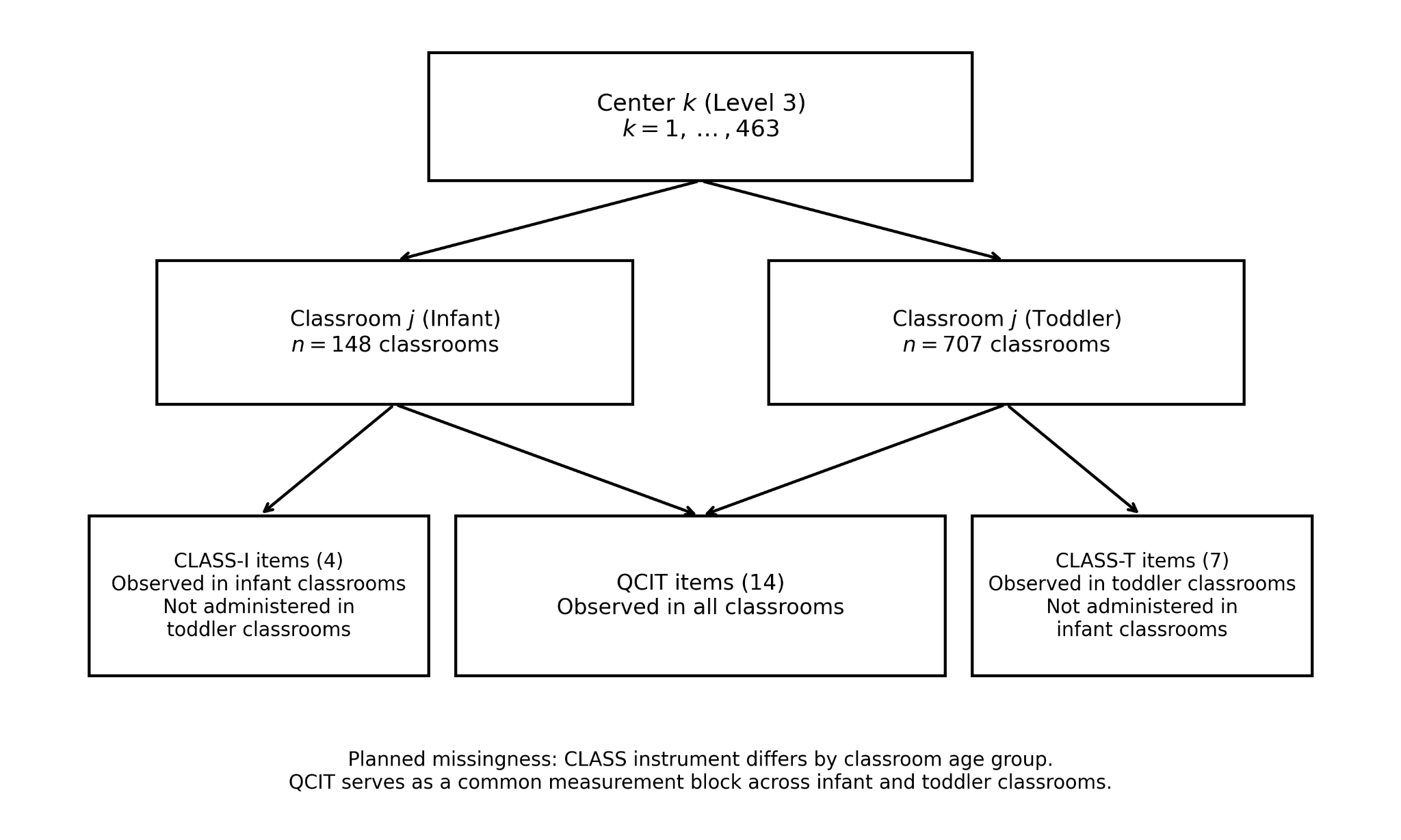}
\caption{Three-level hierarchical measurement structure with design-based missingness. Level~3 represents EHS centers ($k=1,\ldots,463$). Level~2 represents classrooms nested within centers and distinguished by age group: infant classrooms ($n=148$) and toddler classrooms ($n=707$). Level~1 represents observed items. QCIT items (14) are observed in all classrooms and therefore form the common measurement block, whereas CLASS-Infant items (4) are observed only in infant classrooms and CLASS-Toddler items (7) only in toddler classrooms.}
\label{fig:B1}
\end{figure}

\begin{figure}[htbp]
\centering
\includegraphics[width=0.85\textwidth]{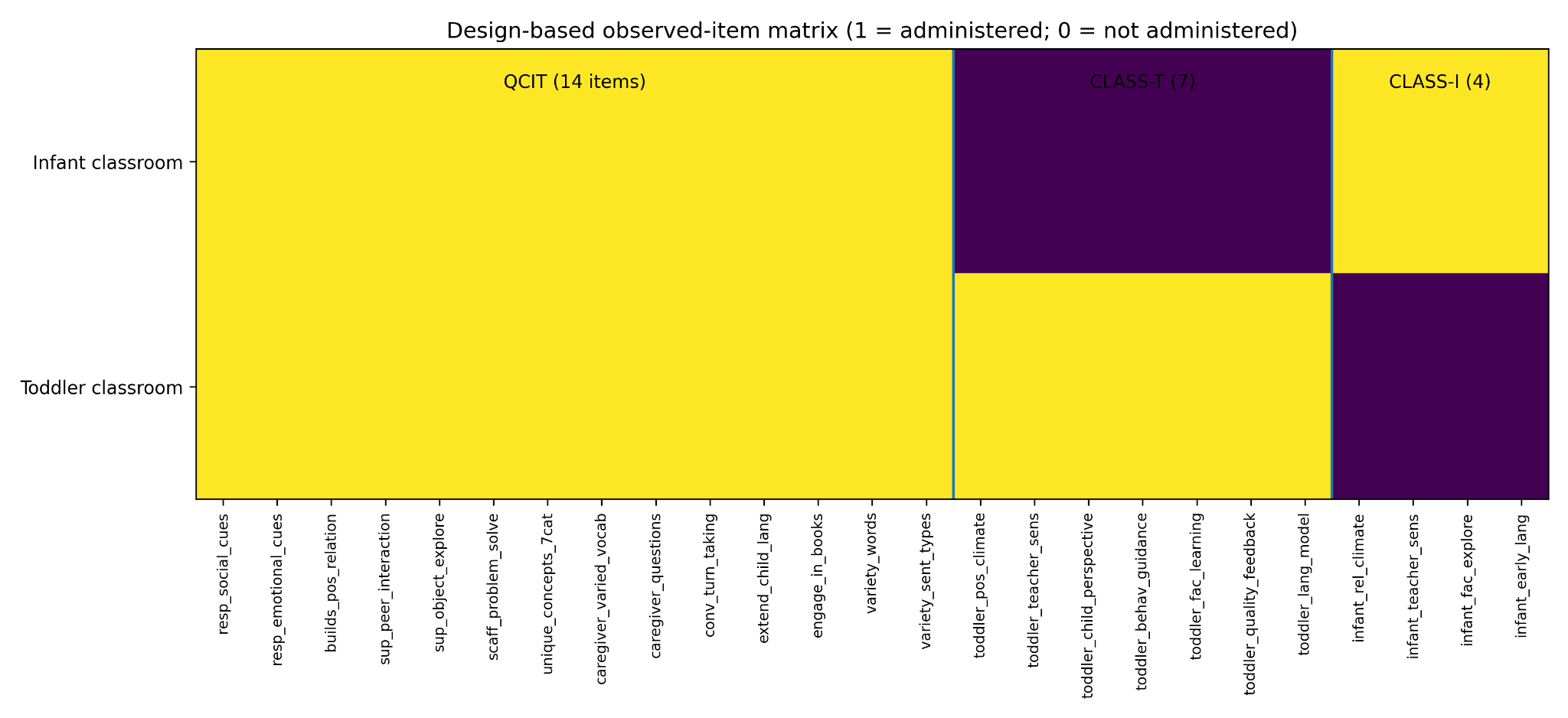}
\caption{Design-based observed-item matrix for QCIT and CLASS instruments. Yellow cells indicate administered items and purple cells indicate items not administered by design. Rows correspond to classroom age groups (infant vs.\ toddler), and columns correspond to individual items grouped by instrument.}
\label{fig:B2}
\end{figure}

\subsection{Disjoint Factor Correlations and the QCIT ``Anchor'' Logic}
\label{subsec:disjoint-factors}

This section formalizes which elements of the classroom-level covariance matrix $\boldsymbol{\Psi}^{(2)}$ are identified under the BabyFACES instrument-by-age design, and clarifies the precise sense in which QCIT provides an empirical ``anchor.''

\subsubsection{Block Partition of Factors and Covariance Matrix}
\label{subsubsec:block-partition}

Reorder the six classroom factors into three blocks:
\begin{equation}
\boldsymbol{\eta}_j=
\big(\boldsymbol{\eta}_{Q,j}^\top,\boldsymbol{\eta}_{T,j}^\top,\eta_{I,j}\big)^\top,
\label{eq:B21}
\end{equation}
where $\boldsymbol{\eta}_{Q,j}\in\mathbb{R}^3$ are the three QCIT factors, $\boldsymbol{\eta}_{T,j}\in\mathbb{R}^2$ are the two CLASS-Toddler factors, and $\eta_{I,j}\in\mathbb{R}$ is the CLASS-Infant factor. Conformably, partition the covariance matrix as
\begin{equation}
\boldsymbol{\Psi}^{(2)}=
\begin{pmatrix}
\boldsymbol{\Psi}_{QQ} & \boldsymbol{\Psi}_{QT} & \boldsymbol{\Psi}_{QI}\\
\boldsymbol{\Psi}_{TQ} & \boldsymbol{\Psi}_{TT} & \boldsymbol{\Psi}_{TI}\\
\boldsymbol{\Psi}_{IQ} & \boldsymbol{\Psi}_{IT} & \Psi_{II}
\end{pmatrix},
\label{eq:B22}
\end{equation}
with $\boldsymbol{\Psi}_{QQ}\in\mathbb{R}^{3\times 3}$, $\boldsymbol{\Psi}_{TT}\in\mathbb{R}^{2\times 2}$, $\Psi_{II}\in\mathbb{R}$, and $\boldsymbol{\Psi}_{TI}\in\mathbb{R}^{2\times 1}$. The key ``disjoint'' correlations correspond to $\boldsymbol{\Psi}_{TI}$ (or equivalently $\boldsymbol{\Psi}_{IT}$).

\subsubsection{Group-Specific Measurement Blocks and Which Parameters Enter the Likelihood}
\label{subsubsec:group-specific}

Let $g\in\{\mathrm{Infant},\mathrm{Toddler}\}$ index classroom type. For each type $g$, define the observed item vector $\mathbf{y}^{(g)}_{j,k}$ and the corresponding loading matrix $\boldsymbol{\Lambda}^{(g)}$ that selects the administered items. Because the design deterministically omits one CLASS block, the relevant loading matrices have structural zeros:

\textit{Infant classrooms} observe QCIT + CLASS-Infant items, hence their loading matrix has no columns for $\boldsymbol{\eta}_{T}$:
\begin{equation}
\boldsymbol{\Lambda}^{(\mathrm{Infant})}=
\big[\boldsymbol{\Lambda}^{(\mathrm{Infant})}_Q \;\; \mathbf{0} \;\; \boldsymbol{\Lambda}^{(\mathrm{Infant})}_I\big].
\label{eq:B23}
\end{equation}

\textit{Toddler classrooms} observe QCIT + CLASS-Toddler items, hence their loading matrix has no column for $\eta_I$:
\begin{equation}
\boldsymbol{\Lambda}^{(\mathrm{Toddler})}=
\big[\boldsymbol{\Lambda}^{(\mathrm{Toddler})}_Q \;\; \boldsymbol{\Lambda}^{(\mathrm{Toddler})}_T \;\; \mathbf{0}\big].
\label{eq:B24}
\end{equation}

Using \cref{eq:B6}, the marginal covariance matrices for the observed item vectors in each group are
\begin{equation}
\Var(\mathbf{y}^{(\mathrm{Infant})}_{j,k})=
\boldsymbol{\Lambda}^{(\mathrm{Infant})}
\boldsymbol{\Psi}^{(2)}
\boldsymbol{\Lambda}^{(\mathrm{Infant})\top}
+
\sigma_\alpha^2\mathbf{1}\mathbf{1}^\top
+
\sigma_\varepsilon^2\mathbf{I},
\label{eq:B25}
\end{equation}
\begin{equation}
\Var(\mathbf{y}^{(\mathrm{Toddler})}_{j,k})=
\boldsymbol{\Lambda}^{(\mathrm{Toddler})}
\boldsymbol{\Psi}^{(2)}
\boldsymbol{\Lambda}^{(\mathrm{Toddler})\top}
+
\sigma_\alpha^2\mathbf{1}\mathbf{1}^\top
+
\sigma_\varepsilon^2\mathbf{I}.
\label{eq:B26}
\end{equation}

Substituting \cref{eq:B23,eq:B24} into \cref{eq:B25,eq:B26} yields the key identification fact. Because the infant loading matrix has no columns for $\boldsymbol{\eta}_T$, the infant-group covariance $\boldsymbol{\Lambda}^{(\mathrm{Infant})}\boldsymbol{\Psi}^{(2)}\boldsymbol{\Lambda}^{(\mathrm{Infant})\top}$ depends only on $\{\boldsymbol{\Psi}_{QQ},\boldsymbol{\Psi}_{QI},\Psi_{II}\}$; all terms involving $\boldsymbol{\Psi}_{QT},\boldsymbol{\Psi}_{TT},\boldsymbol{\Psi}_{TI}$ vanish. Likewise, because the toddler loading matrix has no column for $\eta_I$, the toddler-group covariance depends only on $\{\boldsymbol{\Psi}_{QQ},\boldsymbol{\Psi}_{QT},\boldsymbol{\Psi}_{TT}\}$; all terms involving $\boldsymbol{\Psi}_{QI},\Psi_{II},\boldsymbol{\Psi}_{TI}$ vanish.

\begin{proposition}[Nonidentification of the cross-CLASS block]
Under model \cref{eq:B2,eq:B3,eq:B4} and design \cref{eq:B23,eq:B24}, the observed-data likelihood depends on the cross-CLASS covariance block $\boldsymbol{\Psi}_{TI}$ only through the feasibility constraint $\boldsymbol{\Psi}^{(2)}\succ 0$. In particular, $\boldsymbol{\Psi}_{TI}$ does not enter the model-implied second moments of any observed classroom type and is therefore not point-identified from the observed data.
\end{proposition}

\begin{proof}[Proof sketch]
The observed-data likelihood is determined by the collection of marginal distributions of observed item vectors $\mathbf{y}^{(\mathrm{Infant})}_{j,k}$ and $\mathbf{y}^{(\mathrm{Toddler})}_{j,k}$. Under Gaussianity, these are fully characterized by their means and covariances. \Cref{eq:B25,eq:B26} show that the covariances are $\boldsymbol{\Lambda}^{(g)}\boldsymbol{\Psi}^{(2)}\boldsymbol{\Lambda}^{(g)\top}$ plus level-3 and residual terms. Inserting \cref{eq:B23} and \cref{eq:B24} makes $\boldsymbol{\Psi}_{TI}$ drop out algebraically because it is always premultiplied and postmultiplied by a zero block. Hence the likelihood is constant in $\boldsymbol{\Psi}_{TI}$ over all values that keep $\boldsymbol{\Psi}^{(2)}$ positive definite.
\end{proof}

This is the multilevel latent-variable analogue of the basic factor-model identification principle: a covariance between two latent variables cannot be learned when those latent variables are never jointly measured by observed indicators within any unit, absent additional restrictions. In this sense, $\boldsymbol{\Psi}_{TI}$ is at best \textit{set-identified}, with the identified set determined by positive-definiteness (\cref{subsubsec:identified-set}; cf.\ \citealt{manski2003partial}).

\subsubsection{What QCIT Identifies: The ``Anchor'' Logic}
\label{subsubsec:anchor-logic}

Although $\boldsymbol{\Psi}_{TI}$ is nonidentified, QCIT plays an essential identification role for the rest of the covariance structure. Because QCIT items are observed in \textit{both} infant and toddler classrooms, the QCIT block identifies $\boldsymbol{\Psi}_{QQ}$ in the full sample, and each age group identifies the cross-covariances between QCIT and its own CLASS block:

\begin{itemize}[nosep]
\item In infant classrooms, the joint observation of QCIT + CLASS-Infant identifies $\boldsymbol{\Psi}_{QI}$ (QCIT--Infant) as part of the covariance structure in \cref{eq:B25}.
\item In toddler classrooms, the joint observation of QCIT + CLASS-Toddler identifies $\boldsymbol{\Psi}_{QT}$ (QCIT--Toddler) as part of the covariance structure in \cref{eq:B26}.
\end{itemize}

Thus QCIT provides an \textit{empirical overlap} that links the infant and toddler classrooms in a single joint model and makes QCIT--CLASS associations estimable despite the disjoint CLASS blocks. Without QCIT, the measurement model would decompose into two disconnected submodels (one for infant classrooms with CLASS-Infant, one for toddler classrooms with CLASS-Toddler), and no cross-instrument covariance statements would be available within a single likelihood.

\subsubsection{The Identified Set for $\boldsymbol{\Psi}_{TI}$ under Positive-Definiteness}
\label{subsubsec:identified-set}

Even though $\boldsymbol{\Psi}_{TI}$ does not enter the likelihood, it cannot take arbitrary values because the full covariance matrix must remain positive definite. Conditioning on the identified blocks $\boldsymbol{\Psi}_{QQ}$, $\boldsymbol{\Psi}_{QT}$, $\boldsymbol{\Psi}_{QI}$, $\boldsymbol{\Psi}_{TT}$, and $\Psi_{II}$, the feasible values of $\boldsymbol{\Psi}_{TI}$ are characterized by a Schur-complement constraint.

Let $\mathbf{A}=\boldsymbol{\Psi}_{QQ}$, $\mathbf{B}=\boldsymbol{\Psi}_{QT}$, $\mathbf{C}=\boldsymbol{\Psi}_{QI}$, $\mathbf{D}=\boldsymbol{\Psi}_{TT}$, $e=\Psi_{II}$, and let $\mathbf{f}=\boldsymbol{\Psi}_{TI}\in\mathbb{R}^{2\times 1}$. Assume $\mathbf{A}\succ 0$. Define the conditional (Schur-complement) blocks
\begin{equation}
\mathbf{M}=\mathbf{D}-\mathbf{B}^\top\mathbf{A}^{-1}\mathbf{B},
\qquad
v=e-\mathbf{C}^\top\mathbf{A}^{-1}\mathbf{C}.
\label{eq:B27}
\end{equation}
Then $\boldsymbol{\Psi}^{(2)}\succ 0$ holds if and only if $\mathbf{M}\succ 0$, $v>0$, and
\begin{equation}
(\mathbf{f}-\mathbf{B}^\top\mathbf{A}^{-1}\mathbf{C})^\top \mathbf{M}^{-1}(\mathbf{f}-\mathbf{B}^\top\mathbf{A}^{-1}\mathbf{C}) \;<\; v.
\label{eq:B28}
\end{equation}
\Cref{eq:B28} describes an ellipsoid in $\mathbb{R}^2$: the likelihood does not pick out a unique $\mathbf{f}$, but positive definiteness restricts $\mathbf{f}$ to a bounded region. Any point estimate of $\boldsymbol{\Psi}_{TI}$ produced by an ML optimizer in this setting should therefore be understood as one feasible completion of $\boldsymbol{\Psi}^{(2)}$ rather than a data-identified parameter.

A natural identifying restriction (not imposed in the paper) would be to assume conditional independence of infant and toddler CLASS factors given QCIT, which in the Gaussian case corresponds to setting the conditional covariance $\Cov(\boldsymbol{\eta}_{T},\eta_I \mid \boldsymbol{\eta}_Q)=0$. This restriction would fix $\mathbf{f}=\mathbf{B}^\top\mathbf{A}^{-1}\mathbf{C}$ (the center of the ellipsoid in \cref{eq:B28}) and thereby fully identify the cross-CLASS covariances. The paper instead reports $\boldsymbol{\Psi}_{TI}$ only as a supplementary, design-weakly-determined quantity (main text Table~2 note) and does not base substantive conclusions on these elements.

\subsubsection{Reporting and Interpretation of Cross-CLASS Correlations}
\label{subsubsec:reporting-cross-class}

Because $\boldsymbol{\Psi}_{TI}$ is not point-identified, cross-CLASS correlations (CLASS-Infant vs.\ CLASS-Toddler) should be treated as supplementary and reported, if at all, with explicit identification caveats. By contrast, correlations within QCIT, within CLASS-Toddler, and between QCIT and each age-appropriate CLASS block are directly supported by overlapping measurement blocks and admit the usual interpretation. The QCIT ``anchor'' logic should therefore be understood narrowly: QCIT anchors $\boldsymbol{\Psi}_{QQ}$, $\boldsymbol{\Psi}_{QI}$, and $\boldsymbol{\Psi}_{QT}$, but it does not identify $\boldsymbol{\Psi}_{TI}$ without further restrictions.

\subsubsection{Implications for Factor-Score Prediction}
\label{subsubsec:factor-score-implications}

Because $\boldsymbol{\Psi}_{TI}$ is nonidentified, any posterior prediction of an age-specific CLASS factor in classrooms where that factor's item block is not administered is inherently assumption-sensitive. Even when indirect prediction is possible through correlations with QCIT factors, such predictions amount to a model-imposed extrapolation across age-specific domains. Accordingly, all factor-score uses in the paper restrict attention to factors that are directly measured in the relevant classrooms (notably the QCIT factors, which are observed for all classrooms).

\section{ICC-like Variance Decomposition Under the Fitted GALAMM}
\label{app:icc-decomposition}

This appendix provides the mathematical derivation underlying the item-level ICC-like variance decomposition reported for RQ1-A and displayed in Figure~1 of the main manuscript (``see Appendix C''). The purpose is not to provide additional results, figures, or tables (Figure~1 already reports the decomposition), but to make explicit the \textit{model-implied variance algebra} that maps the fitted three-level GALAMM in \cref{app:galamm} to the Level-1 / Level-2 / Level-3 variance shares shown in Figure~1.

Throughout, notation follows \cref{app:galamm}. We use the term ``ICC-like'' because the reported quantities are best viewed as \textit{variance partition coefficients} (VPCs) derived from a multilevel latent-variable measurement model; they generalize (but are not identical to) the classical ICC from a single random-intercept model.

\subsection{Model Recap and Notation}
\label{subsec:C1-model-recap}

Items are indexed by $i=1,\dots,I$ (here $I=25$), classrooms by $j$, and centers by $k$. The measurement model is fit to standardized item scores (\cref{app:galamm}), which we denote $y_{ijk}$. The three-level Gaussian measurement model (\cref{subsubsec:threelevel-model}) is
\begin{equation}
y_{ijk}\mid \boldsymbol{\eta}_j,\alpha_k
\sim \mathcal{N}\!\left(
\beta_i + \sum_{f=1}^6 \lambda_{if}\eta_{fj} + \alpha_k,\;\sigma_\varepsilon^2
\right),
\label{eq:C1}
\end{equation}
where $\boldsymbol{\eta}_j=(\eta_{1j},\ldots,\eta_{6j})^\top$ are classroom-level latent factors with covariance matrix $\boldsymbol{\Psi}^{(2)}=\Var(\boldsymbol{\eta}_j)$, and $\alpha_k$ is a center-level random intercept with variance $\sigma_\alpha^2$:
\begin{equation}
\boldsymbol{\eta}_j \sim \mathcal{N}(\mathbf{0},\boldsymbol{\Psi}^{(2)}),\qquad
\alpha_k \sim \mathcal{N}(0,\sigma_\alpha^2).
\label{eq:C2}
\end{equation}
The Level-1 residuals $\varepsilon_{ijk}$ are independent of $\boldsymbol{\eta}_j$ and $\alpha_k$, and are homoscedastic across standardized items (a single $\sigma_\varepsilon^2$; \cref{app:galamm}).

Let $\boldsymbol{\Lambda}=[\lambda_{if}]\in\mathbb{R}^{I\times 6}$ be the loading matrix, and let $L_i^\top$ denote its $i$-th row so that $\sum_f \lambda_{if}\eta_{fj}=L_i^\top \boldsymbol{\eta}_j$.

\subsection{Model-Implied Marginal Covariance Decomposition}
\label{subsec:C2-marginal-cov}

For a fixed classroom $j$ in center $k$, define the $I$-vector $\mathbf{y}_{jk}=(y_{1jk},\ldots,y_{Ijk})^\top$, the intercept vector $\boldsymbol{\beta}=(\beta_1,\ldots,\beta_I)^\top$, and the all-ones vector $\mathbf{1}\in\mathbb{R}^I$. Then \cref{eq:C1} is equivalently
\begin{equation}
\mathbf{y}_{jk} = \boldsymbol{\beta} + \boldsymbol{\Lambda}\boldsymbol{\eta}_j + \mathbf{1}\alpha_k + \boldsymbol{\varepsilon}_{jk},
\qquad
\boldsymbol{\varepsilon}_{jk}\sim \mathcal{N}(\mathbf{0},\sigma_\varepsilon^2 \mathbf{I}_I).
\label{eq:C3}
\end{equation}
Taking marginal variances under \cref{eq:C2} yields the model-implied (unconditional) covariance matrix
\begin{equation}
\Var(\mathbf{y}_{jk})=
\underbrace{\boldsymbol{\Lambda}\boldsymbol{\Psi}^{(2)}\boldsymbol{\Lambda}^\top}_{\text{Level 2 (classroom factors)}}
+
\underbrace{\sigma_\alpha^2\,\mathbf{1}\mathbf{1}^\top}_{\text{Level 3 (center intercept)}}
+
\underbrace{\sigma_\varepsilon^2\,\mathbf{I}_I}_{\text{Level 1 (residual)}}.
\label{eq:C4}
\end{equation}
\Cref{eq:C4} is the ``one-line'' object behind Figure~1. The stacked bars in Figure~1 are constructed from the diagonal elements of the three addends in \cref{eq:C4}, normalized to sum to one item-by-item.

\subsection{Single-Item Marginal Variance Decomposition and ICC-like (VPC) Shares}
\label{subsec:C3-single-item}

\subsubsection{Marginal Variance for Item $i$}
\label{subsubsec:C3-1-marginal-var}

Applying the law of total variance to \cref{eq:C1} with $U=(\boldsymbol{\eta}_j,\alpha_k)$,
\begin{equation}
\Var(y_{ijk})=
\E\!\left[\Var(y_{ijk}\mid U)\right]
+
\Var\!\left(\E[y_{ijk}\mid U]\right).
\label{eq:C5}
\end{equation}
The conditional variance is $\sigma_\varepsilon^2$. The conditional mean is $\beta_i + L_i^\top\boldsymbol{\eta}_j + \alpha_k$. Using independence of $\boldsymbol{\eta}_j$ and $\alpha_k$,
\begin{equation}
\Var(y_{ijk})=
\underbrace{L_i^\top \boldsymbol{\Psi}^{(2)} L_i}_{V_{2i}}
+
\underbrace{\sigma_\alpha^2}_{V_3}
+
\underbrace{\sigma_\varepsilon^2}_{V_1}.
\label{eq:C6}
\end{equation}
Under the paper's simple-structure specification (\cref{subsubsec:simple-structure})---each item loads on a single factor $f(i)$---the Level-2 term simplifies to
\begin{equation}
V_{2i}=L_i^\top \boldsymbol{\Psi}^{(2)} L_i=\lambda_{if(i)}^2\,\psi_{f(i),f(i)}.
\label{eq:C7}
\end{equation}
(Under simple structure, the off-diagonals of $\boldsymbol{\Psi}^{(2)}$ affect cross-item covariances, not the marginal variance of a single indicator.)

\subsubsection{ICC-like (VPC) Shares for Item $i$}
\label{subsubsec:C3-2-vpc-shares}

Define the total marginal variance $V_{i,\text{tot}}=V_{2i}+V_3+V_1$. The ICC-like variance shares (i.e., VPCs) reported in Figure~1 are
\begin{equation}
\pi_2(i)=\frac{V_{2i}}{V_{i,\text{tot}}},\qquad
\pi_3(i)=\frac{V_3}{V_{i,\text{tot}}},\qquad
\pi_1(i)=\frac{V_1}{V_{i,\text{tot}}},
\label{eq:C8}
\end{equation}
with $\pi_1(i)+\pi_2(i)+\pi_3(i)=1$ for each item.

\subsubsection{Why ``ICC-like''? A Correlation Interpretation}
\label{subsubsec:C3-3-correlation-interp}

A classical ICC in a random-intercept model is both a variance fraction and a within-cluster correlation. The present model admits analogous correlation interpretations, but with two higher levels and an item-dependent Level-2 contribution.

Let two hypothetical replicates of the same item in the same classroom and center be $\{y^{(a)}_{ijk},y^{(b)}_{ijk}\}$, conditionally independent given $(\boldsymbol{\eta}_j,\alpha_k)$. Then $\Cov(y^{(a)}_{ijk},y^{(b)}_{ijk})=V_{2i}+V_3$, and $\Var(y^{(a)}_{ijk})=\Var(y^{(b)}_{ijk})=V_{2i}+V_3+V_1$, yielding
\begin{equation}
\Corr(y^{(a)}_{ijk},y^{(b)}_{ijk})=\frac{V_{2i}+V_3}{V_{2i}+V_3+V_1}=1-\pi_1(i).
\label{eq:C9}
\end{equation}
Thus $\pi_1(i)$ plays the role of ``one minus an ICC'' for item $i$ in the sense of the proportion of variance not reproducible across independent replicates under the measurement model.

The three-level structure further separates two distinct higher-level correlations. For two different classrooms $j\neq j'$ in the same center $k$, the shared component is only $\alpha_k$, so
\begin{equation}
\Corr(y_{ijk},y_{ij'k})=\frac{\sigma_\alpha^2}{V_{2i}+V_3+V_1}=\pi_3(i).
\label{eq:C10}
\end{equation}
In this sense, $\pi_3(i)$ is an ``ICC across classrooms within the same center,'' and $\pi_2(i)$ captures the incremental within-classroom clustering beyond the center shift. Unlike in random-intercept models, $\pi_2(i)$ is item-specific because $V_{2i}$ depends on the loading pattern $L_i$.

\subsection{``Overall'' Shares Reported in Figure~1 and Their Scope}
\label{subsec:C4-overall-shares}

Figure~1 reports item-specific triples $(\pi_1(i),\pi_2(i),\pi_3(i))$ and a single ``overall'' triplet in the figure note (e.g., 51.2\% / 14.5\% / 34.3\%). In the analysis code, these ``overall'' values summarize a typical item's variance composition under a chosen weighting of items---not the ICC of any particular composite score.

Let $w_i\ge 0$ be item weights with $\sum_i w_i=1$ (equal weights in the paper). Define the weighted average of the Level-2 marginal variances $\bar V_2=\sum_i w_i V_{2i}$. Because $V_3=\sigma_\alpha^2$ and $V_1=\sigma_\varepsilon^2$ are common across items in \cref{eq:C1}, their weighted averages are $\bar V_3=\sigma_\alpha^2$ and $\bar V_1=\sigma_\varepsilon^2$. The ``overall'' shares are therefore
\begin{align}
\bar\pi_2 &= \frac{\bar V_2}{\bar V_2+\sigma_\alpha^2+\sigma_\varepsilon^2}, \notag\\[0.5ex]
\bar\pi_3 &= \frac{\sigma_\alpha^2}{\bar V_2+\sigma_\alpha^2+\sigma_\varepsilon^2}, \label{eq:C11}\\[0.5ex]
\bar\pi_1 &= \frac{\sigma_\varepsilon^2}{\bar V_2+\sigma_\alpha^2+\sigma_\varepsilon^2}. \notag
\end{align}

\subsubsection{Distinction from a Composite-Score Variance Decomposition}
\label{subsubsec:C4-1-composite-distinction}

If one instead defines a composite score $\tilde y_{jk}=\mathbf{w}^\top \mathbf{y}_{jk}$ with $\mathbf{w}=(w_1,\ldots,w_I)^\top$, then by \cref{eq:C4},
\begin{align}
\Var(\tilde y_{jk}) &=
\underbrace{\mathbf{w}^\top \boldsymbol{\Lambda}\boldsymbol{\Psi}^{(2)}\boldsymbol{\Lambda}^\top \mathbf{w}}_{\text{Level 2}}
+
\underbrace{\sigma_\alpha^2(\mathbf{w}^\top\mathbf{1})^2}_{\text{Level 3}}
+
\underbrace{\sigma_\varepsilon^2\|\mathbf{w}\|_2^2}_{\text{Level 1}},
\label{eq:C12}
\end{align}
which generally differs from \cref{eq:C11} because the Level-2 term involves cross-item covariance. Figure~1 is explicitly an item-level decomposition; \cref{eq:C11} is therefore an appropriate descriptive summary of the item-level bars, not a claim about the reliability of a specific composite scoring rule.

\subsection{Identification, Scaling, and Invariance of the Decomposition}
\label{subsec:C5-identification}

\Cref{subsec:identification} describes the identification constraints used to scale latent factors (e.g., marker loadings vs variance constraints). While $\boldsymbol{\Lambda}$ and $\boldsymbol{\Psi}^{(2)}$ individually depend on this scaling choice, the marginal covariance contribution $\boldsymbol{\Lambda}\boldsymbol{\Psi}^{(2)}\boldsymbol{\Lambda}^\top$---and hence each $V_{2i}=L_i^\top\boldsymbol{\Psi}^{(2)}L_i$---is invariant to factor rescaling.

Formally, let $D$ be an invertible diagonal matrix, define a rescaled factor $\boldsymbol{\eta}_j^\star=D\boldsymbol{\eta}_j$, and set $\boldsymbol{\Lambda}^\star=\boldsymbol{\Lambda}D^{-1}$ so that $\boldsymbol{\Lambda}^\star\boldsymbol{\eta}_j^\star=\boldsymbol{\Lambda}\boldsymbol{\eta}_j$. Then $\boldsymbol{\Psi}^{(2)\star}=\Var(\boldsymbol{\eta}_j^\star)=D\boldsymbol{\Psi}^{(2)}D$, and
\begin{align}
\boldsymbol{\Lambda}^\star\boldsymbol{\Psi}^{(2)\star}\boldsymbol{\Lambda}^{\star\top}
&= \boldsymbol{\Lambda}D^{-1}\,(D\boldsymbol{\Psi}^{(2)}D)\,D^{-1}\boldsymbol{\Lambda}^\top \notag\\
&= \boldsymbol{\Lambda}\boldsymbol{\Psi}^{(2)}\boldsymbol{\Lambda}^\top.
\label{eq:C13}
\end{align}
Therefore, the variance shares $\pi_\ell(i)$ in \cref{eq:C8} are uniquely determined by the fitted marginal covariance structure, not by arbitrary factor scaling conventions.

\subsection{Practical Computation from Fitted GALAMM Parameters}
\label{subsec:C6-practical-computation}

Given the fitted GALAMM measurement model (\cref{app:galamm}), the Figure~1 decomposition is computed directly from $\widehat{\boldsymbol{\Psi}}^{(2)}$, $\widehat{\sigma}_\alpha^2$, $\widehat{\sigma}_\varepsilon^2$, and $\widehat{\boldsymbol{\Lambda}}$. In the implementation used for this paper \citep[R/\texttt{galamm};][]{Sorensen2023}, these quantities correspond to:

\begin{enumerate}[nosep]
\item \textit{Level-2 covariance} $\widehat{\boldsymbol{\Psi}}^{(2)}$: the $6\times 6$ random-effects covariance matrix at the classroom level (e.g., \texttt{VarCorr(mod)[["class\_id"]]}).
\item \textit{Level-3 variance} $\widehat{\sigma}_\alpha^2$: the center random-intercept variance (e.g., \texttt{VarCorr(mod)[["center\_id"]][1,1]}).
\item \textit{Level-1 variance} $\widehat{\sigma}_\varepsilon^2$: the squared residual standard deviation (e.g., \texttt{attr(VarCorr(mod), "sc")\^{}2}).
\item \textit{Loadings} $\widehat{\boldsymbol{\Lambda}}$: the item-by-factor loadings matrix (e.g., \texttt{factor\_loadings(mod)}, dropping associated standard-error columns).
\end{enumerate}

Then, for each item $i$, $\widehat V_{2i}=\widehat L_i^\top \widehat{\boldsymbol{\Psi}}^{(2)} \widehat L_i$, and the variance shares follow by normalizing as in \cref{eq:C8}. The ``overall'' shares in the Figure~1 note follow from \cref{eq:C11} with the chosen weights $w_i$ (equal weights in the paper).

\subsection{Interpreting a Large Level-1 Share}
\label{subsec:C7-interpreting-level1}

Finally, it is useful to be explicit about what ``Level 1'' represents here. In the Gaussian measurement model, $\sigma_\varepsilon^2$ is the conditional variance of $y_{ijk}$ given classroom and center effects. It therefore aggregates classical measurement error (rater noise, item idiosyncrasy) and any within-classroom deviations not captured by the prespecified factor structure. In this analysis, $\sigma_\varepsilon^2$ is common across standardized items (\cref{subsubsec:residual-structure}), so item-to-item differences in $\pi_1(i)$ are driven primarily by differences in the implied Level-2 signal $V_{2i}$ (and hence by loadings under simple structure). If one were to generalize the model to allow item-specific residual variances $\sigma_{\varepsilon,i}^2$, the decomposition in \cref{eq:C6,eq:C7,eq:C8} remains valid with $V_1=\sigma_{\varepsilon,i}^2$ varying by item; Figure~1 would then reflect heterogeneity in both loadings and residual variances.

\section{Empirical Bayes Predictions for Classroom-Level Process Quality (``Dose'')}
\label{app:eb-predictions}

This appendix provides a rigorous definition and interpretation of the empirical Bayes (EB) predictions referenced in the main manuscript (``see Appendix D''). In the paper, EB predictions of the classroom-level latent process-quality factors from the fitted GALAMM measurement model (\cref{app:galamm}) are used to operationalize each classroom's process-quality dose in the dose--response analyses (RQ2; \cref{app:dose-response}). The goal here is theoretical: to (i) express the EB predictor in standard mixed-model form and (ii) make precise---in a model-based sense---why the level-2 EB factor score is a classroom-specific latent quality signal that is \textit{net of} (a) Level-1 item noise and (b) center-wide shifts captured by the Level-3 random intercept.

Throughout, notation is aligned with \cref{app:galamm,app:icc-decomposition}.

\subsection{Dose Definition Used in RQ2}
\label{subsec:D0-dose-definition}

The fitted measurement model defines, for each classroom $j$, a six-dimensional vector of latent classroom factors
\begin{equation}
\boldsymbol{\eta}_j=(\eta_{1,j},\ldots,\eta_{6,j})^\top,
\label{eq:D0a}
\end{equation}
corresponding to the six domain factors described in \cref{subsec:item-inventory} (three QCIT domains plus age-appropriate CLASS domains). In RQ2, the ``dose'' is not a single scalar. Rather, for each domain $f\in\{1,\ldots,6\}$, the classroom's dose is the domain-specific EB factor score
\begin{equation}
\text{dose}_{f,j}
\;\equiv\;
\widehat{\eta}^{\,\mathrm{EB}}_{f,j}
\;=\;
\E\!\left[\eta_{f,j}\mid \mathbf{y},\widehat{\boldsymbol{\theta}}\right],
\label{eq:D1}
\end{equation}
where $\mathbf{y}$ denotes the full collection of standardized observed item scores across all classrooms and centers, and $\widehat{\boldsymbol{\theta}}$ denotes the fitted parameter estimates from the measurement model (\cref{app:galamm}). The adjective ``empirical'' reflects that $\boldsymbol{\theta}$ is plug-in estimated rather than integrated out.

For expositional simplicity the main text sometimes writes $\widehat{\eta}_j^{\mathrm{EB}}$ without the domain index; in that notation, $\widehat{\eta}_j^{\mathrm{EB}}$ should be read as ``the EB prediction for the focal domain,'' i.e., $\widehat{\eta}^{\,\mathrm{EB}}_{f,j}$ for the $f$ used in that particular dose--response regression.

\subsection{Measurement Model as a Gaussian Linear Mixed Model}
\label{subsec:D1-gaussian-lmm}

\subsubsection{Observation-Level Model}
\label{subsubsec:D1-1-obs-level}

Let $y_{i,j,k}$ be the standardized score for item $i\in\{1,\ldots,I\}$ in classroom $j$ nested within center $k$. The three-level Gaussian measurement model (\cref{app:galamm}) is
\begin{equation}
y_{i,j,k}\mid \boldsymbol{\eta}_j,\alpha_k
\sim \mathcal{N}\!\left(
\beta_i + L_i^\top\boldsymbol{\eta}_j + \alpha_k,\;\sigma_\varepsilon^2
\right),
\label{eq:D2}
\end{equation}
where $\beta_i$ are item intercepts; $L_i^\top$ is the $i$-th row of the loading matrix $\boldsymbol{\Lambda}\in\R^{I\times 6}$; $\boldsymbol{\eta}_j$ is the classroom-level factor vector; $\alpha_k$ is a center-level random intercept capturing center-wide shifts common to classrooms in center $k$; and $\sigma_\varepsilon^2$ is the (homoscedastic) Level-1 residual variance on the standardized scale.

At the higher levels,
\begin{equation}
\boldsymbol{\eta}_j \sim \mathcal{N}(\mathbf{0},\boldsymbol{\Psi}^{(2)}),
\qquad
\alpha_k \sim \mathcal{N}(0,\sigma_\alpha^2),
\qquad
\boldsymbol{\eta}_j\perp\alpha_k\perp\varepsilon_{i,j,k},
\label{eq:D3}
\end{equation}
consistent with \cref{subsubsec:threelevel-model}.

\subsubsection{Classroom Blocks and Center-Stacked Notation}
\label{subsubsec:D1-2-classroom-blocks}

Because planned missingness implies that not every classroom is observed on every item (\cref{app:galamm}), write the model for the observed item block for each classroom. Let $\mathbf{y}_{j,k}\in\R^{n_{j,k}}$ be the vector of observed item scores for classroom $j$ in center $k$, where $n_{j,k}$ is the number of observed items for that classroom. Let $\boldsymbol{\beta}_{j,k}\in\R^{n_{j,k}}$ be the corresponding intercept vector and let $\boldsymbol{\Lambda}_{j,k}\in\R^{n_{j,k}\times 6}$ be the corresponding submatrix of $\boldsymbol{\Lambda}$. Then \cref{eq:D2} is equivalently
\begin{equation}
\mathbf{y}_{j,k}=
\boldsymbol{\beta}_{j,k}
+
\boldsymbol{\Lambda}_{j,k}\boldsymbol{\eta}_j
+
\mathbf{1}_{n_{j,k}}\alpha_k
+
\boldsymbol{\varepsilon}_{j,k},
\qquad
\boldsymbol{\varepsilon}_{j,k}\sim \mathcal{N}(\mathbf{0},\sigma_\varepsilon^2\mathbf{I}_{n_{j,k}}).
\label{eq:D4}
\end{equation}

Now stack all classrooms within center $k$. Let $J_k$ be the number of classrooms in center $k$ and define $N_k=\sum_{j=1}^{J_k} n_{j,k}$. Let
\begin{equation}
\mathbf{y}_k=(\mathbf{y}_{1,k}^\top,\ldots,\mathbf{y}_{J_k,k}^\top)^\top\in\R^{N_k}
\label{eq:D4a}
\end{equation}
and define the center-specific random-effects vector
\begin{equation}
\mathbf{b}_k=\big(\boldsymbol{\eta}_{1}^\top,\ldots,\boldsymbol{\eta}_{J_k}^\top,\alpha_k\big)^\top\in\R^{6J_k+1}.
\label{eq:D5}
\end{equation}
Then $\mathbf{y}_k$ can be written in standard Gaussian linear mixed model (LMM) form
\begin{equation}
\mathbf{y}_k = \mathbf{X}_k\boldsymbol{\beta} + \mathbf{Z}_k\mathbf{b}_k + \boldsymbol{\varepsilon}_k,
\qquad
\boldsymbol{\varepsilon}_k\sim \mathcal{N}(\mathbf{0},\mathbf{R}_k),
\qquad
\mathbf{R}_k=\sigma_\varepsilon^2\mathbf{I}_{N_k},
\label{eq:D6}
\end{equation}
where $\mathbf{X}_k$ is the item-intercept design matrix and $\mathbf{Z}_k$ collects (i) the classroom-level factor-loading blocks and (ii) the center-intercept column of ones.

Under \cref{eq:D3}, $\mathbf{b}_k$ has a block-diagonal prior covariance
\begin{equation}
\mathbf{b}_k \sim \mathcal{N}\!\left(\mathbf{0},\mathbf{G}_k\right),
\qquad
\mathbf{G}_k=\operatorname{blockdiag}\!\left(\mathbf{I}_{J_k}\otimes \boldsymbol{\Psi}^{(2)},\;\sigma_\alpha^2\right).
\label{eq:D7}
\end{equation}

\subsection{EB Predictions as Posterior Means (BLUPs) of Random Effects}
\label{subsec:D2-eb-blup}

\subsubsection{Conditional Distribution Under Known Parameters}
\label{subsubsec:D2-1-conditional}

Fix a center $k$ and suppose the model parameters
\begin{equation}
\boldsymbol{\theta}=(\boldsymbol{\beta},\boldsymbol{\Psi}^{(2)},\sigma_\alpha^2,\sigma_\varepsilon^2,\boldsymbol{\Lambda})
\label{eq:D7a}
\end{equation}
are known. Under the Gaussian LMM (\cref{eq:D6,eq:D7}), the marginal distribution of $\mathbf{y}_k$ is
\begin{equation}
\mathbf{y}_k \sim \mathcal{N}\!\left(\mathbf{X}_k\boldsymbol{\beta},\;\mathbf{V}_k\right),
\qquad
\mathbf{V}_k=\mathbf{Z}_k\mathbf{G}_k\mathbf{Z}_k^\top+\mathbf{R}_k.
\label{eq:D8}
\end{equation}
Because $(\mathbf{b}_k,\mathbf{y}_k)$ is jointly Gaussian, the conditional distribution of the random effects is exactly multivariate normal:
\begin{align}
\mathbf{b}_k\mid \mathbf{y}_k,\boldsymbol{\theta}
&\sim
\mathcal{N}\Bigg(
\underbrace{\mathbf{G}_k\mathbf{Z}_k^\top\mathbf{V}_k^{-1}\big(\mathbf{y}_k-\mathbf{X}_k\boldsymbol{\beta}\big)}_{\text{posterior mean}},\;
\underbrace{\mathbf{G}_k-\mathbf{G}_k\mathbf{Z}_k^\top\mathbf{V}_k^{-1}\mathbf{Z}_k\mathbf{G}_k}_{\text{posterior variance}}
\Bigg).
\label{eq:D9}
\end{align}
The conditional mean in \cref{eq:D9} is the posterior mean under the Gaussian random-effects prior and, equivalently, the best linear unbiased predictor (BLUP) under known variance components \citep{henderson1950estimation,robinson1991blup}. Since the posterior is normal, the posterior mean equals the posterior mode.

\subsubsection{Empirical Bayes Implementation (Plug-in Hyperparameters)}
\label{subsubsec:D2-2-eb-implementation}

In the fitted measurement model, $\boldsymbol{\theta}$ is estimated from the full dataset (\cref{app:galamm}). The empirical Bayes prediction replaces $\boldsymbol{\theta}$ by its fitted value $\widehat{\boldsymbol{\theta}}$:
\begin{equation}
\widehat{\mathbf{b}}_k^{\,\mathrm{EB}}
\;\equiv\;
\E\!\left[\mathbf{b}_k\mid \mathbf{y},\widehat{\boldsymbol{\theta}}\right]=
\widehat{\mathbf{G}}_k\mathbf{Z}_k^\top\widehat{\mathbf{V}}_k^{-1}\big(\mathbf{y}_k-\mathbf{X}_k\widehat{\boldsymbol{\beta}}\big).
\label{eq:D10}
\end{equation}

Let $\mathbf{S}_{j}$ be the selection matrix that extracts the $6$-vector corresponding to classroom $j$ from $\mathbf{b}_k$, and let $\mathbf{s}_\alpha^\top$ extract the last coordinate (the center intercept). Then the EB predictions used in the paper are
\begin{equation}
\widehat{\boldsymbol{\eta}}_{j}^{\,\mathrm{EB}}=\mathbf{S}_{j}\widehat{\mathbf{b}}_k^{\,\mathrm{EB}},
\qquad
\widehat{\alpha}_k^{\,\mathrm{EB}}=\mathbf{s}_\alpha^\top \widehat{\mathbf{b}}_k^{\,\mathrm{EB}}.
\label{eq:D11}
\end{equation}
The ``dose'' variables in RQ2 are the components $\widehat{\eta}^{\,\mathrm{EB}}_{f,j}$ of $\widehat{\boldsymbol{\eta}}_{j}^{\,\mathrm{EB}}$ (\cref{eq:D1}).

\subsection{In What Sense Does \texorpdfstring{$\widehat{\boldsymbol{\eta}}_{j}^{\,\mathrm{EB}}$}{the EB Factor Score} ``Control for'' Level 1 and Center Effects?}
\label{subsec:D3-control}

The main manuscript describes $\widehat{\eta}^{\,\mathrm{EB}}_{f,j}$ as a measurement-error-adjusted classroom quality signal that partitions out item-level noise and center-wide differences. In the Gaussian setting, this statement can be made precise in two complementary, rigorous ways.

First, as a conditional-expectation object, $\widehat{\boldsymbol{\eta}}_{j}^{\,\mathrm{EB}}$ is the posterior mean of $\boldsymbol{\eta}_j$ under a model that explicitly contains Level-1 residual noise and a Level-3 center intercept. Conditioning on $\mathbf{y}$ therefore updates $\boldsymbol{\eta}_j$ using only the information in $\mathbf{y}$ that is predictive of classroom effects under the full hierarchical covariance $\mathbf{V}_k$.

Second, as a regularized partial coefficient, $\widehat{\boldsymbol{\eta}}_{j}^{\,\mathrm{EB}}$ is the minimizer of a quadratic objective in which classroom effects and center effects are estimated jointly. This implies that center-shared components are attributed to $\alpha_k$ rather than to $\boldsymbol{\eta}_j$.

\subsubsection{Level 1: Accounting for Item-Level Residual Variance via Shrinkage/Weighting}
\label{subsubsec:D3-1-level1}

From \cref{eq:D9}, the posterior mean can be written in the regression form
\begin{equation}
\E\!\left[\mathbf{b}_k\mid \mathbf{y}_k,\boldsymbol{\theta}\right]=
\Cov(\mathbf{b}_k,\mathbf{y}_k)\,\Var(\mathbf{y}_k)^{-1}\,\big(\mathbf{y}_k-\E[\mathbf{y}_k]\big),
\label{eq:D12}
\end{equation}
with $\Var(\mathbf{y}_k)=\mathbf{V}_k$. The Level-1 residual variance $\sigma_\varepsilon^2$ enters as the additive term $\mathbf{R}_k=\sigma_\varepsilon^2\mathbf{I}$ in $\mathbf{V}_k$. Holding fixed the higher-level variance components and loadings, larger $\sigma_\varepsilon^2$ increases $\mathbf{V}_k$ and therefore reduces the magnitude of $\mathbf{V}_k^{-1}(\mathbf{y}_k-\mathbf{X}_k\boldsymbol{\beta})$, shrinking the posterior mean toward zero. In the extreme cases (with other components fixed),
\begin{align}
\sigma_\varepsilon^2\to 0 &\;\Rightarrow\; \E[\mathbf{b}_k\mid \mathbf{y}_k] \text{ approaches an (essentially) least-squares fit}, \notag \\
\sigma_\varepsilon^2\to\infty &\;\Rightarrow\; \E[\mathbf{b}_k\mid \mathbf{y}_k]\to \mathbf{0}.
\label{eq:D13}
\end{align}
Thus Level-1 noise is not ``removed by subtraction''; it is down-weighted by the model-implied precision.

A more local expression is obtained by conditioning on the center intercept. For a given classroom $j$ in center $k$, rewrite \cref{eq:D4} as
\begin{equation}
\mathbf{y}_{j,k}-\mathbf{1}_{n_{j,k}}\alpha_k=
\boldsymbol{\beta}_{j,k}+ \boldsymbol{\Lambda}_{j,k}\boldsymbol{\eta}_j+\boldsymbol{\varepsilon}_{j,k}.
\label{eq:D14}
\end{equation}
Given $\alpha_k$, this is a Bayesian multivariate regression of $\mathbf{y}_{j,k}-\mathbf{1}\alpha_k$ on $\boldsymbol{\eta}_j$ with prior $\boldsymbol{\eta}_j\sim\mathcal{N}(\mathbf{0},\boldsymbol{\Psi}^{(2)})$. Conjugate normal theory yields
\begin{align}
\boldsymbol{\eta}_j \mid \mathbf{y}_{j,k},\alpha_k,\boldsymbol{\theta}
&\sim
\mathcal{N}\Bigg(
\mathbf{B}_{j,k}\big(\mathbf{y}_{j,k}-\boldsymbol{\beta}_{j,k}-\mathbf{1}_{n_{j,k}}\alpha_k\big), \notag \\
&\qquad\quad\;
\left((\boldsymbol{\Psi}^{(2)})^{-1}+\tfrac{1}{\sigma_\varepsilon^2}\boldsymbol{\Lambda}_{j,k}^\top\boldsymbol{\Lambda}_{j,k}\right)^{-1}
\Bigg),
\label{eq:D15}
\end{align}
where
\begin{equation}
\mathbf{B}_{j,k}=
\boldsymbol{\Psi}^{(2)}\boldsymbol{\Lambda}_{j,k}^\top
\big(\boldsymbol{\Lambda}_{j,k}\boldsymbol{\Psi}^{(2)}\boldsymbol{\Lambda}_{j,k}^\top+\sigma_\varepsilon^2\mathbf{I}_{n_{j,k}}\big)^{-1}
\in\R^{6\times n_{j,k}}
\label{eq:D16}
\end{equation}
is the regression-factor-score matrix (equivalently derived via the Woodbury identity). \Cref{eq:D15} shows precisely how Level-1 noise enters both the posterior mean and posterior variance: when the effective information $\boldsymbol{\Lambda}_{j,k}^\top\boldsymbol{\Lambda}_{j,k}/\sigma_\varepsilon^2$ is small (few observed items $n_{j,k}$, weak loadings, or large residual variance), the posterior mean in \cref{eq:D15} shrinks strongly toward the prior mean $\mathbf{0}$.

\subsubsection{Level 3: Separating Center-Wide Shifts from Within-Center Classroom Deviations}
\label{subsubsec:D3-2-level3}

The ``center control'' property follows directly from the fact that $\alpha_k$ enters \cref{eq:D4} as a shared component across all classrooms in center $k$ and is predicted jointly with $\boldsymbol{\eta}_j$. A transparent identity uses the law of iterated expectations:
\begin{equation}
\E[\boldsymbol{\eta}_j\mid \mathbf{y},\boldsymbol{\theta}]=
\E\!\left[
\E[\boldsymbol{\eta}_j\mid \mathbf{y},\alpha_k,\boldsymbol{\theta}]
\;\middle|\;
\mathbf{y},\boldsymbol{\theta}
\right].
\label{eq:D17}
\end{equation}
Because the inner conditional mean in \cref{eq:D15} is affine in $\alpha_k$, and $\mathbf{B}_{j,k}$ does not depend on $\alpha_k$, the outer expectation replaces $\alpha_k$ by its posterior mean:
\begin{equation}
\E[\boldsymbol{\eta}_j\mid \mathbf{y},\boldsymbol{\theta}]=
\mathbf{B}_{j,k}\big(\mathbf{y}_{j,k}-\boldsymbol{\beta}_{j,k}-\mathbf{1}_{n_{j,k}}\E[\alpha_k\mid \mathbf{y},\boldsymbol{\theta}]\big).
\label{eq:D18}
\end{equation}
Under empirical Bayes, \cref{eq:D18} becomes the plug-in factor score
\begin{equation}
\widehat{\boldsymbol{\eta}}_{j}^{\,\mathrm{EB}}=
\widehat{\mathbf{B}}_{j,k}\big(\mathbf{y}_{j,k}-\widehat{\boldsymbol{\beta}}_{j,k}-\mathbf{1}_{n_{j,k}}\widehat{\alpha}_k^{\,\mathrm{EB}}\big).
\label{eq:D19}
\end{equation}

\Cref{eq:D19} is the rigorous sense in which the level-2 EB factor score is ``net of center effects'':
\begin{itemize}[nosep]
\item $\widehat{\alpha}_k^{\,\mathrm{EB}}=\E[\alpha_k\mid \mathbf{y},\widehat{\boldsymbol{\theta}}]$ is learned from all classrooms in center $k$ through the joint posterior (\cref{eq:D9,eq:D10,eq:D11}), and captures the center-wide shift common to classrooms within that center.
\item $\widehat{\boldsymbol{\eta}}_{j}^{\,\mathrm{EB}}$ is computed from within-center residualized item scores $\mathbf{y}_{j,k}-\widehat{\boldsymbol{\beta}}_{j,k}-\mathbf{1}\widehat{\alpha}_k^{\,\mathrm{EB}}$, i.e., after removing the estimated center-wide component.
\item $\widehat{\mathbf{B}}_{j,k}$ then performs Level-1-aware weighting/shrinkage of these residualized item scores into the latent factor space.
\end{itemize}

As a useful limiting case, if $\sigma_\alpha^2\to\infty$ (so that $\alpha_k$ behaves like an unpenalized fixed intercept for each center), the ``residualization'' in \cref{eq:D19} approaches exact within-center demeaning of the item residuals. For finite $\sigma_\alpha^2$, the center intercept itself is partially pooled toward 0, but the decomposition (center-shared component into $\alpha_k$, within-center deviations into $\boldsymbol{\eta}_j$) remains the operative model-based partition.

\subsubsection{Equivalent Characterization as a Regularized Partial Coefficient}
\label{subsubsec:D3-3-regularized}

The EB posterior mean is also the minimizer of a penalized least-squares objective. Fix $\boldsymbol{\theta}$ and define the center-specific residual vector $\mathbf{r}_k=\mathbf{y}_k-\mathbf{X}_k\boldsymbol{\beta}$. Up to additive constants, the negative log posterior of $\mathbf{b}_k$ given $\mathbf{y}_k$ is
\begin{equation}
Q_k(\mathbf{b}_k)=
\frac{1}{2}\big\|\mathbf{r}_k-\mathbf{Z}_k\mathbf{b}_k\big\|_{\mathbf{R}_k^{-1}}^2
+
\frac{1}{2}\mathbf{b}_k^\top \mathbf{G}_k^{-1}\mathbf{b}_k,
\qquad
\| \mathbf{v}\|_{\mathbf{R}_k^{-1}}^2\equiv \mathbf{v}^\top\mathbf{R}_k^{-1}\mathbf{v}.
\label{eq:D20}
\end{equation}
Because $Q_k$ is quadratic, the posterior mean equals the posterior mode and hence the unique minimizer:
\begin{equation}
\E[\mathbf{b}_k\mid \mathbf{y}_k,\boldsymbol{\theta}]=
\arg\min_{\mathbf{b}_k} Q_k(\mathbf{b}_k).
\label{eq:D21}
\end{equation}

Partition $\mathbf{Z}_k=[\mathbf{Z}_{\eta,k}\ \mathbf{Z}_{\alpha,k}]$ and $\mathbf{b}_k=(\boldsymbol{\eta}_{1:J_k}^\top,\alpha_k)^\top$, where $\mathbf{Z}_{\alpha,k}=\mathbf{1}_{N_k}$. Let $\mathbf{G}_{\eta,k}=\mathbf{I}_{J_k}\otimes\boldsymbol{\Psi}^{(2)}$ and $G_\alpha=\sigma_\alpha^2$. The first-order conditions for \cref{eq:D21} are the mixed-model equations
\begin{equation}
\left(
\mathbf{Z}_k^\top\mathbf{R}_k^{-1}\mathbf{Z}_k+\mathbf{G}_k^{-1}
\right)\widehat{\mathbf{b}}_k=
\mathbf{Z}_k^\top\mathbf{R}_k^{-1}\mathbf{r}_k.
\label{eq:D22}
\end{equation}
Solving by block elimination yields an explicit expression for the classroom effects:
\begin{equation}
\widehat{\boldsymbol{\eta}}_{1:J_k}=
\left(
\mathbf{Z}_{\eta,k}^\top\mathbf{R}_k^{-1}\mathbf{M}_{\alpha,k}\mathbf{Z}_{\eta,k}
+\mathbf{G}_{\eta,k}^{-1}
\right)^{-1}
\mathbf{Z}_{\eta,k}^\top\mathbf{R}_k^{-1}\mathbf{M}_{\alpha,k}\mathbf{r}_k,
\label{eq:D23}
\end{equation}
where
\begin{equation}
\mathbf{M}_{\alpha,k}=
\mathbf{I}_{N_k}-
\mathbf{Z}_{\alpha,k}\left(\mathbf{Z}_{\alpha,k}^\top\mathbf{R}_k^{-1}\mathbf{Z}_{\alpha,k}+G_\alpha^{-1}\right)^{-1}\mathbf{Z}_{\alpha,k}^\top\mathbf{R}_k^{-1}
\label{eq:D24}
\end{equation}
is a ridge-type residual-maker that partials out the center intercept. \Cref{eq:D23} makes the ``control'' logic explicit: the EB/BLUP estimate of classroom latent quality is computed from $\mathbf{r}_k$ after removing the component attributable to the center intercept (via $\mathbf{M}_{\alpha,k}$), and with additional shrinkage toward zero induced by $\mathbf{G}_{\eta,k}^{-1}$. Thus $\widehat{\boldsymbol{\eta}}_j^{\,\mathrm{EB}}$ is a regularized partial coefficient for the classroom-level random effects in a model that simultaneously includes the center-level random intercept.

\subsection{Intuition Box: Why Shrinkage Improves on Raw Composites}
\label{subsec:D4-intuition}

\begin{quote}
\textit{Scalar shrinkage and MSE in one line.}
Suppose a single latent classroom trait $\theta_j$ (mean 0, variance $\tau^2$) is observed through an unbiased but noisy classroom signal $\tilde\theta_j=\theta_j+u_j$, with $u_j\sim\mathcal{N}(0,s_j^2)$ independent of $\theta_j$. Then the posterior mean is
\begin{equation}
\E[\theta_j\mid \tilde\theta_j]=
\underbrace{\frac{\tau^2}{\tau^2+s_j^2}}_{\text{shrinkage factor }w_j\in(0,1)}
\tilde\theta_j,
\label{eq:D25}
\end{equation}
and its Bayes risk under squared-error loss is
\begin{equation}
\E\!\left[(\E[\theta_j\mid \tilde\theta_j]-\theta_j)^2\right]=
\frac{\tau^2 s_j^2}{\tau^2+s_j^2}
< s_j^2=
\E\!\left[(\tilde\theta_j-\theta_j)^2\right].
\label{eq:D26}
\end{equation}
Thus shrinkage trades a small amount of bias for a variance reduction that lowers MSE on average---precisely the motivation for EB/BLUP estimators in multilevel settings \citep{Walters2024}.
\end{quote}

The classroom factor-score setting is a multivariate generalization of \cref{eq:D25,eq:D26}. In \cref{eq:D19}, the ``raw signal'' is the vector of residualized item scores $\mathbf{y}_{j,k}-\widehat{\boldsymbol{\beta}}_{j,k}-\mathbf{1}\widehat{\alpha}_k^{\,\mathrm{EB}}$. The scalar shrinkage factor $w_j$ is replaced by the matrix $\widehat{\mathbf{B}}_{j,k}$, which depends on (i) the number of observed items $n_{j,k}$, (ii) the loading pattern $\boldsymbol{\Lambda}_{j,k}$, and (iii) the residual variance $\sigma_\varepsilon^2$. Classrooms with fewer observed items or weaker loadings carry less information about $\boldsymbol{\eta}_j$ and therefore exhibit stronger shrinkage toward the population mean (0 under the factor-mean constraint).

This is also why EB factor scores are generally preferable to raw domain averages when the estimand is latent classroom quality rather than an ad hoc composite. A raw composite implicitly assigns fixed weights, treats all item noise as signal, and (unless explicitly residualized) conflates within-center classroom differences with center-wide shifts. The EB factor score uses the fitted hierarchical measurement model to (i) weight items according to their estimated relation to the latent factor (via loadings), (ii) shrink noisy classroom-level estimates toward the population mean, and (iii) attribute center-shared variation to $\alpha_k$ rather than to classroom-specific $\boldsymbol{\eta}_j$.

\subsection{Posterior Variance and the ``Plug-in Dose'' Caveat in Downstream Regressions}
\label{subsec:D5-posterior-variance}

The EB dose $\widehat{\eta}^{\,\mathrm{EB}}_{f,j}$ is a posterior mean and therefore comes with a posterior (conditional) variance that quantifies remaining uncertainty about $\eta_{f,j}$ after observing that classroom's items and borrowing strength from the fitted model.

From \cref{eq:D9}, the conditional variance of the full random-effects vector is
\begin{equation}
\Var(\mathbf{b}_k\mid \mathbf{y}_k,\boldsymbol{\theta})=
\mathbf{G}_k-\mathbf{G}_k\mathbf{Z}_k^\top\mathbf{V}_k^{-1}\mathbf{Z}_k\mathbf{G}_k,
\label{eq:D27}
\end{equation}
and the conditional variance of the classroom factor vector is obtained by selection:
\begin{equation}
\Var(\boldsymbol{\eta}_j\mid \mathbf{y}_k,\boldsymbol{\theta})=
\mathbf{S}_{j}
\Var(\mathbf{b}_k\mid \mathbf{y}_k,\boldsymbol{\theta})
\mathbf{S}_{j}^\top.
\label{eq:D28}
\end{equation}
Under empirical Bayes, \cref{eq:D27,eq:D28} are evaluated at $\widehat{\boldsymbol{\theta}}$ to yield plug-in posterior variances (often called prediction-error variances in the mixed-model literature).

In the paper's downstream dose--response regressions (\cref{app:dose-response}), we treat $\widehat{\eta}^{\,\mathrm{EB}}_{f,j}$ as an observed exposure variable. This is common in applied work, but it effectively ignores the uncertainty in \cref{eq:D28}. In principle, one could propagate factor-score uncertainty by (i) fitting a joint Bayesian model of measurement and outcomes, (ii) drawing multiple imputations of $\boldsymbol{\eta}_j$ from its posterior and combining downstream estimates, or (iii) applying explicit measurement-error corrections. We do not pursue these extensions here; \cref{app:dose-response} reports inference conditional on the plug-in EB doses.

Finally, note that $\widehat{\eta}^{\,\mathrm{EB}}_{f,j}$ remains an error-prone proxy for $\eta_{f,j}$. To the extent the remaining factor-score error behaves approximately like classical measurement error in the second-stage regression, dose--response slopes are biased toward zero (attenuation). A practical implication is that using EB/BLUP factor scores typically reduces attenuation relative to raw composites, but it does not eliminate it \citep{Walters2024}.

\subsection{Terminology and QC Alignment}
\label{subsec:D6-terminology}

The paper uses ``EB predictions'' to refer to the quantities in \cref{eq:D1,eq:D11}. In the Gaussian LMM framework used here,
\begin{equation}
\text{EB prediction}=\text{posterior mean}=\text{BLUP}=\text{conditional mode}
\label{eq:D29}
\end{equation}
for the random effects, conditional on $\widehat{\boldsymbol{\theta}}$. Thus, any of these equivalent terms refers to the same numerical quantity in this appendix and in \cref{app:dose-response}.

\subsection{Summary}
\label{subsec:D7-summary}

The paper's classroom ``dose'' variables are the domain-specific empirical Bayes factor scores $\widehat{\eta}^{\,\mathrm{EB}}_{f,j}$, defined as posterior means of the level-2 latent factors under the fitted three-level Gaussian measurement model. In the mixed-model representation, these EB scores coincide with BLUPs (and conditional modes) and are computed jointly with the level-3 center intercept $\alpha_k$. Consequently, the EB factor score (i) attenuates the influence of Level-1 item noise through precision-weighting and shrinkage and (ii) attributes center-shared shifts to $\alpha_k$ rather than to classroom-specific $\boldsymbol{\eta}_j$, a fact that is explicit in the residualized form (\cref{eq:D19}) and in the partialing-out operator $\mathbf{M}_{\alpha,k}$ in \cref{eq:D24}. Posterior variances quantify remaining factor-score uncertainty, but the downstream dose--response regressions treat the plug-in EB scores as observed exposures (\cref{app:dose-response}).

\section{Covariate Balancing and Dose--Response Estimation}
\label{app:dose-response}

This appendix provides the formal identification and estimation details underlying the paper's dose--response analyses (RQ2). The main text summarizes the substantive findings and refers readers here for (i) the within-center estimand induced by center-mean centering, (ii) the construction and diagnostics of covariate-balancing weights for continuous treatments, and (iii) the weighted linear and weighted generalized additive model (GAM) estimators used to recover domain-specific dose--response functions.

Throughout, the unit of analysis is the classroom. Notation is aligned with \cref{app:galamm,app:eb-predictions} (measurement model and empirical Bayes doses) and \cref{app:sample} (analytic sample).

\subsection{Notation and Objects Carried from Appendices B--D}
\label{subsec:E0-notation}

Let $k\in\{1,\dots,K\}$ index centers and $j\in\{1,\dots,J_k\}$ index classrooms within center $k$. We write $c(j)=k$ for the center of classroom $j$. For each classroom, let
\begin{itemize}[nosep]
    \item $Y_{p,jk}$ denote the classroom-level average of child outcome $p$ (teacher- or parent-reported), constructed by aggregating child-level scores within classroom $j$ in center $k$.
    \item $\widehat\eta^{\,\mathrm{EB}}_{f,jk}$ denote the empirical Bayes (EB) factor score for latent classroom process quality in domain $f\in\{1,\dots,6\}$, defined and interpreted in \cref{app:eb-predictions}.
    \item $X_{jk}\in\mathbb{R}^{q}$ denote the vector of observed teacher/classroom covariates used for confounding adjustment (in our application $q=26$).
\end{itemize}
RQ2 estimates domain-specific dose--response relationships between $Y_{p,jk}$ and $\widehat\eta^{\,\mathrm{EB}}_{f,jk}$ after removing center-level heterogeneity and balancing observed covariates.

\subsection{Variables Used in RQ2}
\label{subsec:E1-variables}

\subsubsection{Outcomes and Center-Mean Centering}
\label{subsubsec:E1-1-outcomes}

Let $Y_{p,jk}$ be the classroom mean of outcome $p$ in classroom $j$ at center $k$. For teacher-reported outcomes, $p$ includes (as in the main text) the English CDI IRT score and the BITSEA competence and problem scores \citep{Fenson2000,BriggsGowan2006}. Parent-reported outcomes are used for robustness (\cref{app:robustness}) and are treated identically in the estimation pipeline described here.

A key feature of the design is that centers create clustering and center-level common shocks (resources, leadership, staffing practices, neighborhood selection, etc.) that can jointly influence both classroom process quality and child outcomes. To isolate comparisons within centers, we remove center means from classroom outcomes:
\begin{equation}
Z_{p,jk}\;=\;Y_{p,jk}-\overline Y_{p,\cdot k},
\qquad 
\overline Y_{p,\cdot k}:=\frac{1}{J_k}\sum_{j=1}^{J_k}Y_{p,jk}.
\label{eq:E1}
\end{equation}
In words, $Z_{p,jk}$ is the deviation of classroom $j$'s mean outcome from the mean outcome of other classrooms in the same center. In a linear fixed-effects regression, the corresponding ``within'' transformation amounts to regressing the demeaned outcome $Y_{p,jk}-\overline Y_{p,\cdot k}$ on the demeaned dose $D_{jk}-\overline D_{\cdot k}$ (equivalently, including a full set of center indicators in a regression of $Y_{p,jk}$ on $D_{jk}$). Outcome demeaning alone removes additive center-level shifts in $Y$; exact fixed-effects implementation additionally residualizes the dose, as detailed in \cref{subsubsec:E2-2-centering}.

For readability we suppress $p$ in the remainder and write $Z_{jk}$ for the centered classroom outcome.

\subsubsection{Dose Variables: EB Factor Scores as Continuous Treatments}
\label{subsubsec:E1-2-dose}

The paper's ``dose'' variables are the domain-specific empirical Bayes factor scores from the three-level measurement model (\cref{app:galamm,app:eb-predictions}). For each domain $f$,
\begin{equation}
D_{f,jk}
\;\equiv\;
\widehat\eta^{\,\mathrm{EB}}_{f,jk}
=\mathbb{E}\!\left[\eta_{f,jk}\mid \mathbf{y},\widehat{\boldsymbol\theta}\right],
\label{eq:E2}
\end{equation}
where $\eta_{f,jk}$ is the latent classroom process-quality factor and $\mathbf{y}$ is the full set of observed items. The six domains are: QCIT Social-Emotional, QCIT Cognitive, QCIT Language-Literacy, CLASS-T Emotional-Behavioral, CLASS-T Learning, and CLASS-I Responsive (\cref{subsec:item-inventory}). RQ2 fits separate dose--response models for each $(p,f)$ pair (Table~3 in the main text).

Two features of $D_{f,jk}$ matter for the causal analysis:
\begin{enumerate}[nosep]
    \item \textit{Measurement-error reduction.} $D_{f,jk}$ is a precision-weighted, partially pooled signal that down-weights Level-1 item noise (\cref{app:eb-predictions}). Relative to raw composites, this reduces regression dilution (attenuation) when the dose is used as a regressor (\cref{subsec:E3-shrinkage}).
    \item \textit{Center-wide shifts are modeled explicitly.} The measurement model estimates center-level random intercepts jointly with classroom factors. Conceptually, $D_{f,jk}$ represents classroom-specific deviations in latent process quality after accounting for center-wide shifts in item ratings (\cref{app:eb-predictions}). This strengthens the intended ``within-center'' interpretation when combined with outcome demeaning.
\end{enumerate}

\subsubsection{Observed Covariates Used for Balancing}
\label{subsubsec:E1-3-covariates}

Let $X_{jk}$ denote the $q=26$ observed teacher/classroom characteristics used to address classroom-level confounding in RQ2. The covariates are measured at the teacher/classroom level and span (i) teacher background and human capital (e.g., race/ethnicity, experience, credentials), (ii) professional supports (e.g., coaching and training), (iii) classroom structural features (e.g., child--adult ratio, enrollment), and (iv) teacher psychosocial measures (e.g., depressive symptoms, caregiving beliefs about adult roles in learning, and job satisfaction; \citealp{Eaton2004}). Categorical measures are represented via indicator expansions, so $q$ counts the full set of numerical regressors entering the balancing step.

The balancing step requires complete covariate data for $X_{jk}$ (\cref{app:sample}). Unless otherwise noted, all weighting and outcome models are estimated on the corresponding complete-case classroom sample for each outcome $p$.

\subsection{Target Estimand and the Role of Within-Center Centering}
\label{subsec:E2-estimand}

\subsubsection{Potential-Outcomes Estimand for a Continuous ``Dose''}
\label{subsubsec:E2-1-estimand}

Fix a domain $f$ and suppress the $f$ index for readability. Let $D_{jk}\in\mathbb{R}$ denote the realized dose in classroom $j$ at center $k$. For any dose level $d$ in the support of $D$, let $Y_{jk}(d)$ denote the potential outcome that would be observed for classroom $j$ in center $k$ if (counterfactually) the classroom's latent process quality were set to $d$. The population dose--response function (DRF) is
\begin{equation}
\mu(d) \;:=\; \mathbb{E}\!\left[\,Y_{jk}(d)\,\right].
\label{eq:E3}
\end{equation}
In the main text we focus on within-center comparisons to reduce sensitivity to between-center selection. An equivalent way to express the estimand is through a center-fixed-effects marginal structural model (MSM),
\begin{equation}
\mathbb{E}\!\left[Y_{jk}(d)\mid k\right] = \alpha_k + m(d),
\label{eq:E4}
\end{equation}
where $\alpha_k$ captures center-specific baseline differences and $m(\cdot)$ is the (common) causal DRF shape. If effect heterogeneity across centers is present (i.e., $m$ depends on $k$), the estimators below recover an average of center-specific DRFs weighted by the distribution of doses and the weighting scheme; we do not attempt to identify center-by-dose interactions in the current application.

\subsubsection{Center Demeaning as a Fixed-Effects Transformation}
\label{subsubsec:E2-2-centering}

Consider the working model
\begin{equation}
Y_{jk} = \alpha_k + m(D_{jk}) + \varepsilon_{jk},
\qquad 
\mathbb{E}[\varepsilon_{jk}\mid D_{jk},k]=0.
\label{eq:E5}
\end{equation}
Let $\overline Y_{\cdot k}$ and $\overline D_{\cdot k}$ denote within-center means. Subtracting within-center means yields the ``within'' transformation
\begin{equation}
Y_{jk}-\overline Y_{\cdot k}=
\bigl(m(D_{jk})-\overline m_{\cdot k}\bigr)
+
(\varepsilon_{jk}-\overline\varepsilon_{\cdot k}).
\label{eq:E6}
\end{equation}
If $m(d)=\beta_1 d$ is linear, then $\overline m_{\cdot k}=\beta_1\overline D_{\cdot k}$ and
\begin{equation}
Z_{jk}=
\beta_1\,(D_{jk}-\overline D_{\cdot k})
+(\varepsilon_{jk}-\overline\varepsilon_{\cdot k}),
\label{eq:E7}
\end{equation}
so the fixed-effects estimator of $\beta_1$ is identical to the OLS regression of demeaned outcomes on demeaned doses. In our application, the EB dose is estimated from a measurement model that includes a center-level random intercept, which helps separate center-wide rating shifts from classroom-specific factors (\cref{app:eb-predictions}). Nevertheless, the linear fixed-effects MSM is obtained by residualizing both the outcome and the dose within centers (or equivalently, by including a full set of center indicators in the outcome model). Accordingly, the analysis implements the within-center logic either by working with within-center residuals $Z_{jk}=Y_{jk}-\overline Y_{\cdot k}$ and $\widetilde D_{jk}=D_{jk}-\overline D_{\cdot k}$, or by including center fixed effects directly; both routes yield numerically identical estimates of $\beta_1$ in linear MSMs.

\subsubsection{What Is (and Is Not) Claimed Causally}
\label{subsubsec:E2-3-causal}

The causal interpretation of the estimated DRFs requires assumptions that cannot be verified from observed data alone. We therefore make two distinctions explicit.
\begin{enumerate}[nosep]
    \item \textit{Within-center causal interpretation (conditional).} Under the identification assumptions in \cref{subsec:E4-identification}---most importantly, within-center unconfoundedness given observed covariates---our weighted estimators recover the causal effect of increasing latent classroom process quality within a center, holding fixed the center's stable characteristics.
    \item \textit{No causal interpretation for between-center differences.} Any association between center-average process quality and center-average child outcomes is not interpreted causally here; it is precisely the kind of association most vulnerable to unobserved center-level confounding, which motivates the within-center design.
\end{enumerate}

\subsection{Empirical Bayes Doses, Measurement Error, and ``Shrinkage on the Right''}
\label{subsec:E3-shrinkage}

RQ2 treats $D_{jk}=\widehat\eta^{\,\mathrm{EB}}_{jk}$ as a continuous exposure. Because $D_{jk}$ is estimated from finite item information, it is an error-prone proxy for latent classroom quality $\eta_{jk}$. This section clarifies how EB shrinkage relates to classical attenuation bias and why we do not apply shrinkage to the outcome.

\subsubsection{Classical Attenuation with a Noisy Regressor}
\label{subsubsec:E3-1-attenuation}

Suppose the latent data-generating model for a centered outcome is
\begin{equation}
Z_{jk} = \beta_0 + \beta_1 \eta_{jk} + \varepsilon_{jk},
\qquad 
\mathbb{E}[\varepsilon_{jk}\mid \eta_{jk}]=0,
\label{eq:E8}
\end{equation}
but we observe a noisy proxy $\widehat\eta_{jk}=\eta_{jk}+u_{jk}$ with $\mathbb{E}[u_{jk}]=0$ and $\mathbb{E}[u_{jk}\eta_{jk}]=0$. Then the naive regression of $Z_{jk}$ on $\widehat\eta_{jk}$ yields attenuation:
\begin{equation}
\text{plim}\,\widehat\beta_1^{\,\text{naive}}=
\beta_1\cdot \frac{\Var(\eta_{jk})}{\Var(\eta_{jk})+\Var(u_{jk})}
\;\in\;(0,\beta_1),
\label{eq:E9}
\end{equation}
the standard errors-in-variables result \citep[e.g.,][]{wooldridge2010econometric,Walters2024}.

\subsubsection{EB Posterior Means as Regression-Calibration Regressors}
\label{subsubsec:E3-2-regression-calibration}

Under a Gaussian signal-plus-noise model, the oracle posterior mean of $\eta_{jk}$ given $\widehat\eta_{jk}$ is
\begin{equation}
\eta^*_{jk}
:=\mathbb{E}[\eta_{jk}\mid \widehat\eta_{jk}]=
\lambda\,\widehat\eta_{jk} + (1-\lambda)\mu_\eta,
\qquad 
\lambda=\frac{\sigma_\eta^2}{\sigma_\eta^2+s^2},
\label{eq:E10}
\end{equation}
where $\sigma_\eta^2=\Var(\eta_{jk})$ and $s^2=\Var(u_{jk})$. \citet{Walters2024} shows that regressing $Z_{jk}$ on the oracle posterior mean $\eta^*_{jk}$ recovers the same slope as regressing on the true latent regressor $\eta_{jk}$ in the linear model:
\begin{equation}
\frac{\Cov(Z_{jk},\eta^*_{jk})}{\Var(\eta^*_{jk})}=
\beta_1.
\label{eq:E11}
\end{equation}
This is the formal sense in which ``shrinkage on the right fixes attenuation'' in linear regressions: using the conditional expectation of the latent regressor given its noisy estimate implements an errors-in-variables correction.

In our application, $D_{jk}=\widehat\eta^{\,\mathrm{EB}}_{jk}$ is an empirical (plug-in) version of the posterior mean computed from the fitted multilevel measurement model (\cref{app:eb-predictions}). This EB/BLUP construction is a direct multilevel generalization of classic James--Stein and parametric empirical Bayes shrinkage ideas \citep{james1961estimation,morris1983parametric}. The multivariate mixed-model structure differs from the scalar toy model above, but the core logic carries over: BLUP/EB factor scores are conditional expectations under the fitted Gaussian model and therefore behave like regression-calibration regressors. Importantly, \cref{app:eb-predictions} emphasizes a conservative practical point: because $D_{jk}$ is still an estimated proxy (and we condition on $\widehat{\boldsymbol\theta}$), some attenuation may remain in finite samples, and we interpret dose--response slopes as conservative when residual measurement error is present \citep{Walters2024}.

\subsubsection{Why We Do Not Shrink the Outcome}
\label{subsubsec:E3-3-no-shrink-outcome}

Shrinking the dependent variable is not an analog of measurement-error correction. If $Z_{jk}$ is observed with noise that is independent of the regressor, classical measurement error in the outcome inflates variance but does not bias the slope. By contrast, replacing $Z_{jk}$ with an empirical Bayes posterior mean $Z^*_{jk}$ partially pools outcomes toward a prior mean, reducing variance at the cost of systematic distortion in cross-classroom variation. \citet{Walters2024} shows that this ``shrinkage on the left'' attenuates regression coefficients toward zero even in simple settings. Accordingly, RQ2 applies EB shrinkage only to the dose (through the measurement model) and uses the centered classroom outcome $Z_{jk}$ without further pooling.

\subsection{Identification Assumptions for Continuous Treatments and Balancing Weights}
\label{subsec:E4-identification}

\subsubsection{Identification Assumptions}
\label{subsubsec:E4-1-assumptions}

For each domain $f$, identification of the DRF $\mu_f(d)$ from observational data relies on standard assumptions adapted to a continuous treatment setting \citep{robins2000marginal,imbens2000role,hirano2004propensity,imai2004causal}. Let $D$ denote the dose and $X$ denote observed covariates.
\begin{enumerate}[nosep]
    \item \textit{Consistency.} If $D_{jk}=d$, then $Y_{jk}=Y_{jk}(d)$.
    \item \textit{No interference (SUTVA at the classroom level).} $Y_{jk}(d)$ depends only on classroom $j$'s own dose, not on doses in other classrooms.
    \item \textit{Within-center weak unconfoundedness.} For all $d$ in the dose support,
    \begin{equation}
    Y_{jk}(d)\;\perp\!\!\!\perp\; D_{jk}\;\bigm|\;X_{jk},k.
    \label{eq:E12}
    \end{equation}
    This is a selection-on-observables assumption within centers: after conditioning on the observed teacher/classroom covariates and center membership, remaining variation in $D_{jk}$ is as-if random.
    \item \textit{Overlap / positivity (within centers).} The conditional density $f_{D\mid X,k}(d\mid x,k)$ is bounded away from zero on the relevant support: within strata defined by $(X,k)$, each dose level of interest has positive probability.
\end{enumerate}

Assumption \eqref{eq:E12} is substantively demanding. Our design mitigates its plausibility burden by (i) eliminating center-wide confounding via centering/fixed effects and (ii) using a rich set of $q=26$ classroom covariates that cover teacher background, classroom structure, and teacher psychosocial characteristics.

\subsubsection{Two Families of Estimators for Continuous Treatments}
\label{subsubsec:E4-2-estimators}

Under the identification assumptions in \cref{subsubsec:E4-1-assumptions}, a common identification route uses the generalized propensity score (GPS), defined as the conditional density $r(d,x,k)=f_{D\mid X,k}(d\mid x,k)$ \citep{imbens2000role,hirano2004propensity}. One can then estimate the DRF using outcome modeling, stratification, matching, or weighting. For weighting, a stabilized inverse-probability weight takes the form
\begin{equation}
w^{\text{IPW}}_{jk}=
\frac{f_D(D_{jk})}{f_{D\mid X,k}(D_{jk}\mid X_{jk},k)},
\label{eq:E13}
\end{equation}
which reweights the sample to a pseudo-population in which (ideally) $D$ is independent of $X$ within centers \citep{robins2000marginal}. In practice, however, estimating $f_{D\mid X,k}$ well is difficult, especially with many covariates and non-Gaussian dose distributions, motivating methods that target balance directly.

We therefore view our approach through the lens of covariate balancing weights for continuous treatments. We compared two broad families of methods (\cref{tab:E1}; \cref{fig:E1,fig:E2}):
\begin{enumerate}[nosep]
    \item GPS-based weighting with flexible nuisance models, including generalized linear models \citep{robins2000marginal,naimi2014constructing}, generalized boosted models \citep{McCaffrey2004,zhu2015boosting}, covariate-balancing GPS \citep{fong2018cbgps}, and BART-based GPS approaches \citep{chipman2010bart,hill2011challenges}.
    \item Convex optimization weights that enforce balance by construction, including entropy balancing for continuous treatments \citep{Tubbicke2022,vegetabile2021nonparametric}, generic constrained optimization-based balancing \citep{greifer2020estimating}, and energy/independence balancing weights \citep{huling2023independence}.
\end{enumerate}

Empirically, the convex optimization methods delivered the strongest reduction in dose--covariate associations in our data (\cref{fig:E1,fig:E2}), and we selected entropy balancing as the primary weighting method for all reported RQ2 estimates.

\begin{table}[htbp]
\centering
\caption{Overview of candidate weighting methods for continuous treatments.}
\label{tab:E1}
\footnotesize
\begin{adjustbox}{max width=\textwidth}
\begin{tabular}{@{}>{\raggedright}p{2.2cm}>{\raggedright}p{2.8cm}>{\raggedright}p{5.5cm}>{\raggedright\arraybackslash}p{4cm}@{}}
\toprule
Method Family & Method & Key Principle & Core References \\
\midrule
GPS-based & GLM-based GPS weighting & Specify a parametric model for $f_{D\mid X,k}(d\mid x,k)$ and weight by a stabilized density ratio $f_D(d)/f_{D\mid X,k}(d\mid x,k)$. & \citet{robins2000marginal}; \citet{naimi2014constructing} \\
\addlinespace
GPS-based & GBM weighting & Estimate the GPS flexibly via generalized boosted models and form stabilized weights; tuning can prioritize balance over prediction. & \citet{McCaffrey2004}; \citet{zhu2015boosting} \\
\addlinespace
GPS-based & Covariate-balancing GPS (CBGPS) & Use moment conditions (GMM) to estimate the GPS while directly targeting covariate balance for continuous treatments. & \citet{fong2018cbgps} \\
\addlinespace
GPS-based & BART-based GPS weighting & Flexibly model the GPS using Bayesian additive regression trees and form stabilized weights. & \citet{chipman2010bart}; \citet{hill2011challenges} \\
\addlinespace
Convex optimization & Entropy balancing (EBCT) & Solve a convex entropy-minimization problem subject to balance constraints that render the dose orthogonal to covariates (and possibly higher-order dose terms). & \citet{hainmueller2012entropy}; \citet{Tubbicke2022}; \citet{vegetabile2021nonparametric} \\
\addlinespace
Convex optimization & General constrained optimization weights & Solve a constrained optimization problem under user-specified balance constraints (exact or approximate). & \citet{greifer2020estimating} \\
\addlinespace
Convex optimization & Energy/independence balancing & Choose weights that minimize a dependence measure between $(D,X)$ and a target distribution, encouraging near-independence under weighting. & \citet{huling2023independence} \\
\bottomrule
\end{tabular}
\end{adjustbox}
\vspace{0.5em}

\noindent\footnotesize\textit{Note.} GPS = generalized propensity score; GLM = generalized linear model; GBM = generalized boosted model; CBGPS = covariate-balancing generalized propensity score; BART = Bayesian additive regression trees; EBCT = entropy balancing for continuous treatments; KL = Kullback--Leibler. This table summarizes the method inventory used for the weighting-method comparison in \cref{fig:E1,fig:E2}.
\end{table}

\subsection{Entropy Balancing for Continuous Treatments: Construction and Interpretation}
\label{subsec:E5-ebct}

\subsubsection{Optimization Problem (EBCT)}
\label{subsubsec:E5-1-optimization}

Entropy balancing for continuous treatments (EBCT) extends \citeauthor{hainmueller2012entropy}'s (\citeyear{hainmueller2012entropy}) entropy balancing from binary exposures to continuous doses \citep{Tubbicke2022}. Fix a domain $f$ and let $D_{jk}$ be the corresponding dose.

Let $b(X_{jk})$ denote a (possibly expanded) covariate vector that can include main effects and, if desired, higher-order terms and interactions among covariates. Let $\widetilde b(X_{jk}) := b(X_{jk})-\overline{b(X)}$ be the centered covariate expansion, where $\overline{b(X)}$ is a chosen target (typically the sample mean under base weights). Let $\widetilde D_{jk}^r:=D_{jk}^r-\overline{D^r}$ denote centered polynomial terms of the dose for orders $r=1,\dots,p$. Define the EBCT balancing-function vector
\begin{equation}
g_{jk}(p)=
\Bigl[
\widetilde b(X_{jk})^\top,\;
\widetilde D_{jk},\dots,\widetilde D_{jk}^p,\;
\widetilde b(X_{jk})^\top\widetilde D_{jk},\dots,\widetilde b(X_{jk})^\top\widetilde D_{jk}^p
\Bigr]^\top.
\label{eq:E14}
\end{equation}

EBCT chooses positive weights $w_{jk}$ by solving the convex program
\begin{equation}
\begin{aligned}
\min_{\{w_{jk}\}} \quad 
& \sum_{j,k} w_{jk}\log\!\left(\frac{w_{jk}}{q_{jk}}\right)\\
\text{s.t.}\quad 
& \sum_{j,k} w_{jk}\,g_{jk}(p)=0,\\
& \sum_{j,k} w_{jk}=1,\qquad w_{jk}>0,
\end{aligned}
\label{eq:E15}
\end{equation}
where $q_{jk}$ are analyst-chosen base weights (typically uniform weights $q_{jk}\propto 1$, but they can incorporate design weights if desired). The constraints enforce (weighted) mean preservation for $b(X)$ and $D^r$ and enforce zero weighted correlation between covariates and dose-polynomial terms through the cross-moment constraints $\widetilde b(X)\widetilde D^r$ \citep{Tubbicke2022}.

\subsubsection{Exponential Tilting Form and Interpretation}
\label{subsubsec:E5-2-exponential-tilting}

Because \eqref{eq:E15} is a strictly convex entropy minimization problem with linear constraints, any feasible solution is unique. The Lagrangian yields an exponential-tilting representation:
\begin{equation}
w_{jk}=
\frac{q_{jk}\exp\!\bigl(\lambda^\top g_{jk}(p)\bigr)}
{\sum_{j',k'} q_{j'k'}\exp\!\bigl(\lambda^\top g_{j'k'}(p)\bigr)},
\label{eq:E16}
\end{equation}
for a vector of dual parameters $\lambda$ chosen so that the balance constraints hold exactly. This representation clarifies EBCT as the ``closest'' (in KL divergence) reweighting of the base weights that satisfies the imposed moment restrictions---equivalently, the maximum-entropy distribution subject to balance.

\begin{figure}[htbp]
\centering
\includegraphics[width=0.95\textwidth]{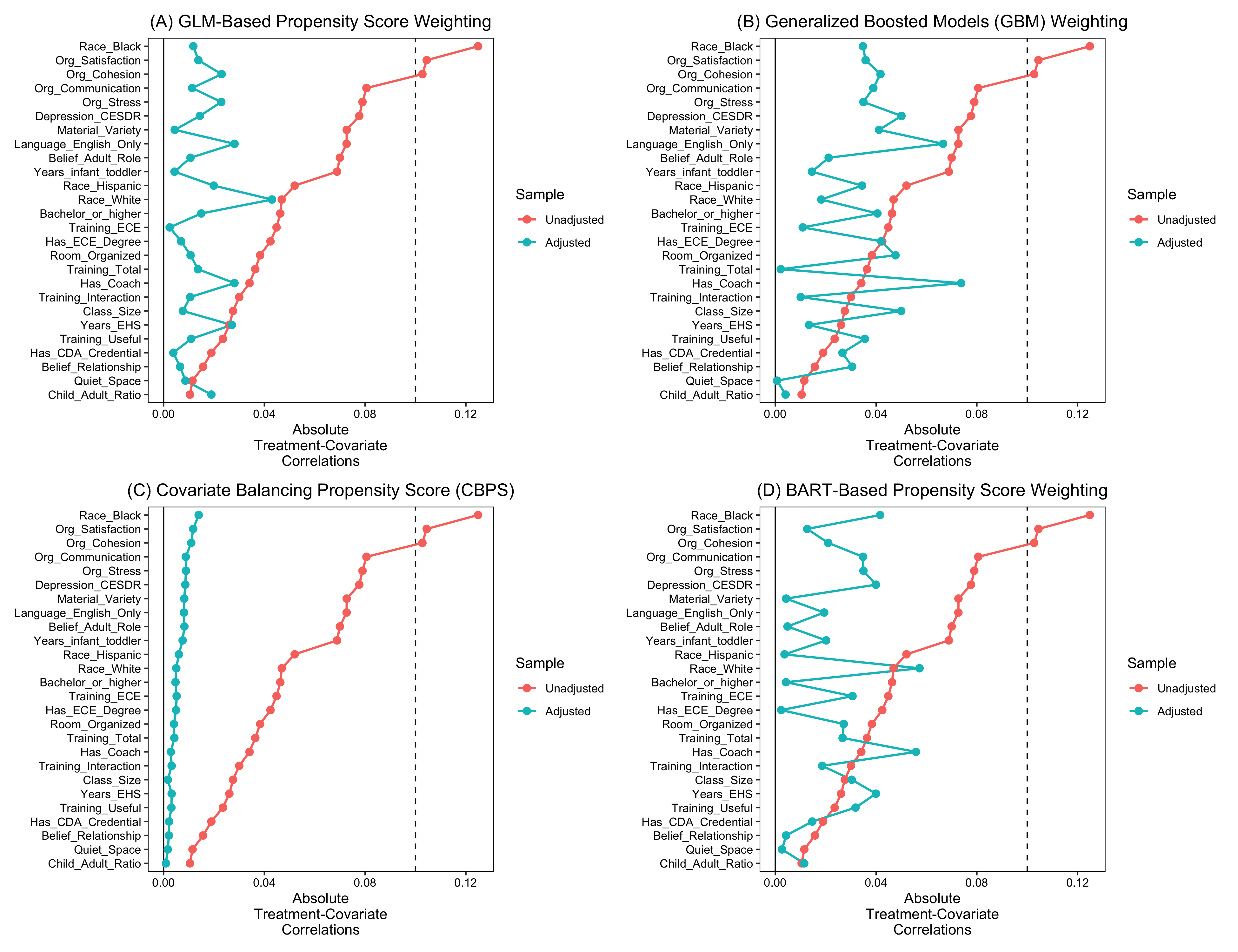}
\caption{Covariate balance before and after weighting for four propensity score-based methods.}
\label{fig:E1}
\vspace{0.5em}

\noindent\footnotesize\textit{Note.} Each panel displays the absolute weighted correlation between individual covariates (rows) and the continuous dose (QCIT Cognitive domain EB factor score) for one GPS-based weighting method. The left edge of each panel shows the unadjusted (unweighted) correlation; the right edge shows the weighted correlation after applying the method. Correlations closer to zero indicate better balance. The four methods shown are GLM-based GPS weighting, GBM weighting, covariate-balancing GPS (CBGPS), and BART-based GPS weighting. The same diagnostic workflow is applied to each domain-specific dose; the QCIT Cognitive dose is shown as an illustrative example.
\end{figure}

\subsubsection{Why Balancing Moments Can Identify Causal Parameters}
\label{subsubsec:E5-3-balance-causal}

EBCT is designed to make the dose orthogonal to observed covariates in a way tailored to continuous-treatment outcome models.

A particularly transparent link to causal identification arises in linear MSMs. Suppose the true conditional mean is additive in $X$,
\begin{equation}
\mathbb{E}[Z_{jk}\mid D_{jk},X_{jk},k]=
\beta_0 + \beta_1 D_{jk} + \beta_X^\top b(X_{jk}),
\label{eq:E17}
\end{equation}
and consider the weighted regression of $Z_{jk}$ on $D_{jk}$ alone. The usual omitted-variable-bias formula (in the weighted population) implies
\begin{equation}
\beta_1^{\,\text{(omit }X\text{)}}=
\beta_1
+
\beta_X^\top
\frac{\Cov_w\!\bigl(b(X_{jk}),D_{jk}\bigr)}
{\Var_w(D_{jk})}.
\label{eq:E18}
\end{equation}
If the weights satisfy $\Cov_w(b(X),D)=0$ (as enforced by EBCT when $p\ge 1$), then the bias term vanishes and the weighted regression of $Z$ on $D$ recovers $\beta_1$ even without including $X$ in the outcome model. More generally, if one fits an outcome model involving basis functions $\phi(D)$ (e.g., polynomials, spline bases), then balancing the corresponding cross-moments between $b(X)$ and $\phi(D)$ eliminates the leading bias term for estimators in that model class. This is the guiding principle of EBCT: it enforces a set of orthogonality conditions between covariates and dose functions rich enough to support flexible DRF estimation, without explicitly estimating the GPS density.

\subsubsection{Implementation in This Study}
\label{subsubsec:E5-4-implementation}

We estimate a separate EBCT weight vector $w^{(f)}$ for each domain-specific dose $D_{f,jk}$. Because weighting targets the treatment assignment mechanism, weights are domain-specific but not outcome-specific: for a fixed domain $f$, the same weights $w^{(f)}$ are applied to all teacher outcomes and to all parent outcomes (\cref{app:robustness}). This ensures methodological consistency across outcomes and avoids outcome-adaptive weighting.

Balancing diagnostics and weight-distribution diagnostics are evaluated for each domain; examples are shown in \cref{fig:E1,fig:E2} and Figure~2 of the main text.

\begin{figure}[htbp]
\centering
\includegraphics[width=0.95\textwidth]{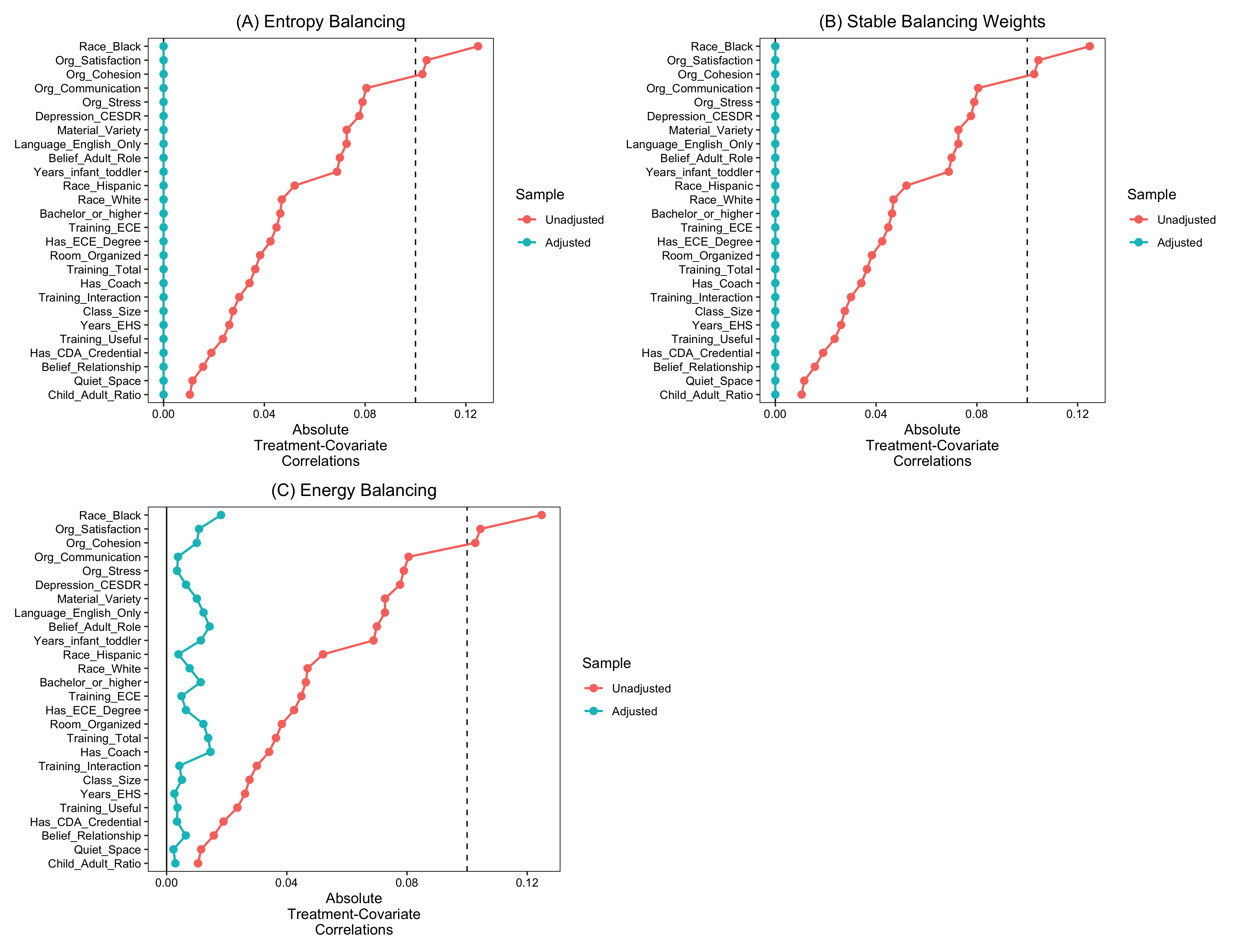}
\caption{Covariate balance before and after weighting for three convex optimization-based methods.}
\label{fig:E2}
\vspace{0.5em}

\noindent\footnotesize\textit{Note.} Each panel displays the absolute weighted correlation between individual covariates (rows) and the continuous dose (QCIT Cognitive domain EB factor score) for one convex optimization-based weighting method. The left edge of each panel shows the unadjusted (unweighted) correlation; the right edge shows the weighted correlation after applying the method. Correlations closer to zero indicate better balance. The three methods shown are entropy balancing for continuous treatments (EBCT), general constrained optimization weights, and energy/independence balancing. EBCT achieves near-zero correlations for essentially all covariates, motivating its selection as the primary weighting method for RQ2.
\end{figure}

\subsection{Diagnostics: Balance, Weight Behavior, and Effective Sample Size}
\label{subsec:E6-diagnostics}

Covariate balancing methods are only as credible as their diagnostics. We emphasize three checks.

\subsubsection{Dose--Covariate Balance}
\label{subsubsec:E6-1-balance}

For continuous treatments, a natural balance diagnostic is the (weighted) correlation between each covariate and the dose (or, more generally, between each element of $b(X)$ and each basis function of the dose). \Cref{fig:E1,fig:E2} plot absolute dose--covariate correlations before and after weighting across candidate methods, and Figure~2 in the main text provides an illustrative comparison (GBM vs EBCT) for one domain. In the selected EBCT specification, these correlations are pushed close to zero for essentially all covariates, indicating that the observed covariate distribution is nearly orthogonal to the dose in the weighted pseudo-population.

As recommended by \citet{Tubbicke2022}, an additional diagnostic (not shown) is to estimate pseudo-DRFs that replace the outcome with each covariate; if balancing is successful, the resulting pseudo-DRFs should be flat over the dose range.

\subsubsection{Weight Distribution and Practical Positivity}
\label{subsubsec:E6-2-weights}

Even when balance is excellent, extreme weights can produce unstable finite-sample estimates and indicate near-violations of overlap/positivity. We therefore inspect the weight distribution (minimum, median, maximum, dispersion) and assess whether a small number of classrooms receive disproportionate influence.

When extreme weights occur, one can apply trimming or truncation (e.g., capping weights at a high percentile) at the cost of changing the estimand to one defined on the trimmed support. In our application, EBCT delivered comparatively stable weights relative to GPS-based methods in the candidate set (\cref{fig:E1,fig:E2}), and we therefore report untrimmed EBCT results as primary.

\subsubsection{Effective Sample Size (ESS)}
\label{subsubsec:E6-3-ess}

A summary measure of weight concentration is the effective sample size:
\begin{equation}
\mathrm{ESS}=
\frac{\bigl(\sum_{j,k} w_{jk}\bigr)^2}{\sum_{j,k} w_{jk}^2}.
\label{eq:E19}
\end{equation}
With weights normalized to sum to 1, this reduces to $\mathrm{ESS}=1/\sum w_{jk}^2$. ESS is reported and compared across weighting methods during method selection; it quantifies how much information remains after reweighting and complements balance diagnostics.

\subsection{Outcome Models: Weighted Linear MSM and Weighted GAM DRFs}
\label{subsec:E7-outcome-models}

Fix a domain $f$ and suppress the domain index. Let $w_{jk}$ be the selected EBCT weight and let $Z_{jk}$ be the center-mean-centered classroom outcome.

\subsubsection{Weighted Linear Regression}
\label{subsubsec:E7-1-linear}

The linear MSM is
\begin{equation}
\mathbb{E}_w[Z_{jk}\mid D_{jk}]=
\gamma_0+\gamma_1 D_{jk},
\label{eq:E20}
\end{equation}
estimated by weighted least squares (WLS):
\begin{equation}
(\widehat\gamma_0,\widehat\gamma_1)=
\arg\min_{\gamma_0,\gamma_1}
\sum_{j,k} w_{jk}\,\bigl(Z_{jk}-\gamma_0-\gamma_1 D_{jk}\bigr)^2.
\label{eq:E21}
\end{equation}
Under successful balancing and the identification assumptions in \cref{subsec:E4-identification}, $\widehat\gamma_1$ estimates the within-center causal linear effect of increasing the latent process-quality dose by one unit (in the scale of the EB factor score).

\subsubsection{Weighted Generalized Additive Models (GAMs)}
\label{subsubsec:E7-2-gam}

To allow a flexible DRF, we estimate the GAM-MSM
\begin{equation}
\mathbb{E}_w[Z_{jk}\mid D_{jk}]=
\gamma_0 + f(D_{jk}),
\label{eq:E22}
\end{equation}
where $f(\cdot)$ is an unknown smooth function. In spline form $f(d)=\sum_{m=1}^M \theta_m B_m(d)$, estimation solves a penalized weighted least-squares problem,
\begin{align}
(\widehat\gamma_0,\widehat{\boldsymbol\theta})=
\arg\min_{\gamma_0,\boldsymbol\theta}
\Biggl\{
\sum_{j,k} w_{jk}\,\Bigl(Z_{jk}-\gamma_0-\sum_{m=1}^M \theta_m B_m(D_{jk})\Bigr)^2
\;+\;
\lambda\,\boldsymbol\theta^\top S\,\boldsymbol\theta
\Biggr\},
\label{eq:E23}
\end{align}
where $S$ is the spline penalty matrix. We implement $f(\cdot)$ using penalized regression splines (thin-plate regression spline basis in the default implementation of \texttt{mgcv}) and select the smoothing parameter $\lambda$ by restricted maximum likelihood (REML), following standard recommendations for stable smoothing-parameter estimation \citep{Wood2017}. The fitted curve $\widehat f(\cdot)$ is interpreted as the estimated causal DRF in the weighted pseudo-population.

A key diagnostic for nonlinearity is the effective degrees of freedom (edf) of the smooth term. In penalized spline GAMs, $\mathrm{edf}\approx 1$ indicates that the penalty has shrunk the smooth close to a linear function, whereas $\mathrm{edf}>1$ indicates curvature beyond a straight line. We report edf and the smooth-term significance test in Table~3 of the main text to summarize the degree of nonlinearity.

\subsubsection{Inference}
\label{subsubsec:E7-3-inference}

All inference is conditional on the estimated weights and on the plug-in EB doses. This is standard in two-stage MSM implementations. Because weights can induce heteroskedasticity and because residual dependence may remain within centers, we use robust variance estimation in the implementation (details in the analysis scripts). The GAM uncertainty bands are computed from the model's approximate covariance for the spline coefficients \citep{Wood2017}, yielding pointwise confidence intervals for $\widehat f(d)$.

\subsection{Sensitivity Analyses and Robustness}
\label{subsec:E8-sensitivity}

Two sensitivity checks are embedded in the RQ2 workflow.
\begin{enumerate}[nosep]
    \item \textit{Weighting-method comparison.} \Cref{fig:E1,fig:E2} compare balance across seven continuous-treatment weighting methods, spanning GPS-based and convex-optimization approaches (\cref{tab:E1}). The main text Figure~2 provides a focused comparison between GBM-based GPS weighting and EBCT. The qualitative conclusion---near-zero covariate--dose associations after EBCT and stable weights---motivates choosing EBCT as the primary method.
    \item \textit{Alternative outcomes (\cref{app:robustness}).} Parent-reported outcomes serve as a robustness check against potential bias from shared teacher-reporting methods. \Cref{app:robustness} applies the same domain-specific EBCT weights and the same linear/GAM MSM specifications described in this appendix, ensuring methodological comparability.
\end{enumerate}

\section{Supplemental Dose--Response Analyses: Teacher-Reported BITSEA Problem and Parent-Reported Outcomes}
\label{app:robustness}

This appendix consolidates two sets of supplemental results referenced in the main text but omitted for brevity: (i) dose--response curves for teacher-reported BITSEA Problem scores, which complement the main-text figures for CDI IRT and BITSEA Competence, and (ii) complete regression results for parent-reported child outcomes, which serve as a robustness check against potential common-source bias. Throughout, the analytic approach follows the specifications documented in \cref{app:eb-predictions,app:dose-response}, including center-mean centering, entropy balancing weights, and both linear and generalized additive model (GAM) specifications. Analytic sample sizes by outcome and reporter are detailed in \cref{app:sample} (\cref{tab:A2}).

\subsection{Dose--Response Curves for Teacher-Reported BITSEA Problem Scores}
\label{subsec:F1-bitsea-problem}

The main text reports that BITSEA Problem scores show no significant associations with any classroom process quality domain in either linear or GAM models (Table~3, all $p > 0.10$). \Cref{fig:F1} presents the dose--response curves underlying these null findings, allowing readers to visually verify that the relationships are essentially flat across domains.

\begin{figure}[htbp]
\centering
\includegraphics[width=0.95\textwidth]{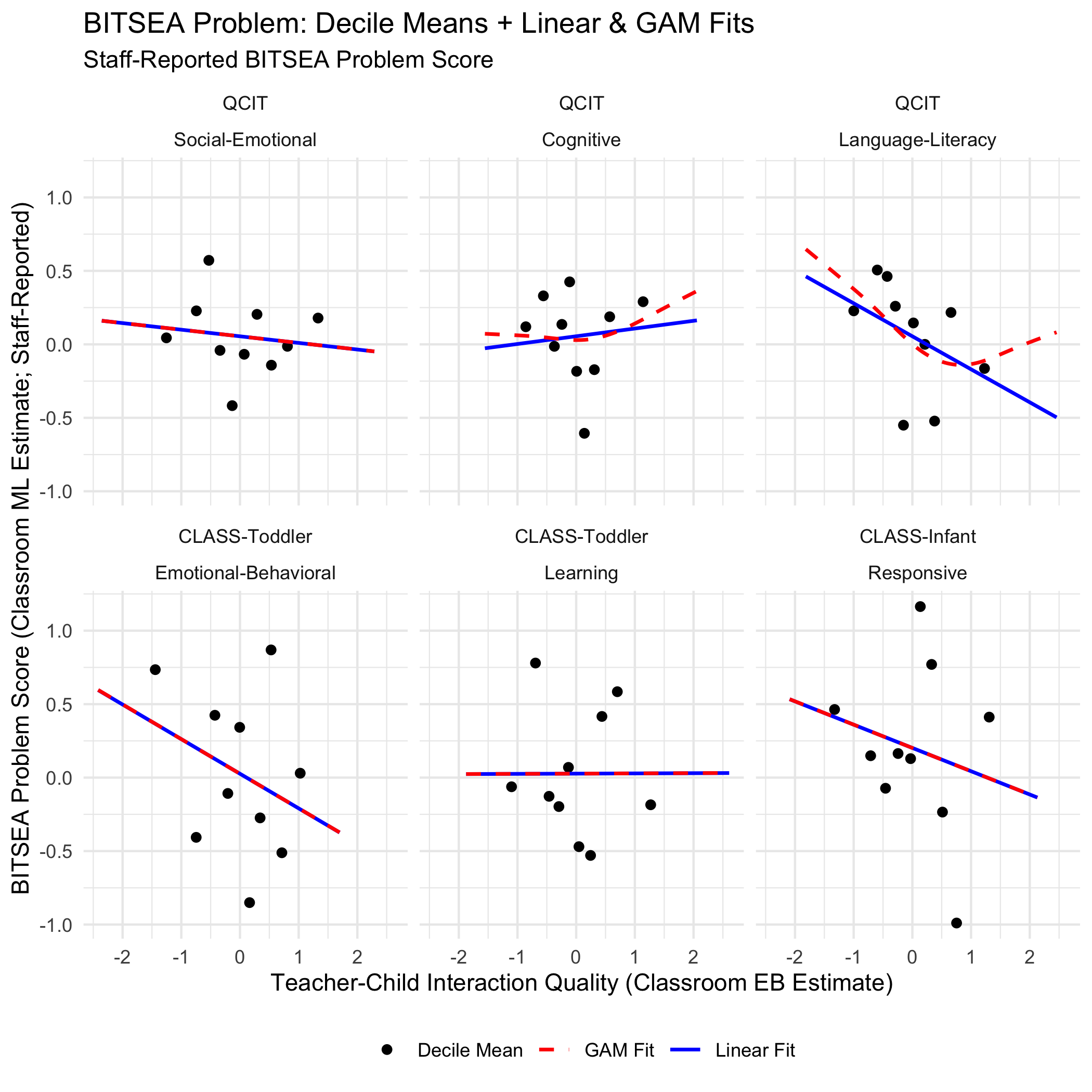}
\caption{Estimated dose--response curves for teacher-reported BITSEA Problem scores across QCIT and CLASS domains.}
\label{fig:F1}
\vspace{0.5em}

\noindent\footnotesize\textit{Note.} The $x$-axis represents teacher--child interaction quality measured as empirical Bayes (EB) estimates from the three-level GALAMM model. The $y$-axis represents center-mean-centered classroom outcomes. Black dots indicate decile means of the dose variable. Blue solid lines represent entropy-balanced weighted linear regression fits. Red dashed lines represent entropy-balanced weighted GAM fits.
\end{figure}

As shown in \cref{fig:F1}, the GAM curves (red dashed lines) closely track the linear fits (blue solid lines), with effective degrees of freedom (edf) near 1.0 for most domains, indicating minimal nonlinearity. Some domains exhibit slight negative slopes---particularly CLASS-T Emotional-Behavioral (coef $= -0.24$, $p = 0.17$) and QCIT Language-Literacy (coef $= -0.22$, $p = 0.22$)---suggesting a possible protective effect whereby higher quality is associated with fewer problem behaviors. However, these estimates remain statistically nonsignificant with wide confidence bands, precluding strong conclusions.

Several non-exclusive mechanisms may explain these null patterns. First, problem behaviors in infant--toddler samples can exhibit restricted range and low base rates, limiting power to detect modest associations at the classroom level. Second, after centering outcomes within centers, remaining variation in behavioral problems may reflect child- and family-level contexts only weakly coupled with observed classroom interaction processes. Third, as discussed in \cref{app:eb-predictions}, measurement error in the EB ``dose'' (factor-score uncertainty) can attenuate estimated dose--response slopes. Fourth, the theoretical framework (Table~1 in the main text) posits that emotional-behavioral support primarily promotes social-emotional competence; the pathway to \textit{reduced} problem behaviors may require a longer time horizon or involve indirect mechanisms that cross-sectional data cannot capture.

These considerations suggest that the absence of significant associations should not be interpreted as evidence that classroom process quality is unrelated to behavioral outcomes. Rather, the null findings may reflect measurement limitations, developmental timing, or the need for longitudinal designs tracing cumulative exposure effects.

\subsection{Robustness Check: Parent-Reported Child Outcomes}
\label{subsec:F2-parent-reported}

A potential threat to internal validity in the primary analyses is common-source bias: teachers both deliver classroom interactions (the ``dose'') and assess child outcomes, potentially inflating observed associations through shared-method variance. To probe this possibility, we replicated the dose--response analyses using parent-reported child outcomes, which are assessed independently by parents who observe children in home contexts.

\subsubsection{Analytic Sample for Parent-Reported Outcomes}
\label{subsubsec:F2-1-sample}

The parent-reported sample follows the same exclusion criteria as the teacher-reported sample (\cref{app:sample}), with the additional requirement that parents completed the relevant child assessment modules. Due to higher nonresponse rates for parent interviews, the analytic samples are smaller: 1,505 children (720 classrooms; 428 centers) for parent-reported CDI IRT, 1,696 children (751 classrooms; 437 centers) for parent-reported BITSEA Competence, and 1,704 children (751 classrooms; 437 centers) for parent-reported BITSEA Problem. The same entropy balancing weights---constructed using the 26 classroom-level covariates described in \cref{app:dose-response}---are applied to ensure methodological comparability with the primary results.

\subsubsection{Full Regression Results for Parent-Reported Outcomes}
\label{subsubsec:F2-2-results}

\Cref{tab:F1} presents the complete weighted linear and GAM estimates for parent-reported CDI IRT, BITSEA Competence, and BITSEA Problem scores across the six process-quality domains.

\begin{table}[htbp]
\centering
\caption{Weighted linear and GAM dose--response estimates for parent-reported child outcomes across QCIT and CLASS domains}
\label{tab:F1}
\small
\begin{adjustbox}{max width=\textwidth}
\begin{tabular}{@{}llrrrrc@{}}
\toprule
Response & Dose & Est. & SE & $p$-value & edf & GAM $p$ \\
\midrule
\multicolumn{7}{@{}l}{\textit{English CDI IRT Score (Parent-reported)}} \\
 & QCIT Social-Emotional       &  0.42 & 0.29 & 0.16    & 1.67 & 0.23 \\
 & QCIT Cognitive              &  1.11 & 0.39 & $<$0.01 & 1.00 & $<$0.01 \\
 & QCIT Language-Literacy      &  1.27 & 0.34 & $<$0.01 & 1.00 & $<$0.01 \\
 & CLASS-T Emotional-Behavioral&  0.59 & 0.30 & 0.05    & 1.00 & 0.05 \\
 & CLASS-T Learning            &  1.04 & 0.31 & $<$0.01 & 1.05 & $<$0.01 \\
 & CLASS-I Responsive          &  0.51 & 0.76 & 0.50    & 1.96 & 0.52 \\
\addlinespace
\multicolumn{7}{@{}l}{\textit{BITSEA Competence Score (Parent-reported)}} \\
 & QCIT Social-Emotional       &  0.03 & 0.07 & 0.68    & 1.00 & 0.68 \\
 & QCIT Cognitive              &  0.14 & 0.10 & 0.15    & 1.00 & 0.15 \\
 & QCIT Language-Literacy      &  0.14 & 0.09 & 0.12    & 1.59 & 0.20 \\
 & CLASS-T Emotional-Behavioral&  0.16 & 0.08 & 0.04    & 1.30 & 0.07 \\
 & CLASS-T Learning            &  0.14 & 0.08 & 0.11    & 1.00 & 0.11 \\
 & CLASS-I Responsive          & $-$0.15 & 0.22 & 0.49  & 1.00 & 0.49 \\
\addlinespace
\multicolumn{7}{@{}l}{\textit{BITSEA Problem Score (Parent-reported)}} \\
 & QCIT Social-Emotional       &  0.13 & 0.16 & 0.41    & 1.58 & 0.46 \\
 & QCIT Cognitive              &  0.39 & 0.21 & 0.06    & 2.30 & 0.05 \\
 & QCIT Language-Literacy      &  0.27 & 0.18 & 0.14    & 2.42 & 0.12 \\
 & CLASS-T Emotional-Behavioral& $-$0.22 & 0.18 & 0.22  & 1.00 & 0.22 \\
 & CLASS-T Learning            & $-$0.25 & 0.19 & 0.19  & 1.00 & 0.19 \\
 & CLASS-I Responsive          &  0.29 & 0.39 & 0.45    & 1.31 & 0.60 \\
\bottomrule
\end{tabular}
\end{adjustbox}
\vspace{0.5em}

\noindent\footnotesize\textit{Note.} Estimates from entropy-balanced weighted regressions adjusting for 26 classroom-level covariates (\cref{app:dose-response}). ``Dose'' = empirical Bayes estimates of latent process quality from the three-level GALAMM (\cref{app:galamm,app:eb-predictions}). Outcomes are center-mean-centered. edf = effective degrees of freedom (edf $\approx 1$ indicates linearity; edf $> 1$ indicates nonlinearity).
\end{table}

\subsubsection{Interpretation of Parent-Reported Results}
\label{subsubsec:F2-3-interpretation}

The parent-reported results largely reproduce the domain-matching patterns observed in the teacher-reported analyses, strengthening confidence that the findings reflect genuine associations rather than common-source artifacts.

For parent-reported CDI IRT scores, the key language/cognitive domains remain significant predictors: QCIT Cognitive (coef $= 1.11$, $p < 0.01$), QCIT Language-Literacy (coef $= 1.27$, $p < 0.01$), and CLASS-T Learning (coef $= 1.04$, $p < 0.01$). These effect sizes are comparable to those from teacher reports, and the GAM diagnostics indicate predominantly linear relationships (edf $\approx 1.00$). CLASS-T Emotional-Behavioral also shows a marginally significant positive association (coef $= 0.59$, $p = 0.05$), which was nonsignificant in the teacher-reported models---possibly reflecting that parents observe children in home contexts where the benefits of classroom emotional support manifest in more generalized language use.

For parent-reported BITSEA Competence, CLASS-T Emotional-Behavioral maintains its significant positive association (coef $= 0.16$, $p = 0.04$), consistent with the teacher-reported finding that emotional-behavioral support promotes social-emotional competence. The coefficient is slightly attenuated relative to the teacher-reported estimate (0.21), which is expected given that parent assessments capture child behaviors in a different context. Other domains show positive but nonsignificant associations, mirroring the teacher-reported pattern.

For parent-reported BITSEA Problem, the results remain largely nonsignificant, paralleling the teacher-reported null findings. QCIT Cognitive shows borderline evidence of nonlinearity (edf $= 2.30$, $p = 0.05$) with a positive linear slope (coef $= 0.39$, $p = 0.06$). This unexpected pattern should be interpreted cautiously given the marginal significance, multiple comparisons, and absence of a similar effect in teacher reports.

\subsubsection{Summary: Evidence Against Common-Source Bias}
\label{subsubsec:F2-4-summary}

Convergence between teacher- and parent-reported findings---two reporters who observe children in different contexts---reduces the likelihood that the primary domain-specific associations are artifacts of shared reporting methods. Parents and teachers have different relationships with children and complete assessments independently. Research consistently shows that parent--teacher concordance on child behavioral assessments is typically low to moderate, making the replication of domain-specific patterns across reporters particularly compelling.

The convergence of key findings---language/cognitive supports predicting communicative outcomes and emotional-behavioral support predicting social-emotional competence---across two independent informants strengthens the causal interpretation of the dose--response relationships. At the same time, the continued absence of significant effects for BITSEA Problem scores across both reporters reinforces the conclusion that classroom process quality, as measured by CLASS and QCIT, does not show short-term associations with reduced behavioral problems in this infant--toddler sample.


\end{document}